\documentclass[11 pt]{article}
\usepackage{amsmath}
\usepackage{amsthm}
\usepackage{amssymb}
\usepackage{amsfonts}
\usepackage{bbm}
\usepackage{epic}
\usepackage{eepic}
 \usepackage{times}
\usepackage[matrix,arrow]{xy}
\usepackage{macros}
\usepackage{epsfig}

\theoremstyle{plain}

\newtheorem{thm}{Theorem}
\newtheorem{propn}{Proposition}
\newtheorem{lem}{Lemma}

\theoremstyle{remark}

\newcommand{\sst}{\scriptscriptstyle}

\renewcommand{\1}{\one}
\renewcommand{\2}{\two}

\newcommand{\beq}{\begin{equation}}
\newcommand{\eeq}{\end{equation}}

\newcommand{\pa}{\partial}
\newcommand{\ot}{\otimes}

\newcommand{\SRN}{{\rm N}}

\newcommand{\la}{\lambda}

\newcommand{\bz}{\bar{z}}

\newcommand{\CA}{{\mathcal A}}
\newcommand{\CB}{{\mathcal B}}

\newcommand{\CC}{{\mathcal C}}
\newcommand{\CD}{{\mathcal D}}

\newcommand{\CF}{{\mathcal F}}

\newcommand{\CL}{{\mathcal L}}
\newcommand{\CM}{{\mathcal M}}

\newcommand{\CO}{{\mathcal O}}

\newcommand{\CR}{{\mathcal R}}

\newcommand{\CW}{{\mathcal W}}

\newcommand{\SA}{{\mathsf A}}
\newcommand{\SB}{{\mathsf B}}
\newcommand{\SC}{{\mathsf C}}
\newcommand{\SD}{{\mathsf D}}

\newcommand{\SL}{{\mathsf L}}

\newcommand{\SM}{{\mathsf M}}

\newcommand{\SO}{{\mathsf O}}

\newcommand{\SQ}{{\mathsf Q}}

\newcommand{\ST}{{\mathsf T}}

\newcommand{\SY}{{\mathsf Y}}
\newcommand{\SZ}{{\mathsf Z}}

\newcommand{\su}{{\mathsf u}}
\newcommand{\sv}{{\mathsf v}}

\newcommand{\one}{{\mathfrak 1}}
\newcommand{\two}{{\mathfrak 2}}

\newcommand{\BC}{{\mathbb C}}

\newcommand{\BS}{{\mathbb S}}

\newcommand{\BZ}{{\mathbb Z}}

\newcommand{\srn}{{\sst\rm N}}
\newcommand{\SRM}{{\rm M}}
\newcommand{\srm}{{\sst\rm M}}

\newcommand{\rf}[1]{(\ref{#1})}
\renewcommand{\bz}{{\mathbf z}}

\newcommand{\aufz}
{\begin{list}{$\bullet$}{\topsep0cm \itemsep0cm \parsep0cm}}
\newcommand{\eaufz}{\end{list}}

\begin{document}

\begin{flushright}
LPENSL-TH-12-5\\
YITP-SB-12-17
\end{flushright}

\vspace{2cm}

\title{}

\vspace{1cm}

{\centerline{\LARGE{\bf{The $\tau_2$-model and the chiral Potts model revisited:}}}}

\vspace{0.5cm}

{\centerline{\LARGE{\bf{completeness of Bethe equations from Sklyanin's SOV method}}}}

\vspace{50pt}
\begin{large}
\begin{center}
{\bf N.~Grosjean}\footnote[1]{Laboratoire de Physique, UMR 5672
du CNRS, ENS Lyon,  France, nicolas.grosjean@ens-lyon.fr}~~{\bf and}~~
{\bf G.~Niccoli}\footnote[2]{ YITP, Stony Brook University, New York 11794-3840, USA, niccoli@max2.physics.sunysb.edu},~~
\end{center}
\end{large}
\vspace{4cm}

{\bf Abstract}\ \ {\small The most general cyclic representations of the quantum integrable $\tau_2$-model are analyzed. The complete characterization of the $\tau_2$-spectrum (eigenvalues and eigenstates) is achieved in the framework of Sklyanin's Separation of Variables (SOV) method by extending and adapting the ideas first introduced in \cite{FrbtGN10,FrbtGN11}: i) The determination of the $\tau_2$-spectrum is reduced to the classification of the solutions of a given functional equation in a class of polynomials. ii) The determination of the $\tau_2$-eigenstates is reduced to the classification of the solutions of an associated Baxter equation. These last solutions are proven to be polynomials for a quite general class of $\tau_2$-self-adjoint representations and the completeness of the associated Bethe ansatz type equations is derived. Finally, the following results are derived for the inhomogeneous chiral Potts model: i) Simplicity of the spectrum, for general representations. ii) Complete characterization of the chiral Potts spectrum (eigenvalues and eigenstates) and completeness of Bethe ansatz type equations, for the self-adjoint representations of $\tau_2$-model on the chiral Potts algebraic curves.}
\begin{quotation} 
\end{quotation}

\newpage
\tableofcontents
\newpage
\section{Introduction}
The quantum integrable $\tau_2$-model has been first introduced in \cite{FrbtBS}
where its Lax operator was constructed as a general solution to the Yang-Baxter equations
w.r.t. the 6-vertex R-matrix. There, for the subvariety of the cyclic
representations parameterized by points on the algebraic curves associated to the chiral Potts (chP) model\footnote{These curves are parametrized by the equations required for the periodicity of dilogarithm functions, i.e. (3.5) of \cite{FrbtBS}.}, the analysis of the spectrum (eigenvalues) was made by the standard construction of the Baxter $\SQ$-operator. This last
operator was shown to coincide with the transfer matrix of the integrable chiral Potts  (chP) model
\cite{FrbtAMcP0}-\cite{FrbtBa89-1}; in this way a first remarkable
connection between these two apparently very different models\footnote{The differences between them are
simply expressed in the 2-dimensional statistical mechanics formulation. Both
are characterized by Boltzmann weights which satisfy the star-triangle (Yang-Baxter)
equations but while those of the $\tau_2$-model satisfy the ordinary difference property in the rapidities those of the chP-model do not. In fact, one of the interesting properties of the chP-model is that its weights are uniformized by curves of genus $g>1$, hence they cannot satisfy the difference property. Let us recall that
the first solutions of the star-triangle equations with this non-difference property were obtained
in \cite{FrbtauYMcPTY,FrbtMcPTS,FrbtauYMcPT}, while in \cite{FrbtBaPauY} the general solutions for the chP-model were derived.} was established. The
eigenvalues analysis was further developed in \cite{FrbtBBP90,FrbtBa89} where additional
functional equations for the transfer matrices of these two models, like those corresponding to the fusion hierarchy\footnote{The approach of fusion hierarchy of
commuting transfer matrices was first introduced in \cite{FrbtKRS,FrbtKR}.}
of commuting transfer matrices\footnote{The $\tau_2$-transfer matrix is the second element in this hierarchy,
this explains the name given to the model.}, were shown. In the special
case of the superintegrable chP-model, the role of Bethe
ansatz type equations for the spectral analysis was first
pointed out in \cite{FrbtAMcP0,FrbtAMcP1,FrbtAMcP2}. In
\cite{FrbtTarasovSChP} the algebraic Bethe ansatz (ABA) description of the superintegrable chP-spectrum was rigorously introduced on the
basis of the connection between the $\tau_2$-model and the chP-model. In particular, the Bethe
ansatz construction was applied to the $\tau_2$-transfer
matrix, which led to the reproduction of the Baxter results \cite{FrbtBa89-1} on the subset of the translation-invariant eigenvectors of the superintegrable chP-model.
More recently \cite{FrbtBa04}, Baxter has extended the eigenvalues analysis of the $\tau_2$-model to completely general cyclic representations. The main tool used there was the construction
of a generalized $\SQ$-operator which satisfies the Baxter equation with the $\tau_2$-transfer
matrix\footnote{There, it was considered as a kind of “generalized chiral Potts” model transfer
matrix, however, the fact that even the commutativity of the elements of such operator
family is not given leaves the definition of generalized chP-model only formal.} and the
consequent extension to these representations of the functional relations for the fused transfer
matrices.
\subsection{Motivations}
Let us comment that, in the above literature, the spectral analysis suffers in general from at least one of the following problems: i) Analysis reduced to the eigenvalues only: no eigenstates construction for the functional methods based on the Baxter $\SQ$-operator and the fusion matrices. ii) Reduced applicability: the ABA applies only to very special representations of the $\tau_2$-model and of the chP-model.  iii) Lack of completeness proof, i.e. the completeness of the spectrum description is not given by the method. In particular, this last problem affects almost all the Bethe ansatz methods\footnote{As the coordinate Bethe ansatz \cite{FrbtBe31,FrbtBax82,FrbtABBQ87}, the algebraic Bethe ansatz \cite{FrbtSF78}-\cite{FrbtF95}, the analytic Bethe ansatz  \cite{FrbtR83-I,FrbtR83-II}.}, which leads to the characterization of the spectrum by solutions of an associated
system of Bethe ansatz type equations.

A rigorous proof of the completeness of these spectrum characterizations is then a
fundamental goal which was achieved in the literature only for a few examples of integrable
quantum models, including for example the XXX Heisenberg model \cite{FrbtMTV} and the cyclic representations of the lattice sine-Gordon model \cite{FrbtNT}. It is
worth remarking that in the case of the superintegrable chP-model and the $\tau_2$-representations to which the ABA applies such a proof is missing. Instead, for the reduction to
the trivial algebraic curve case, i.e. the Fateev-Zamolodchikov model \cite{FrbtFZ}, the study of the completeness was
addressed by a numerical method in  \cite{FrbtADMcCoy92} (see also \cite{FrbtFMcCoy01} and \cite{FrbtNR03} for further applications of this method). This stresses the importance of the proof of the completeness of the system of Bethe ansatz
type equations for the $\tau_2$-model and the chP-model that is given in this article.

Note that all our results are derived in the framework of Sklyanin's SOV method \cite{FrbtSk1}-\cite{FrbtSk3}, which for the $\tau_2$-model was first developed \footnote{See
also the series of works \cite{FrbtGIPST07,FrbtGIPST08} where a first result
concerning the computation of the form factors of local operators by SOV was achieved for the special
case of the generalized Ising model.} in \cite{FrbtGIPS06}. However, we cannot get the spectrum characterization by functional equations directly from the SOV method. Indeed, in the case of cyclic representations \cite{FrbtTa} of integrable quantum models, this characterization is related to solutions of an associated finite system of Baxter-like equations. If this system is compatible with the eigenvalue equation associated to some Baxter $\SQ$-operator, the construction of this operator is one standard procedure to achieve a reformulation in terms of functional equations in the SOV characterization. For example, this was the strategy followed in \cite{FrbtNT} for the special cyclic
representations of the sine-Gordon model. In \cite{FrbtGIPS06}, the functional
relations of fusion and truncation of transfer matrices (generalized by Baxter in
\cite{FrbtBa04}) were proven to be compatible with the SOV characterization, leading to the derivation of a functional equation reformulation for the $\tau_2$-model. 

Note that the standard Baxter construction of a $\SQ$-operator by gauge
transformations and the subsequent derivation of the fused transfer
matrices can only be applied when the existence of some model dependent quantum
dilogarithm functions \cite{FrbtFK2}-\cite{FrbtBT06} is proven, which can represent a
concrete technical problem. Then, it is relevant to ponder whether it is possible to bypass this kind of constructions by providing
different derivations of the functional equation reformulation of the SOV characterization of
the spectrum. We show that this is indeed possible for completely general cyclic representations of the $\tau_2$-model; we achieve such a result following the approach first introduced in \cite{FrbtGN10}. We prove that the $\tau_2$-transfer matrix (completed with the $\Theta$-charge for some subvariety of
representations) defines a complete set of commuting charges, i.e. it has simple
spectrum\footnote{It is worth remarking that the simplicity implies that the fused transfer matrices do not really add further information useful for the spectrum characterization.}. Moreover, we provide the complete characterization of all its eigenvalues and
eigenstates as solutions of given functional equations in well
defined classes of functions. The same statements are proven
for the inhomogeneous chP-transfer matrix. 
In the $\tau_2$-self-adjoint representations, this
characterization leads to the reconstruction of the $\tau_2$-eigenbasis which, for the reduction to representations on the
chP algebraic curves, are proven to be a
simultaneous eigenbasis of the inhomogeneous chP-transfer matrix. Furthermore, we
show that there exists a quite general subvariety of $\tau_2$-self-adjoint representations, for
which we can reconstruct the eigenvalues and the full basis of eigenstates in terms of the
solutions of an associated system of Bethe ansatz type equations. Finally, we prove the same
statement of completeness for the chP-spectrum, which corresponds to the special $\tau_2$-self-adjoint representations on the chP algebraic curves.

In the paper \cite{FrbtGMN12-SG} an approach has been developed in the framework
of the quantum inverse scattering method (QISM) \cite{FrbtSF78}-\cite{FrbtF95}, \cite{FrbtTh81}-\cite{FrbtJ90} to
achieve the complete solution of lattice integrable quantum models by the
exact characterization of their spectrum and the computation of the matrix
elements of local operators. In particular, in \cite{FrbtGMN12-SG} the approach has been
developed for the relevant case of the lattice quantum sine-Gordon model \cite{FrbtFST,FrbtIK82}
associated by QISM to some special cyclic representations \cite{FrbtTa} of the 6-vertex
Yang-Baxter algebra, and it can be considered as the generalization to this SOV framework of the Lyon group method\footnote{Let us remark that this method has been introduced in \cite{FrbtKMT99} in the algebraic Bethe
ansatz (ABA) framework \cite{FrbtSF78}-\cite{FrbtFST} for the spin-1/2 XXZ quantum
chain \cite{FrbtBe31}, \cite{FrbtH28}-\cite{FrbtLM66} with periodic boundary conditions and then further developed in \cite{FrbtMT00}-\cite{FrbtKKMST07}. See also \cite{FrbtK01, FrbtCM07} and \cite{FrbtKKMNST07}-\cite{FrbtKKMNST08}, for the extension of this method in the ABA framework, to the higher spin XXX quantum chain and to the open spin-1/2
XXZ quantum chain \cite{FrbtSkly88}-\cite{FrbtGZ94} with diagonal boundary
conditions, respectively.}. The possibility to adapt this approach to the $\tau_2$-model using the SOV framework is another motivation to this paper.

\subsection{Organization of the paper}
The paper is organized as follows. In section 2 we introduce the $\tau_2$-model for general cyclic representations, pointing out its main properties and characterizing the representations for which its transfer matrix is self-adjoint. In section 3 we characterize the spectrum (eigenvalues/eigenstates) for the most general representations by using the quantum separation of variables only. In particular, we prove the simplicity of the $\tau_2$-spectrum and we show how to reformulate the SOV characterization of the transfer matrix spectrum in terms of solutions of a functional equation. In section 4 we consider the restriction of the $\tau_2$-model to self-adjoint representations for which we show that a proper\footnote{Here we are pointing out that this $\SQ$-operator define a one-parameter family of commuting operators which commute with the $\tau_2$-transfer matrix.} Baxter $\SQ$-operator is naturally induced from our SOV-characterization of the $\tau_2$-spectrum. Moreover, we characterize the most general $\tau_2$-representations for which such a Baxter $\SQ$-operator is polynomial in the spectral variable, leading to the proof of the completeness of the associate Bethe ansatz equations. In section 5 we introduce an inhomogeneous version of the chiral Potts model and we characterize all its eigenstates by those constructed for the $\tau_2$-transfer matrix. Moreover, we prove the normality of the chiral Potts transfer matrix for the self-adjoint $\tau_2$-representations on the chP algebraic curves, which completely characterizes the chiral Potts spectrum (eigenvalues/eigenstates) in terms of the one of the $\tau_2$-model. Finally, we prove that there exists a nontrivial subvariety of these representations for which the simultaneous spectrum of chiral Potts model and $\tau_2$-model is completely characterized by polynomial solutions to the associated Baxter equation also proving a completeness statements for the solutions to Bethe ansatz equations. In appendix A, we give, in some details, the construction of the SOV representation for the $\tau_2$-model. In appendix B, we reproduce, adapting to our notation, the Baxter construction of the generalized Baxter $\SQ$-operator for the $\tau_2$-model, to point out the connections with our SOV construction of the $\tau_2$-spectrum. Finally, appendix C describes in the most general framework of cyclic representation some useful technical results.

{\par {\small
\section{The $\tau_2$ model}\label{FrbtSOV}
\subsection{Definitions and first properties}
The Lax operator which describes the $\tau _{2}$-model can be parametrized
as follows\footnote{Up to different notations, this Lax operator coincides with that introduced in \cite{FrbtBS}.}:%
\begin{equation}
\SL_{n}(\lambda )\equiv \left( 
\begin{array}{cc}
\lambda \alpha _{n}\sv_{n}-\beta _{n}\lambda ^{-1}\sv_{n}^{-1} & \su_{n}\left(
q^{-1/2}\mathbbm{a}_{n}\sv_{n}+q^{1/2}\mathbbm{b}_{n}\sv_{n}^{-1}\right) \\ 
\su_{n}^{-1}\left( q^{1/2}\mathbbm{c}_{n}\sv_{n}+q^{-1/2}\mathbbm{d}_{n}\sv_{n}^{-1}\right) & \gamma
_{n}\sv_{n}/\lambda -\delta _{n}\lambda /\sv_{n}%
\end{array}%
\right) ,
\end{equation}
where the $\mathbbm{a}_{n}$, $\mathbbm{b}_{n}$, $\mathbbm{c}_{n}$, $\mathbbm{d}_{n}$, $\alpha_n$, $\beta_n$, $\gamma_n$ and $\delta_n$ are constants associated to the site $n$ that satisfy the relations:
\begin{equation}
\alpha_n \gamma _{n}=\mathbbm{a}_{n}\mathbbm{c}_{n},\text{ \ \ \ \ }\beta_n \delta
_{n}=\mathbbm{b}_{n}\mathbbm{d}_{n},
\end{equation}%
and we have denoted with $\su_{n}$ and $\sv_{n}$ the generators of the local
Weyl algebras:%
\begin{equation}
\su_{n}\sv_{m}=q^{\delta _{n,m}}\sv_{m}\su_{n}\text{ \ \ }\forall n,m\in \{1,...,\mathsf{N}\}.
\end{equation}%

We will study the case where $q\equiv e^{-i\pi\beta^2}$ is a $p$-root of unity:
\begin{equation}\label{Frbtbeta}
\beta^2\,=\,\frac{p'}{p}\,,\,\,\,\,\,\, p \equiv 2l+1, p' \equiv 2l'\,\, \text{ and  }\,\, l,l'\in\BZ^{>0}\,\,\, \rightarrow \,\,\,\,q^{p}=1. 
\end{equation}
In this case we can define a finite-dimensional representation of dimension $p$ for each Weyl algebra ${\cal W}_{n}$ by introducing the states:
\begin{equation}\label{Frbtu-basis}
|\, \bz\,\rangle\equiv|\, z_1,\dots,z_\mathsf{N}\,\rangle\,\,\, \text{with}\,\,\, z_i\in\BS_p\equiv\{q^{2n};n=0,\dots,2l\}\,\,\, \text{and}\,\, i\in\{1,...,\mathsf{N}\},
\end{equation}
for which we have:
\begin{equation}\label{Frbtreprdef}
\begin{aligned}
&\sv_n\,|\, z_1,\dots,z_\mathsf{N}\rangle=\, |\, z_1,\dots,q z_n,\dots,z_\mathsf{N}\rangle\,,\\
&\su_n\,|\, z_1,\dots,z_\mathsf{N}\rangle=\,z_n |\, z_1,\dots,z_\mathsf{N}\rangle\,.
\end{aligned}
\end{equation}
From the Lax operators, we can define the monodromy matrix by:
\begin{equation}\label{FrbtMdef}
\SM(\lambda )=\left( 
\begin{array}{cc}
\SA(\lambda ) & \SB(\lambda ) \\ 
\SC(\lambda ) & \SD(\lambda )%
\end{array}%
\right) \equiv \SL_{\mathsf{N}}(\lambda )\cdots \SL_{1}(\lambda ),
\end{equation}%
which satisfies the quadratic relation known as the Yang-Baxter relation:
\begin{equation}\label{FrbtYBA}
R(\la/\mu)\,(\SM(\la)\ot 1)\,(1\ot\SM(\mu))\,=\,(1\ot\SM(\mu))\,(\SM(\la)\ot 1)R(\la/\mu)\,,
\end{equation}
w.r.t. the six-vertex (standard) $R$-matrix:
\begin{equation}\label{FrbtRlsg}
 R(\la) =
 \left( \begin{array}{cccc}
 q^{}\la-q^{-1}\la^{-1} & & & \\ [-1mm]
 & \la-\la^{-1} & q-q^{-1} & \\ [-1mm]
 & q-q^{-1} & \la-\la^{-1} & \\ [-1mm]
 & & &  q\la-q^{-1}\la^{-1}
 \end{array} \right) \,.
\end{equation}
The elements of $\SM(\la)$ then generate a representation $\CR_\mathsf{N}$ of dimension $p^{\mathsf{N}}$ of the so-called Yang-Baxter algebra, this representation is characterized by
$6\mathsf{N}$ parameters.
In particular, the commutation relations \rf{FrbtYBA} lead to the relation $\left[ \SB(\lambda),\SB(\mu) \right]=0$ for all $\lambda$ and $\mu$, and also to the mutual commutativity of the elements of the one-parameter family of operators:
\begin{equation}\label{FrbtTdef}
\tau _{2}(\lambda )\equiv\,{\rm tr}_{\BC^2}^{}\SM(\la)\,=\SA(\lambda
)+\SD(\lambda ),
\end{equation}
known as  transfer matrix. Let us introduce the operator:
\begin{equation}\label{Frbttopological-charge}
\Theta \equiv \prod_{n=1}^{\mathsf{N}}\sv_{n},
\end{equation}
which plays the role of a {\it grading operator} in the Yang-Baxter algebra\footnote{The proof of the lemma is given following the same steps of that of Proposition 6 of \cite{FrbtNT}.}:

\begin{lem}\label{FrbtYB-G-T} $\Theta $ commutes with the transfer matrix. Moreover, it satisfies
the following commutation relations with the elements of the monodromy matrix:
\begin{eqnarray}
\Theta \SC(\lambda ) &=&q\SC(\lambda )\Theta \text{, \ \ \ }[\SA(\lambda ),\Theta
]=0, \\
\SB(\lambda )\Theta &=&q\Theta \SB(\lambda ),\text{ \ \ }[\SD(\lambda ),\Theta ]=0.
\end{eqnarray}
\end{lem}
Besides, the $\Theta $-charge allows to express the asymptotics of the transfer matrix in $\lambda \to 0$ and in $\lambda \to \infty$; in particular, from the known form of the Lax operator, we derive the following expansions:
\begin{align}
\SA(\lambda )& =\left( \lambda ^{\SRN}\Theta \prod_{a=1}^{\SRN}\alpha
_{n}+(-1)^{\SRN}\lambda ^{-\SRN}\Theta ^{-1}\prod_{a=1}^{\SRN}\beta
_{a}\right) +\sum_{i=1}^{\SRN-1} \SA_i \lambda^{\SRN-2i},  \label{Frbtasymp-A} \\
\SD(\lambda )& =\left( \lambda ^{-\SRN}\Theta \prod_{a=1}^{\SRN%
}\gamma _{a}+(-1)^{\SRN}\lambda ^{\SRN}\Theta ^{-1}\prod_{a=1}^{\SRN}\delta
_{a}\right) +\sum_{i=1}^{\SRN-1} \SD_i \lambda^{\SRN-2i},\label{Frbtasymp-D}
\end{align}%
with $\SA_i$ and $\SD_i$ being operators, 
and so\footnote{Here, we have used the short notation $\lim_{\log \lambda \to \pm \infty}$, where $-$ stands for the $\lim_{\lambda \to 0}$ and $+$ stands for the $\lim_{\lambda \to + \infty}$.}
\begin{equation}\label{Frbtasymptotics-t}
\lim_{\log \lambda \rightarrow \mp \infty }\lambda ^{\pm \SRN}\tau
_{2}(\lambda )=\left( \Theta ^{\mp 1}a_{\mp }+\Theta ^{\pm 1}d_{\mp }\right),
\end{equation}
where:%
\begin{equation}\label{FrbtAsymptotic-A-D}
a_{+}\equiv \prod_{a=1}^{\mathsf{N}}\alpha _{a},\text{ \ \ }a_{-}\equiv (-1)^{\mathsf{N}%
}\prod_{a=1}^{\mathsf{N}}\beta _{a},\text{\ \ }d_{+}\equiv (-1)^{\mathsf{N}}\prod_{a=1}^{%
\mathsf{N}}\delta _{a},\text{ \ }d_{-}\equiv \prod_{a=1}^{\mathsf{N}}\gamma _{a}.
\end{equation}
Let us denote by $\Sigma _{\tau_2}$ the set of the eigenvalues $t(\lambda )$ of the transfer matrix $\tau_2(\lambda )$, then:
\begin{equation}\label{Frbtset-t}
\Sigma _{\tau_2}\subset \mathbb{C}_{even}[\lambda ,\lambda ^{-1}]_{\mathsf{N}}\text{
for }\mathsf{N}\text{ even, \ \ \ }\Sigma _{\tau_2}\subset \mathbb{C}_{odd}[\lambda
,\lambda ^{-1}]_{\mathsf{N}}\text{ for }\mathsf{N}\text{ odd},
\end{equation}
where $\mathbb{C}_\epsilon[x,x^{-1}]_{\SRM}$ is the linear space of the Laurent polynomials over the field $\mathbb{C}$ of degree $\SRM$ in the variable $x$ that are even or odd, as stated in the index $\epsilon$. The $\Theta $-charge allows to introduce the grading $\Sigma_{\tau_2}=\bigcup_{k=0}^{2l}\Sigma _{\tau_2}^{k}$, where:
\begin{equation}
\Sigma _{\tau _{2}}^{k}\equiv \left\{ t(\lambda )\in \Sigma _{\tau
_{2}}:\lim_{\log \lambda \rightarrow \mp \infty }\lambda ^{-\mathsf{N}%
}t(\lambda )=\left( q^{\mp k}a_{\mp }+q^{\pm k}d_{\mp }\right) \right\},
\end{equation}
where to any $t(\lambda )\in \Sigma_{\tau_2}^{k}$ corresponds simultaneous eigenstates of
${\tau_2}(\lambda )$ and $\Theta $\ with $\Theta $-eigenvalue $q^{k}$.
\subsubsection{Quantum determinant}
The quantum determinant:
\begin{equation}
{\rm det_q}\SM(\la)\,\equiv\,
\SA(\la)\SD(q^{-1}\la)-\SB(\la)\SC(q^{-1}\la),
\end{equation}
is a central element\footnote{See \cite{FrbtIK81} and also \cite{FrbtIK09} for an historical note on the centrality of the quantum determinant in the Yang-Baxter algebra.} of the Yang-Baxter algebra \rf{FrbtYBA} which has the factorized form:
\begin{equation}
\text{det}_{\text{q}}\SM(\lambda )=\prod_{n=1}^{\mathsf{N}}\text{det}_{\text{q}%
}\SL_{n}(\lambda),
\end{equation}
where:
\begin{equation}
{\rm det_q}\SL_{n}(\la)\equiv \left( \SL_{n}(\lambda )\right)
_{11}\left( \SL_{n}(\lambda /q)\right) _{22}-\left( \SL_{n}\right) _{12}\left( \SL_{n}\right) _{21},
\end{equation}
are the local quantum determinants, which explicitly read: 
\begin{eqnarray}\label{Frbtexplicit-q-det}
{\rm det_q}\SM(\la)&=&\prod_{n=1}^{\mathsf{N}}k_{n}(\frac{\lambda }{\mu _{n,+}}-%
\frac{\mu _{n,+}}{\lambda })(\frac{\lambda }{\mu _{n,-}}-\frac{\mu _{n,-}}{%
\lambda })   \notag \\
&=&(-q)^{\mathsf{N}}\prod_{n=1}^{\mathsf{N}}\frac{\beta _{n}\mathbbm{a}_{n}\mathbbm{c}_{n}}{\alpha _{n}}(\frac{1}{%
\lambda }+q^{-1}\frac{\mathbbm{b}_{n}\alpha _{n}}{\mathbbm{a}_{n}\beta _{n}}\lambda )(\frac{1}{%
\lambda }+q^{-1}\frac{\mathbbm{d}_{n}\alpha _{n}}{\mathbbm{c}_{n}\beta _{n}}\lambda ),
\end{eqnarray}%
where: 
\begin{equation}
k_{n}\equiv \left( \mathbbm{a}_{n}\mathbbm{b}_{n}\mathbbm{c}_{n}\mathbbm{d}_{n}\right) ^{1/2},\text{ \ }\mu
_{n,h}\equiv \left\{ 
\begin{array}{c}
iq^{1/2}\left( \mathbbm{a}_{n}\beta _{n}/\alpha _{n}\mathbbm{b}_{n}\right) ^{1/2}\text{ \ \ }h=+,
\\ 
iq^{1/2}\left( \mathbbm{c}_{n}\beta _{n}/\alpha _{n}\mathbbm{d}_{n}\right) ^{1/2}\text{ \ \ }h=-.%
\end{array}%
\right.
\end{equation}
\subsection{Self-adjoint representations}

\begin{lem}
Let $\epsilon \in \{-1,+1\}$, if the parameters of the representation satisfy the constrains:
\begin{equation}
\mathbbm{c}_{n}=-\epsilon \mathbbm{b}_{n}^{\ast },\text{ \ }\mathbbm{d}_{n}=-\epsilon \mathbbm{a}_{n}^{\ast },\text{ \ }\beta
_{n}=\epsilon \left( \mathbbm{a}_{n}^{\ast }\mathbbm{b}_{n}\right) /\alpha _{n}^{\ast },
\label{FrbtSelf-adjointness Condition}
\end{equation}
then the generators of the Yang-Baxter algebra obey the following Hermitian conjugation relations: 
\begin{equation}\label{FrbtHermit-Monodromy}
\SM(\la)^\dagger\equiv\left( 
\begin{array}{cc}
\SA^{\dagger }(\lambda ) & \SB^{\dagger }(\lambda ) \\ 
\SC^{\dagger }(\lambda ) & \SD^{\dagger }(\lambda )%
\end{array}%
\right) =\left( 
\begin{array}{cc}
\SD(\lambda ^{\ast }) & -\epsilon\SC(\lambda ^{\ast }) \\ 
-\epsilon\SB(\lambda ^{\ast }) & \SA(\lambda ^{\ast })%
\end{array}
\right) ,
\end{equation}
which, in particular, imply that the transfer matrix ${\tau_2}(\la)$ is self-adjoint for real $\lambda $. Besides, the quantum determinant has the expression\footnote{Remark that it only depends on the modulus of the parameters of Lax operators.}: 
\begin{equation}
{\rm det_q}\SM(\la)=q^{\mathsf{N}}\prod_{n=1}^{\mathsf{N}}\frac{|\mathbbm{a}_{n}|^{2}|\mathbbm{b}_{n}|^{2}}{%
|\alpha _{n}|^{2}}(\frac{1}{\lambda }+\epsilon q^{-1}\frac{|\alpha _{n}|^{2}}{%
|\mathbbm{a}_{n}|^{2}}\lambda )(\frac{1}{\lambda }+\epsilon q^{-1}\frac{|\alpha _{n}|^{2}}{%
|\mathbbm{b}_{n}|^{2}}\lambda ).
\end{equation}
\end{lem}
\begin{proof}
It is simple to observe that the Lax operator of the $\tau_2$-model satisfies the equation \rf{FrbtHermit-Monodromy}, which can be also
written as:
\begin{equation}\label{FrbtHermit-L}
\left(\SL_n(\la)\right)^{\dagger}=\sigma_{1+\delta_{1,\epsilon}}\,\SL_n(\la^{*})\,\sigma_{1+\delta_{1,\epsilon}}\,\,\,\,\,\,\longrightarrow\,\,\,\,\,\,\SM(\la)^{\dagger}=\sigma_{1+\delta_{1,\epsilon}}\,\SM(\la^{*})\,\sigma_{1+\delta_{1,\epsilon}},
\end{equation}
that is \rf{FrbtHermit-Monodromy} holds by definition of $\SM(\la)$.
\end{proof}

\section{Characterization of $\tau_2$-spectrum: general representations}\label{FrbtCompatib}

\setcounter{equation}{0}

\subsection{SOV representations}

According to Sklyanin's method \cite{FrbtSk1,FrbtSk2,FrbtSk3}, a
separation of variables (SOV) representation for the spectral problem of $\tau _{2}(\lambda )$ is given by a representation where the commutative family of operators $\mathsf{B}(\lambda )$ is diagonal.

\begin{thm}\label{FrbtSOVthm}
For almost all the values of the $6N$ parameters of the representation,
there exists a SOV representation for the $\tau _{2}$-model, i.e. $\mathsf{B}%
(\lambda )$\ is diagonalizable and has simple spectrum.
\end{thm}
\begin{proof}
See appendix \ref{FrbtSOV construction-0} for a constructive proof of this statement.
\end{proof}
Let $\langle\,\eta_{\mathbf{k}}\,|$ be the generic element of a basis of eigenvectors of $\SB(\la)$:
\begin{equation}\label{FrbtBdef}
\langle\,{\eta_{\mathbf{k}}}\,|\SB(\la)\,=\,\eta
_{_{\mathsf{N}}}^{(k_{\mathsf{N}})}\,b_{\eta_{\mathbf{k}}}(\la)\,\langle\,{\eta_{\mathbf{k}}}\,|\,,\qquad b_{\eta_{\mathbf{k}}}(\la)\,\equiv\,\prod_{a=1}^{\mathsf{N}-1}\left( \la/\eta
_{_{a}}^{(k_{a})}-\eta
_{_{a}}^{(k_{a})}/\la\right)\,,
\end{equation}
and
\begin{equation}\label{FrbtZ_B}
{\eta_{\mathbf{k}}} \in{\SZ_\SB}\,\equiv\,\big\{\,(\eta
_{_{1}}^{(k_{1})}\equiv q^{k_1}\eta^{(0)}_1,\dots,\eta
_{_{\mathsf{N}}}^{(k_{\mathsf{N}})}\equiv q^{k_\mathsf{N}}\eta^{(0)}_\mathsf{N})\,;\,{\mathbf{k}} \equiv (k_1,\dots,k_\mathsf{N})\in\BZ_p^\mathsf{N}\,\big\}\,,
\end{equation}
where $\eta^{(0)}_a$ are constants that are defined in appendix \ref{FrbtSOV construction-0}. Here, the simplicity of the spectrum of $\SB(\la)$ is equivalent to the requirement $\left( \eta^{(0)}_a\right) ^p\neq \left( \eta^{(0)}_b\right) ^p$ for any $a\neq b \in \{1,\dots,\mathsf{N}-1\}$. The action of the remaining generators of the Yang-Baxter algebra on arbitrary states\footnote{From here on to simplify the notation, we will omit the subscript ${\mathbf{k}}$ in ${\eta_{\mathbf{k}}}$ as well as the superscript $(k_{a})$ in the ${\eta_{a}^{(k_{a})}}$ and we will reintroduce them only when it will be strictly required.} $\langle \, \eta |$ then reads:
\begin{align}\label{FrbtSAdef}
\langle\,\eta\,|\SA(\la)\,=\,&\,b_\eta(\la)\left[\la\eta_\SA^{(+)} \langle\,q^{-\delta_{\mathsf{N}}}\eta\,|+\la^{-1} \eta_\SA^{(-)} \langle\,q^{\delta_{\mathsf{N}}}\eta\,|\right]
+\sum_{a=1}^{\mathsf{N}-1}\prod_{b\neq a}\frac{\la/\eta_b-\eta_b/\la}{\eta_a/\eta_b-\eta_b/\eta_a} \,{\tt a}^{(SOV)}(\eta_a)\,\langle\,q^{-\delta_{a}}\eta\,|\,,\\
\langle\,\eta\,|\SD(\la)\,=\,&\,b_\eta(\la)\left[\la\eta_\SD^{(+)}\langle\,q^{\delta_{\mathsf{N}}}\eta\,|+\la^{-1} \eta_\SD^{(-)}\langle\,q^{-\delta_{\mathsf{N}}}\eta\,|\right]+\sum_{a=1}^{\mathsf{N}-1}\prod_{b\neq a}\frac{\la/\eta_b-\eta_b/\la}{\eta_a/\eta_b-\eta_b/\eta_a} \,{\tt d}^{(SOV)}(\eta_a)\,\langle\,q^{\delta_{a}}\eta\,|\, ,
\label{FrbtSDdef}\end{align}
where ${\tt a}^{(SOV)}(\eta_a)$ and ${\tt d}^{(SOV)}(\eta_a)$ are coefficients which have to satisfy the quantum determinant condition:
\begin{equation}\label{Frbtaddet}
{\rm det_q}\SM(\eta_r)\,=\,
{\tt a}^{(SOV)}(\eta_r){\tt d}^{(SOV)}(q^{-1}\eta_r)\,, \quad\forall r=1,\dots,\mathsf{N}-1\,.
\end{equation}Here, we have defined: 
\begin{equation}\label{FrbtZAD-asymp}
\eta_\SA^{(\pm)}=(\pm 1)^{\mathsf{N}-1}a_{\pm}\prod_{n=1}^{\mathsf{N}-1}\eta_n^{\pm1},\,\,\,\,\,\,\,\,\,\, \eta_\SD^{(\pm)}=(\pm 1)^{\mathsf{N}-1}d_{\pm}\prod_{n=1}^{\mathsf{N}-1}\eta_n^{\pm1},
\end{equation}
and the states $\langle\,q^{\pm \delta_{a}}\eta\,|$ are defined by:
\begin{equation}
\langle\,q^{\pm\delta_{a}}\eta\,| \equiv \langle\,\eta_1,\dots,q^{\pm 1}\eta_a,\dots,\eta_\mathsf{N}\,|\,.
\end{equation}
Finally, $\SC(\la)$ is uniquely\footnote{Note that the operator $\SB(\la)$ is invertible except for $\la$ which coincides with a zero of $\SB$, so in general $\SC(\la)$ is defined by (4.5) just inverting $\SB(\la)$. This is enough to fix in an unique way the operator $\SC$ as it is a Laurent polynomial of degree $\mathsf{N}-1$ in $\la$.} defined by the quantum determinant relation.

\subsection{Average values of Yang Baxter generators as central elements}\label{FrbtAvvalapp}
We define the average value $\CO$ of the elements of the monodromy matrix $\SM(\la)$ by:
\begin{equation}\label{Frbtavdef}
\CO(\Lambda)\, \equiv\,\prod_{k=1}^{p}\SO(q^k\la)\,,\qquad \Lambda\, \equiv \,\la^p,
\end{equation}
where $\SO$ can be
$\SA$, $\SB$,
$\SC$ or $\SD$ and the commutativity of the families $\SA(\lambda )$, $\SB(\lambda )$, $\SC(\lambda )$ and $\SD(\lambda )$ implies
that $\CA(\Lambda)$, $\CD(\Lambda)$ are Laurent polynomials of degree $\mathsf{N}$ while $\CB(\Lambda)$, $\CC(\Lambda)$ are Laurent polynomials of degree $\mathsf{N}-1$  in $\Lambda$. Then the following proposition holds:
 
\begin{propn}\label{Frbtcentral} \text{}   
\begin{itemize}
\item[a)] The average values $\CA(\Lambda)$, $\CB(\Lambda)$, $\CC(\Lambda)$, $\CD(\Lambda)$ of the monodromy matrix elements are central. Besides, they satisfy the relation:
\begin{equation}\label{FrbtH-cj-A-D}
(\CA(\Lambda))^{*}=\CD(\Lambda^*), \ \ \ \ \ (\CB(\Lambda))^*=-\epsilon\CC(\Lambda^*),
\end{equation}
in the case of self-adjoint representations.
\item[b)] Let\footnote{A similar statement was first proven in \cite{FrbtTa}.}\begin{align}
\CM^{}(\Lambda)\,\equiv\,\left( 
\begin{array}{cc}
\mathcal{A}(\Lambda ) & \mathcal{B}(\Lambda ) \\ 
\mathcal{C}(\Lambda ) & \mathcal{D}(\Lambda )%
\end{array}%
\right)
\end{align} be the 2$\times$2 matrix whose elements are the average
values  of the elements of the monodromy matrix $\SM(\la)$, it holds then:
\begin{align}\label{FrbtRRel1a}
\CM^{}(\Lambda)\,=\,
\CL_{\mathsf{N}}^{}(\Lambda)\,\CL_{\mathsf{N}-1}^{}(\Lambda)\,\dots\,\CL_1^{}(\Lambda)\,,
\end{align}
where:
\begin{equation}\label{FrbtAverage-L}
\mathcal{L}_{n}(\Lambda )\equiv\left( 
\begin{array}{cc}
\Lambda \alpha _{n}^{p}-\beta _{n}^{p}/\Lambda & q^{p/2}(\mathbbm{a}_{n}^{p}+\mathbbm{b}_{n}^{p})
\\ 
q^{p/2}(\mathbbm{c}_{n}^{p}+\mathbbm{d}_{n}^{p}) & \gamma _{n}^{p}/\Lambda -\Lambda \delta
_{n}^{p}%
\end{array}%
\right)
\end{equation}
is the 2$\times$2 matrix whose elements are the average values of the
elements of the Lax matrix $L_n(\la)$.
\end{itemize}
\end{propn}
\textit{Proof of a).} $\CB(\Lambda)$ is central as it follows by taking directly the average of \rf{FrbtBdef} in the SOV representations:
\begin{equation}\label{FrbtCB}
\CB(\Lambda)\,=\,Z_{{\mathsf{N}}}
\prod_{a=1}^{\mathsf{N}-1}(\Lambda/Z_a-Z_a/\Lambda)\,,\qquad Z_a\equiv \eta_a^p=\left( \eta_a^{(0)} \right)^p\,.
\end{equation}
From $q^p=1$, we have that $\CA^{}(Z_r)$ and $\CD^{}(Z_r)$ are central and related to the coefficients ${\tt a}^{(SOV)}(q^k\eta_r)$ and ${\tt d}^{(SOV)}(q^k\eta_r)$ by
\begin{equation}\label{FrbtADaver}
\CA^{}(Z_r)\,\equiv\,\prod_{k=1}^{p}{\tt a}^{(SOV)}(q^k\eta_r)\,,\qquad
\CD^{}(Z_r)\,\equiv\,\prod_{k=1}^{p}{\tt d}^{(SOV)}(q^k\eta_r)\,,\qquad \forall r\in\{1,\dots,\mathsf{N}-1\}.
\end{equation}
$\mathcal{A}(\Lambda )\Lambda ^{\mathsf{N}}$ and $\mathcal{D}(\Lambda)\Lambda^{\mathsf{N}}$ are polynomials in $\Lambda^{2} $ of degree $\mathsf{N}$. The relations \rf{FrbtADaver} and the simplicity of the $\SB$-spectrum give $\mathsf{N}$ points in which these polynomials are central elements, and the centrality of the asymptotics of $\mathcal{A}(\lambda )$ and $\mathcal{D}(\lambda )$, as trivially follows from \rf{Frbtasymp-A}-\rf{Frbtasymp-D}, yields the centrality of $\mathcal{A}(\Lambda )$ and $\mathcal{D}(\Lambda )$.
\hspace{7.0cm}$\square$

\textit{Proof of b).} By using that $\SB(\la)$ is diagonalizable and with simple spectrum in the entire chain as well as in each subchain, the point \textit{b)} follows inductively by the simple extension to our representations of the recursion relations on the averages of the Yang-Baxter generators of Proposition 3 of \cite{FrbtNT}. 

\hspace{16.8cm}$\square$

\subsection{SOV characterization of $\tau_2$-model}\label{FrbtSOV-Gen}
The spectral problem for ${\tau_2}(\la)$ in the SOV representations is equivalent to the discrete system of Baxter-like equations for the wave-function $\Psi_t(\eta)\equiv\langle\,\eta\,|\,t\,\rangle$ of a $\tau_2$-eigenstate $|\,t\,\rangle$:
\begin{equation}\label{FrbtSOVBax1}
t(\eta _{r})\Psi_t(\eta)\,=\,{\tt a}^{(SOV)}(\eta _{r})\Psi_t( q^{-\delta_{r}} \eta)+{\tt d}^{(SOV)}(\eta _{r})\Psi_t(q^{\delta_{r}} \eta)\, \qquad \text{ \ }\forall r\in \{1,...,\mathsf{N}-1\},
\end{equation}
plus the equation in the variable $\eta_\mathsf{N}$:
\begin{equation}\label{FrbtSOVBax2}
\Psi_{t}(q^{\delta_{\mathsf{N}}}\eta)\,=\,q^{-k}\Psi_{t}(\eta),
\end{equation}
for $t(\lambda )\in \Sigma _{{\tau_2}}^{k}\ \ $with $k\in \{0,...,2l\}$
where $(\eta _{1},...,\eta _{\mathsf{N}})\in \SZ_{\SB}$ and
\begin{equation}\label{FrbtT_r}
q^{\pm\delta_{r}} \eta \equiv(\eta_1,\dots,q^{\pm 1}\eta_r,\dots,\eta_\mathsf{N}).
\end{equation}
Let us remark that the coefficients ${\tt a}^{(SOV)}(\eta_r)$ and ${\tt d}^{(SOV)}(\eta_r)$ which satisfy \rf{Frbtaddet} and the average conditions \rf{FrbtADaver} are fixed up to the following gauge transformations:
\begin{equation}\label{Frbtgauge}
{\tt \bar a}^{(SOV)}(\eta_r)\,=\,{\tt a}^{(SOV)}(\eta_r)\frac{f(\eta_rq^{-1})}{f(\eta_r)}\,,
\qquad
{\tt \bar d}^{(SOV)}(\eta_r)\,=\,{\tt d}^{(SOV)}(\eta_r)\frac{f(\eta_rq)}{f(\eta_r)}\,,
\end{equation}
which just amounts to a renormalization in the states of the $\SB$-eigenbasis:
\begin{equation}
\langle\,\eta\,|\,\rightarrow\,\prod_{r=1}^{\mathsf{N}-1}f^{-1}(\eta_r)\langle\,\eta\,|\,.
\end{equation}
Here, we make the following choice of gauge:%
\begin{equation}
{\tt a}^{(SOV)}(\lambda )\equiv \left( \frac{\mathcal{A}(\Lambda )}{\prod_{a=1}^{p}%
\text{\textsc{a}}(\lambda q^{a})}\right) ^{1/p}\text{\textsc{a}}(\lambda ),%
\text{ \ \ \ \ \ \ \ \ \ \ \ \ }{\tt d}^{(SOV)}(\lambda )\equiv \left( \frac{\mathcal{D}(\Lambda )%
}{\prod_{a=1}^{p}\text{\textsc{d}}(\lambda q^{a})}\right) ^{1/p}\text{%
\textsc{d}}(\lambda ),
\end{equation}%
where:%
\begin{equation}\label{FrbtL-poly-a-d}
\text{\textsc{a}}(\lambda )\equiv \prod_{n=1}^{\mathsf{N}}(\beta _{n}\alpha
_{n})^{1/2}(\frac{\lambda }{\mu _{n,+}}-\frac{\mu _{n,+}}{\lambda }),\qquad\text{\textsc{d}}(\lambda )\equiv \prod_{n=1}^{\mathsf{N}}(\frac{\mathbbm{a}_{n}\mathbbm{b}_{n}\mathbbm{c}_{n}\mathbbm{d}_{n}}{%
\alpha _{n}\beta _{n}})^{1/2}(\frac{q\lambda }{\mu _{n,-}}-\frac{\mu _{n,-}}{%
q\lambda }).
\end{equation}%

{\bf Remark 1.} The average values $Z_r$, $Z^\pm_\SA$, $Z^\pm_\SD$, $\CA^{}(Z_r)$ and $\CD^{}(Z_r)$ are cleary unchanged by gauge transformations. Moreover, as central elements of the representation, they are not modified by similarity transformations, and they characterize gauge-invariant parameters of the SOV representation. Moreover,
Proposition \ref{Frbtcentral} uniquely defines the average values of the monodromy matrix elements and allows to establish that
$\CA^{}(\Lambda)$, $\CB^{}(\Lambda)$, $\CC^{}(\Lambda)$ and $\CD^{}(\Lambda)$ are polynomials of maximal degree 1 in each of the parameters $\alpha^p_n,\,\beta^p_n,\,a^p_n,\,b^p_n,\,c^p_n,\,d^p_n$ of
the $\tau_2$-representation. Note that this also implies that the gauge-invariant dates of the SOV representations (up to permutations) are uniquely defined in terms of these parameters. In particular, denoting by $\sigma_n^{{\mathsf{N}-1}}(Z)$ the elementary symmetric polynomial of degree $n$ in the variables $Z_r$, the $\sigma_n^{{\mathsf{N}-1}}(Z)/\sigma_{\mathsf{N}-1}^{\mathsf{N}-1}(Z)$ are
polynomials of degree 1 in the parameters of
the $\tau_2$-representation.

\subsection{Simplicity of $\tau_2$-spectrum}\label{FrbtSimple-T}
In this section we show that the spectrum of the transfer matrix ${\tau_2}(\lambda )$ is non-degenerate (or simple). Let us prove that:
\begin{lem}\label{FrbtAD-average-Z}
For almost all the values of the parameters $(\alpha_{n}^{p},\,\beta_{n}^{p},\,\mathbbm{a}_{n}^{p},\,\mathbbm{b}_{n}^{p},\,\mathbbm{c}_{n}^{p},\,\mathbbm{d}_{n}^{p})$ the average values of the monodromy matrix elements $\SA(\la)$ and  $\SD(\la)$ satisfy the inequalities:
\begin{equation}
\CA(Z_{a})\neq \CD(Z_{a}),\,\,\,\,\,\,\,\,\,\,\,\, \forall a\in \{1,...,\mathsf{N}-1\},
\end{equation}
where $Z_{a}$ are the zeros of the average value of $\SB(\la)$.
\end{lem}
\begin{proof}
Let us define the functions\footnote{From here, we will use the index $\mathsf{N}$ when it will be needed to point out that we are referring to the chain with $\mathsf{N}$ sites and we will omit it otherwise.}:
\begin{equation}
\CF_a(\alpha_{n}^{p},\,\beta_{n}^{p},\,\mathbbm{a}_{n}^{p},\,\mathbbm{b}_{n}^{p},\,\mathbbm{c}_{n}^{p},\,\mathbbm{d}_{n}^{p})\equiv\CA_\mathsf{N}(Z_{a})-\CD_\mathsf{N}(Z_{a}),\,\,\,\,\,\, \forall a\in \{1,...,\mathsf{N}-1\}.
\end{equation}
The Proposition \ref{Frbtcentral} and the Remark 1 about the functional dependence of $Z_1^{},\dots,Z_{\mathsf{N}-1}^{}$ w.r.t. these parameters implies
that it is sufficient to show that the functions $\CF_a$ are nonzero for some special value of the parameters in order to prove that they are nonzero for almost all the values of the parameters $\alpha_{n}^{p},\,\beta_{n}^{p},\,\mathbbm{a}_{n}^{p},\,\mathbbm{b}_{n}^{p},\,\mathbbm{c}_{n}^{p},\,\mathbbm{d}_{n}^{p}$, which will prove the lemma. Note that the following identities hold:
\begin{equation}
\CF_{a}\mid _{\left\{ \begin{tiny}
\begin{array}{l}
\alpha _{1}\beta _{1}=\mathbbm{b}_{1}\mathbbm{d}_{1}=\mathbbm{c}_{1}\mathbbm{a}_{1}\text{ \ \ \ with \ \ }%
(\mathbbm{c}_{1}^{p}+\mathbbm{d}_{1}^{p})\neq (\mathbbm{a}_{1}^{p}+\mathbbm{b}_{1}^{p}) \\ 
\\ 
\mathbbm{d}_{n}=\mathbbm{a}_{n},\mathbbm{c}_{n}=\mathbbm{b}_{n},\beta _{n}=-\mathbbm{b}_{n}\mathbbm{a}_{n}/\alpha _{n},\,\,\,n\in
\{2,...,\mathsf{N}\},%
\end{array}\end{tiny}
\right. }=q^{p/2}((\mathbbm{c}_{1}^{p}+\mathbbm{d}_{1}^{p})-(\mathbbm{a}_{1}^{p}+\mathbbm{b}_{1}^{p}))\CB_{\2,\,\mathsf{N}%
-1}(Z_{a}),  \label{Frbtspecial-choice}
\end{equation}
Here, we have used the decomposition of the chain in a first subchain $\1$, formed by the site 1, and a second subchain $\2$, formed by the remaining sites. We also have used the identities:
\begin{eqnarray}
(\CC_{\1,\,1}-\CB_{\1,\,1})&=&q^{p/2}((\mathbbm{c}_{1}^{p}+\mathbbm{d}_{1}^{p})-(\mathbbm{a}_{1}^{p}+\mathbbm{b}_{1}^{p})),\,\,\quad\quad\quad\CA_{\1,\,1}(\Lambda)=\CD_{\1,\,1}(\Lambda), \\ 
\CB_{\2,\,\mathsf{N}-1}(\Lambda)&=&\CC_{\2,\,\mathsf{N}-1}(\Lambda),\,\,\,\quad\quad\quad\quad\quad\quad\CD_{\2,\,\mathsf{N}-1}(\Lambda)=\CA_{\2,\,\mathsf{N}-1}(\Lambda),
\end{eqnarray}
which follows by using Proposition \ref{Frbtcentral} for the special choice of the parameters done in \rf{Frbtspecial-choice}. Then, we have only to show that we can always take $\CB_{\2,\,\mathsf{N}-1}(Z_{a})\neq0$ for any $a\in\{1,...,\mathsf{N}-1\}$.

Note that from the identities:
\begin{eqnarray}
{\rm det_q}\CM_{\2,\,\mathsf{N}-1}(\Lambda)&=&(\CA_{\2,\,\mathsf{N}-1}(\Lambda))^2-(\CB_{\2,\,\mathsf{N}-1}(\Lambda))^2,\,\,\,\,\\
\CB_\mathsf{N}(\Lambda)&=&\CA_{\1,\,1}(\Lambda)\CB_{\2,\,\mathsf{N}-1}(\Lambda)+\CB_{\1,\,1}\CA_{\2,\,\mathsf{N}-1}(\Lambda),
\end{eqnarray}
we have that $\CB_{\2,\,\mathsf{N}-1}(Z_{a})=0$ if and only if $Z_{a}$ is a double zero of ${\rm det_q}\CM_{\2,\,\mathsf{N}-1}
(\Lambda)$. However, this is not the case in general: by averaging the quantum determinant \rf{Frbtexplicit-q-det}, we get the formula:
\begin{equation}
{\rm det_q}\CM_{\2,\,\mathsf{N}-1}(\Lambda)\equiv\prod_{n=2}^{\mathsf{N}}\mathbbm{a}_{n}^{p}\mathbbm{b}_{n}^{p}\prod_{h=\pm 1}(\Lambda /M_{n,h}-M_{n,h}/\Lambda ), \quad M_{n,+}=\mathbbm{a}_{n}^{p}/\alpha _{n}^{p},\,\,M_{n,-}=\mathbbm{b}_{n}^{p}/\alpha _{n}^{p},
\end{equation}
so the function ${\rm det_q}\CM_{\2,\,\mathsf{N}-1}(\Lambda)$ has not double zeros for general values of the parameters $\mathbbm{a}_{n},\,\mathbbm{b}_{n},\,\alpha _{n}$.
\end{proof}
\begin{thm}
For almost all the values of the parameters $\alpha_{n}^{p},\,\beta_{n}^{p},\,\mathbbm{a}_{n}^{p},\,\mathbbm{b}_{n}^{p},\,\mathbbm{c}_{n}^{p},\,\mathbbm{d}_{n}^{p}$ of a $\tau_2$-representation, the spectrum of
$\tau_2(\lambda )$ is simple.
\end{thm}

\begin{proof}
We have to prove that, up to normalization, for any given $t(\la)\in\Sigma_{\tau_2}$ there exists only one solution to the system \rf{FrbtSOVBax1} and \rf{FrbtSOVBax2}. Let us denote by $\bar\Psi_{t}(\eta)$ with $\eta\in\SZ_\SB$ any other solution corresponding to the same  $\tau_2$-eigenvalue $t(\la)$. Then, we can define the q-Wronskian:
\begin{equation}\label{Frbtq-W}
W_{t,r}(\eta)\,=\,\Psi_{t}(\eta)\,\bar\Psi_{t}(q^{-\delta_r} \eta)-\bar\Psi_{t}(\eta)\,\Psi_{t}(q^{-\delta_r}\eta),\,\,\,\,\forall r\in\{1,...,\mathsf{N}-1\},
\end{equation}
which by the system of Baxter-like equations \rf{FrbtSOVBax1} satisfies the equations:
\begin{equation}
{\tt a}^{(SOV)}(\eta_r )\,W_{t,r}(\eta)\,=\,{\tt d}^{(SOV)}(\eta_r )\,W_{t,r}(q^{\delta_r}\eta),\,\,\,\,\forall r\in\{1,...,\mathsf{N}-1\}.
\end{equation}
Thanks to the cyclicity the averages of the above equations read:
\begin{equation}
({\CA}(Z_r)-{\CD}(Z_r))\,\CW_{t,r}(\eta)\,=0\,\,\,\,\,\text{ with }\,\,\,\,\,\CW_{t,r}(\eta)\equiv\prod_{k=0}^{2l}W_{t,r}(\eta_1,\dots,q^{k}\eta_r,\dots,\eta_\mathsf{N}),
\end{equation}
which by Lemma \ref{FrbtAD-average-Z} implies:
\begin{equation}
W_{t,r}(\eta)= 0,\,\,\,\,\,\,\,\,\forall \eta\in \SZ_\SB\,\,\,\text{and}\,\,\,\forall r\in\{1,...,\mathsf{N}-1\}.
\end{equation}
Note that this implies that the wave functions are proportional to the following factorization:
\begin{equation}\label{FrbtQeigens-even}
\Psi_t(\eta) \propto \eta _{\mathsf{N}}^{-k}\prod_{r=1}^{\mathsf{N}-1}Q_t(\eta_r)\quad\,\,\,\, \text{ for}\,\,\,t(\la)\in\Sigma_{\tau_2}^{k},
\end{equation}
where the proportionality factor could depend on fixed parameters, such as $\left( \eta_a^{(0)} \right) ^p$, and the $Q_t(\eta_r)$ is the solution of the discrete system of Baxter-like equations  \rf{FrbtSOVBax1} for the fixed $r\in\{1,...,\mathsf{N}-1\}$.
\end{proof}
\subsection{Characterization of $\tau_2$-eigenvalues as solutions of a functional equation}\label{FrbtEigenvalues-T}
Let us define the following functions:
\begin{equation}
\bar{\textsc{a}}(\lambda )\equiv \alpha (\lambda )\text{\textsc{a}%
}(\lambda ),\qquad\,\,\bar{\textsc{d}}(\lambda )\equiv \alpha^{-1} (q\lambda)%
\text{\textsc{d}}(\lambda ),
\end{equation}%
where \textsc{a}$(\lambda )$ and \textsc{d}$(\lambda )$ are the Laurent
polynomials defined in \rf{FrbtL-poly-a-d}.
Then, they always satisfy the condition:
\begin{equation}\label{Frbtrelation-q-det}
\text{det}_{\text{q}}\SM(\lambda )=\bar{\textsc{a}}(\lambda )\bar{\textsc{d}}(\lambda /q),
\end{equation}
while the function $\alpha (\lambda )$ is defined by the requirement:
\begin{equation}\label{Frbtrelation-averages}
\prod_{n=1}^{p}\bar{\textsc{a}}(\lambda q^{n})+\prod_{n=1}^{p}\bar{\textsc{d}}(\lambda q^{n})=\mathcal{A}%
(\Lambda )+\mathcal{D}(\Lambda ).
\end{equation}
Note that this last condition is a second order equation in the average $\prod_{n=1}^{p}\alpha (q^n\lambda )$ and then we have only two possible choices for the averages of the functions $\bar{\textsc{a}}(\lambda )$ and $\bar{\textsc{d}}(\lambda )$:
\begin{equation}
\prod_{n=1}^{p}\bar{\textsc{a}}(\lambda q^{n})=\Omega _{\epsilon}\left( \Lambda \right) ,\text{
\ \ }\prod_{n=1}^{p}\bar{\textsc{d}}(\lambda q^{n})=\Omega _{-\epsilon}\left( \Lambda \right) ,
\end{equation}
where $\epsilon=\mp$ and $\Omega _{\pm}$ are the two eigenvalues of the $2\times2$ matrix $\mathcal{M}(\Lambda )$ composed by the averages of the Yang-Baxter
generators.
Let us introduce the one-parameter family $D(\la)$ of $p\times p$ matrix:
\begin{equation}\label{FrbtD-matrix}
D(\la) \equiv
\begin{pmatrix}
t(\la)   &-\bar{\textsc{d}}(\la)&   0        &\cdots & 0 & -\bar{\textsc{a}}(\la)\\
-\bar{\textsc{a}}(q\la)& t(q\la)&-\bar{\textsc{d}}(q\la)& 0     &\cdots & 0 \\
      0       & {\quad} \ddots      & &     &     &         \vdots   \\
  \vdots           &     &  \cdots    &  &       &     \vdots   \\
     \vdots         &     &   & \cdots &       &   \vdots     \\
     \vdots   &            &    &  &  \ddots{\qquad}     &   0 \\
 0&\ldots&0& -\bar{\textsc{a}}(q^{2l-1}\la)& t(q^{2l-1}\la) &
-\bar{\textsc{d}}(q^{2l-1}\la)\\
-\bar{\textsc{d}}(q^{2l}\la)   & 0      &\ldots      &     0  & -\bar{\textsc{a}}(q^{2l}\la)& t(q^{2l}\la)
\end{pmatrix},
\end{equation}
where for now $t(\la )$ is only a Laurent polynomial of degree $\mathsf{N}$ in $\la$, even for $\mathsf{N}$ even and odd for $\mathsf{N}$ odd.
\begin{lem}
\label{FrbtdetD}The determinant of the matrix $D(\la)$ is a Laurent
polynomial of maximal degree $\mathsf{N}$ in $\Lambda \equiv \lambda ^{p}$, even for $\mathsf{N}$ even and odd for $\mathsf{N}$ odd.
\end{lem}
\begin{proof}
Let us start observing that $D$$(\lambda q)$ is obtained from $D(\lambda )$ by exchanging the first and $p$-th column and then the first and $p$-th row, so that 
\begin{equation}
\det_{p}\text{$D$}(\lambda q)=\det_{p}\text{$D$}(\lambda )%
\text{ \ \ }\forall \lambda \in \mathbb{C},
\end{equation}
then $\det_{p}\text{$D$}(\lambda)$ is a function of $\Lambda$. Let us observe now that $\det_{p}D(\Lambda )$ admits the following expansion:
\begin{eqnarray}\label{FrbtdetD-exp}
\det_{p}D(\Lambda )&=&-({\CA}(\Lambda)+{\CD}(\Lambda))-\text{\textsc{a}}(\lambda )\text{\textsc{d}}(\lambda/q)\det_{2l-1}D_{(1,2l+1),(1,2l+1)}(\lambda )  \notag \\
&&-\text{\textsc{a}}(\lambda q)\text{\textsc{d}}(\lambda)\det_{2l-1}D_{(1,2),(1,2)}(\lambda )+t(\lambda )\det_{2l}D_{1,1}(\lambda )\text{,}
\end{eqnarray}
where $D_{(h,k),(h,k)}(\lambda )$ denotes the $(2l-1)\times (2l-1)$
sub-matrix of $D(\lambda )$ obtained removing the rows and columns $h$ and $k$. The \textit{tridiagonality} of the matrices $D_{1,1},\,D_{(1,2),(1,2)},\,D_{(1,2l+1),(1,2l+1)}$ implies that their determinants coincide with those of the matrices obtained substituting the functions $\bar{\textsc{a}}(\lambda )$, $\bar{\textsc{d}}(\lambda )$ with the Laurent polynomials $\text{\textsc{a}}(\lambda )$, $\text{\textsc{d}}(\lambda )$.  Then, the lemma follows as all the terms in the expansions \rf{FrbtdetD-exp} are Laurent polynomials with the same properties stated in the lemma.
\end{proof}
The interest toward the function $\det_{p}D(\Lambda )$ comes from the following result: 
\begin{lem}
\label{FrbtCharact-Sigma1}Any $t(\la)\in\Sigma _{{\tau_2}}$ is a solution of the functional equation:
\begin{equation}
\det_{p}\text{$D$}(\Lambda)\equiv0.
\end{equation}
\end{lem}

\begin{proof}
First note that the determinant $\det_{p}\text{$D$}(\Lambda)$ depends from the coefficients $\bar{\textsc{a}}(\mu )$ and $\bar{\textsc{d}}(\mu/q )$ only through their products \rf{Frbtrelation-q-det} computed in $\mu\equiv q^a\lambda$ with $a=1,...,p$\  (i.e. by the quantum determinant) and the sum of their averages \rf{Frbtrelation-averages} computed in $\Lambda$ (i.e. by the sum of the averages of the operators $\SA$ and $\SD$). Indeed, this statement trivially follows from the expansion \rf{FrbtdetD-exp} and from the tridiagonality of the matrices $D_{1,1},\,D_{(1,2),(1,2)},\,D_{(1,2l+1),(1,2l+1)}$.

From the previous lemma, we only have to show that the determinant is zero in $\mathsf{N}-1$ points, and that the asymptotics in $\pm \infty$ are also zero. Let us observe that the SOV characterization of the $\tau_2$-spectrum implies that the system of equations \rf{FrbtSOVBax1} admits a non-zero solution, i.e. $t(\lambda )\in\Sigma _{{\tau_2}}$ only if:
\begin{equation}
\det_{p}\text{$D$}(\eta^{p} _{a})=0\text{ \ \ }\forall a\in \{1,...,\mathsf{N}-1\} \  \ \text{ and} \  \ (\eta _{1},...,\eta _{\mathsf{N}})\in 
\SZ_\SB.  \label{Frbtcompatibility}
\end{equation}
Indeed, the SOV coefficients lead to the
same values of the quantum determinant and the averages of $\SA$ and $\SD$ in the zeros of
$\SB$. Besides, we have:
\begin{equation}
\lim_{\log \Lambda \rightarrow \mp \infty }\Lambda ^{\pm \mathsf{N}}\det_{p}\text{%
$D$}(\Lambda )=0,  \label{Frbtasymp-compatibility}
\end{equation}
which simply follows by observing that:
\begin{equation}
\lim_{\log \Lambda \rightarrow \mp \infty }\Lambda ^{\pm \mathsf{N}}\det_{p}\text{$%
D$}(\Lambda )=-\det_{p}\left\Vert a_{\mp }\delta _{i,j-1}+d_{\mp }\delta
_{i,j+1}-(q^{\mp k}a_{\mp }+q^{\pm k}d_{\mp }))\delta _{i,j}\right\Vert =0,
\end{equation}
for $t(\lambda )\in\Sigma _{{\tau_2}}^{k}$ and $k\in\{0,...,2l\}$.
\end{proof}

{\bf Remark 2.}  \ The same type of functional equation $\det D(\Lambda )=0$ also appears for different quantum integrable models in \cite{FrbtBR89,FrbtNe02,FrbtNe03}. There, it stands for the functional relations which result from the truncated fusions of transfer matrix eigenvalues. In the case of the $\tau_2$-model this type of fusion leads to the same type of equation, as it has been derived in \cite{FrbtBBP90,FrbtBa89,FrbtBa04,FrbtGIPS06}.

\subsection{Construction of $\tau_2$-eigenstates}\label{FrbtT-eigenstates}
Thanks to the previous results we can give a complete characterization of the set $\Sigma_{\tau_2}$ and construct one $\tau_2$-eigenstate $|t\rangle$ for any $t(\la)\in\Sigma_{\tau_2}$. 
\begin{thm}\label{FrbtC:T-eigenstates}
$\Sigma _{{\tau _{2}}}$ coincides with the set of solutions to:
\begin{equation}
\det_{p}\text{$D$}(\Lambda )=0,\text{ \ \ }\forall \Lambda \in \mathbb{C},
\label{FrbtI-Functional-eq}
\end{equation}in the class of functions (\ref{Frbtset-t}). Then, we can associate a $\tau _{2}$-eigenstate: 
\begin{equation}
\Psi _{t}(\eta )\equiv \langle \,\eta _{1},...,\eta _{\mathsf{N}}\,|\,t\,\rangle
=\eta _{\mathsf{N}}^{-k}\prod_{r=1}^{\mathsf{N}-1}Q_{t}(\eta _{r}),
\label{FrbtQeigenstate-even}
\end{equation}to any $t(\lambda)\in \Sigma _{{\tau_2}}^{k}$, where $Q_{t}(\la)$ is the unique solution (up to quasi-constants) corresponding to $t(\la)$ of the
Baxter equation:
\begin{equation}
t(\la)Q_{t}(\la)=\bar{\textsc{a}}(\la)Q_{t}(\la/q)+\bar{\textsc{d}}(\la)Q_{t}(q\la).
\label{FrbtBaxter-eq-eigenvalues}
\end{equation}
\end{thm}

\begin{proof}
Lemma \ref{FrbtCharact-Sigma1} implies that any element of $\Sigma _{\tau _{2}}$ is a solution of \rf{FrbtI-Functional-eq}.
The reverse inclusion is given by the fact that the condition (\ref{FrbtI-Functional-eq}) is just the requirement of the existence of a non-zero
solution $Q_{t}(\la)$ of the equation (\ref{FrbtBaxter-eq-eigenvalues}). Then, $t(\la)$ and the $\Psi _{t}(\eta )$ defined in \rf{FrbtQeigenstate-even} are solutions of the discrete system of Baxter-like equations 
\rf{FrbtSOVBax1} and \rf{FrbtSOVBax2} and so they define a $\tau _{2}$-eigenvalue
and the corresponding $\tau _{2}$-eigenstate. Let us show that $Q_{t}(\lambda )$ is unique; denoting with $\bar Q_{t}(\lambda )$ any other solution, we can define the q-Wronskian:
\begin{equation}\label{Frbtq-W}
W_t(\la)\,=\,Q_t(\la)\bar Q_t(q^{-1}\la)-\bar Q_t(\la)Q_t(q^{-1}\la)\,,
\end{equation}
which by the Baxter equation satisfies the equation:
\begin{equation}
\bar{\textsc{a}}(\lambda )\,W_t(\lambda )\,=\,\bar{\textsc{d}}(\lambda )\,W_t(q \lambda )\,.
\end{equation}
Thanks to the cyclicity the average of the above equation reads:
\begin{equation}
(\prod_{k=1}^{p}\bar{\textsc{a}}(\lambda q^k)-\prod_{k=1}^{p}\bar{\textsc{d}}(\lambda q^k))\,\CW_t(\Lambda)\,=0\,\,\,\,\,\text{ with }\,\,\,\,\,\CW_t(\Lambda)\equiv\prod_{k=1}^{p}W_t(\lambda q^k),
\end{equation}
which implies $W(\lambda )\equiv 0$, being: 
\begin{equation}
\prod_{k=1}^{p}\bar{\textsc{a}}(\lambda q^k)\neq\prod_{k=1}^{p}\bar{\textsc{d}}(\lambda q^k)\qquad\,\,\, \text{or equivalently}\qquad\,\,\, \Omega _{+}\left( \Lambda \right) \neq \Omega _{-}\left(
\Lambda \right).
\end{equation}
It is then easy to see that this implies that $\bar Q_t(\la)\equiv Q_t(\la)$ up to quasi-constant normalization.
\end{proof}

\section{Characterization of $\tau_2$-spectrum: self-adjoint representations}
\setcounter{equation}{0}
\subsection{SOV reconstruction of a Baxter $\SQ$-operator}
The interest toward the self-adjoint representations of the $\tau_2$-model is due to the use
of the spectral theorem. The $\tau_2(\la)$ transfer matrix is then diagonalizable and the
Theorem \ref{FrbtC:T-eigenstates} leads to the construction of the $\tau_2$-eigenbasis; this in
particular means that the self-adjointness of $\tau_2(\la)$ forces the Baxter equation
\rf{FrbtBaxter-eq-eigenvalues} to have a complete set of independent solutions. 
Moreover, it is interesting to note that in the self-adjoint representations,
from the results of our SOV analysis, a Baxter $\SQ$-operator for the $\tau_{2}$-model is automatically reconstructed. In particular, we can define the operator:
\begin{equation}\label{FrbtQ-by-SOV-def}
\text{Q}(\lambda )|t\rangle \equiv
Q_{t}(\lambda )|t\rangle 
\end{equation}%
where $Q_{t}(\lambda )$ satisfies the Baxter equation \rf{FrbtBaxter-eq-eigenvalues} with $t(\lambda
)\in \Sigma _{t}$ and $|t\rangle $ is the corresponding $\tau _{2}$-eigenstate.
Then, the set of the $\tau _{2}$-eigenstates being complete and the
function $Q_{t}(\lambda )$ being uniquely defined (up to quasi constants), the
operator Q$(\lambda )$ is well defined, up to a quasi-constant scalar quantity, and clearly satisfies the
properties:%
\begin{equation}
\text{Q}(\lambda )\tau _{2}(\lambda )=\bar{\textsc{a}}(\lambda)\text{Q}(\lambda /q)+\bar{\textsc{d}}(\lambda )\text{Q}(\lambda q),\text{ \  \ }[\text{Q}(\lambda ),\text{Q}(\mu )]=0,\text{ \  \ }[\tau _{2}(\lambda ),\text{Q}(\mu )]=0\text{ \ \  }\forall \lambda ,\mu \in \mathbb{C}.
\end{equation}
In the next subsections we restrict our attention to special self-adjoint representations of the
$\tau_2$-model. Similarly to the case of the sine-Gordon model described in \cite{FrbtNT,FrbtGN10,FrbtGN11}, we show that for these self-adjoint representations the transfer matrix spectrum (eigenvalues
and eigenstates) is completely characterized in terms of {\bf polynomial} solutions of
the associated functional Baxter equation. Moreover, we prove the {\bf completeness} of the
set of the transfer matrix eigenstates constructed by
using SOV from the solutions of the associated
Bethe ansatz equations.

\subsection{Construction of polynomial Baxter equation solutions from $\tau_2$-eigenvalues}\label{FrbtBax-FE}
Let us consider the subvariety in the self-adjoint $\tau_2$-representations of real dimension $4\mathsf{N}$ characterized by the constrains:
\begin{equation}\label{FrbtS-Adj-c}
\prod_{h=1}^{N}\frac{\alpha _{h}^{\ast }}{\alpha _{h}}=1,\quad\quad\text{ \ \ }
\frac{\mathbbm{b}_{n}}{\mathbbm{b}_{n}^{\ast }}= 
\frac{\mathbbm{a}_{n}}{\mathbbm{a}_{n}^{\ast }},\quad\quad\text{ \ \ }\frac{\alpha _{n+1}^{\ast }\alpha
_{n}^{\ast }}{\alpha _{n+1}\alpha _{n}} =\frac{b _{n+1}^{\ast }\mathbbm{b}_{n}}{\mathbbm{b}_{n+1} b^{\ast }_{n}},\quad\text{ \ \ \ \ \ }\forall n\in \{1,...,\mathsf{N}\}.
\end{equation}
\begin{lem}\label{FrbtBaxter-coeff-S-adj}
The Laurent polynomials:
\begin{equation}
{\tt a}(\lambda )\equiv i^{\mathsf{N}}\prod_{n=1}^{\mathsf{N}}\frac{\beta _{n}}{\lambda }%
(1-i^{(1+\epsilon)/2}q^{-1/2}\frac{|\alpha _{n}|}{|\mathbbm{a}_{n}|}\lambda )(1-i^{(1+\epsilon)/2}q^{-1/2}\frac{|\alpha
_{n}|}{|\mathbbm{b}_{n}|}\lambda ),\qquad{\tt d}(\lambda )\equiv q^{\mathsf{N}}{\tt a}(-\lambda
q)  \label{FrbtS-adj-coeff},
\end{equation}
where $\epsilon=\pm$, satisfy the equations:
\begin{equation}
\det\,\hspace{-0.08cm} _{\text{q}}\SM(\lambda )={\tt a}(\lambda ){\tt d}(\lambda /q),\text{ \ \ }%
\forall \lambda \in \mathbb{C},  \label{Frbtdet-D-1}
\end{equation}%
and
\begin{equation}
\mathcal{A}(\Lambda )+\mathcal{D}(\Lambda )=\prod_{n=1}^{p}{\tt a}(\lambda
q^{n})+\prod_{n=1}^{p}{\tt d}(\lambda q^{n}),\text{ \ \ }\forall \lambda \in 
\mathbb{C}.  \label{Frbtdet-D-2}
\end{equation}
\end{lem}
\begin{proof} 
The condition \rf{Frbtdet-D-1} is trivially verified. The condition \rf{Frbtdet-D-2} is instead a consequence of Proposition \ref{FrbtAB_B-vs-Average} and, in particular, of the average properties \rf{Frbtcompatibility-detD}-\rf{FrbtQ-det-ChP-Baxter} once we prove that the coefficients ${\tt a}(\la)$ and $a_B(\la)$, defined in \rf{FrbtQ-coeff-a}, have the same average values. This last statement trivially follows by direct verification once we make the identifications:
\begin{equation}
r_{n}\equiv \left(-\epsilon\frac{\mathbbm{b}_{n-1}\alpha _{n-1}}{\mathbbm{b}_{n-1}^{\ast }\alpha
_{n-1}^{\ast }}\right)^{1/2}=\left(-\epsilon\frac{\alpha _{n}^{\ast }\mathbbm{a}_{n}}{\alpha
_{n}\mathbbm{a}_{n}^{\ast }}\right)^{1/2}\text{ \ \ \ \ \ }\forall n\in
\{1,...,\mathsf{N}\},
\end{equation}
which are consistent with the requirement of cyclicity \rf{Frbtcyclicity} under the
constrains \rf{FrbtS-Adj-c}.  
\end{proof} 
\textbf{Remark 3.} Let us remark that the lemma can also be proven under the
condition:
\begin{equation}
\mathcal{A}(\Lambda )=\mathcal{D}(\Lambda ),  \label{FrbtA=D}
\end{equation}%
and that the self-adjoint $\tau _{2}$-representations for which \rf{FrbtA=D}
holds also in any quantum site, i.e. those for which $(\alpha _{1},...,\alpha _{\mathsf{N}},\mathbbm{a}_{1},...,\mathbbm{a}_{\mathsf{N}},\mathbbm{b}_{1},...,\mathbbm{b}_{\mathsf{N}})\in 
\mathbb{R}^{3\mathsf{N}},$ is a subvariety of the one defined in \rf{FrbtS-Adj-c}.

It is worth remarking that, in the class of self-adjoint representations
defined by the constrain (\ref{FrbtS-Adj-c}), it results:%
\begin{equation}
a_{+}=d_{+}\text{ \ \ and \ \ }a_{-}=d_{-}.
\end{equation}
Then, in these representations, the transfer matrix spectrum presents a
double degeneracy for $\left\vert k\right\vert \neq 0$, since the asymptotic
behaviour of a $t(\lambda )\in \Sigma _{{\tau _{2}}}^{k}$ and a $t^{\prime
}(\lambda )\in \Sigma _{{\tau _{2}}}^{-k}$ coincide. Such a degeneracy in
the $\tau _{2}$-spectrum is resolved by the $\Theta $-charge, i.e. the couple $%
(\tau _{2}(\lambda ),\Theta )$ has simple spectrum. The proof of this last
statement can be given by adapting the proofs of Proposition 5 and Lemma 2
of \cite{FrbtNT} to these self-adjoint representations of the $\tau _{2}$-model.
\begin{thm}\label{FrbtDerivation-Baxter-functional}
Let us take a $t(\lambda )\in \Sigma _{{\tau _{2}}}^{k}$, then there exists, up to multiplicative quasi-constants, a unique polynomial of the form:
\begin{equation}\label{FrbtQ_t-definition}
Q_{t}(\lambda )=\lambda ^{a_{t}}\prod_{h=1}^{2l\mathsf{N}-(b_{t}+a_{t})}(%
\lambda _{h}-\lambda ),\,\,\,\,\,\,\,\,0\leq a_{t}\leq 2l,\,\,0\leq
b_{t}+a_{t}\leq 2l\mathsf{N},
\end{equation}
which satisfy the Baxter equation: 
\begin{equation}\label{Frbttq-Baxter}
t(\lambda )Q_{t}(\lambda )=\mathtt{a}(\lambda )Q_{t}(\lambda q^{-1})+\mathtt{%
d}(\lambda )Q_{t}(\lambda q)\ \ \ \ \forall \lambda \in \mathbb{C},
\end{equation}
where:
\begin{equation}
a_{t}=\pm k\,\,\mathsf{mod}\,p,\,\,\,\,\,\,\,\,\,b_{t}=\pm k\,\,\mathsf{mod}%
\,p\text{.}  \label{Frbtasymptitics-Q}
\end{equation}
\end{thm}

\begin{proof}The uniqueness up to multiplicative quasi-constants is a consequence of Theorem \ref{FrbtC:T-eigenstates}. So we only have to prove the existence of such a
polynomial $Q_{t}(\lambda )$ and the fact that it has the form defined in \rf{FrbtQ_t-definition} and \rf{Frbtasymptitics-Q}. First, let us define the matrix $\tilde{D}(\lambda)$ by replacing $\bar{\textsc{a}}$ and $\bar{\textsc{d}}$ by $\mathtt{a}$ and $\mathtt{d}$ in $D(\lambda)$. From the Lemma \ref{FrbtBaxter-coeff-S-adj} and since the determinant of $D(\lambda)$ only depends of $\bar{\textsc{a}}$ and $\bar{\textsc{d}}$ via their average values and the quantum determinant, $D(\lambda)$ and $\tilde{D}(\lambda)$ have the same determinant. It is worth noticing that the condition $t(\lambda )\in \Sigma _{{\tau _{2}}}^{k}$
implies that the $p\times p$ matrix $\tilde{D}(\lambda )$ has rank $2l$ for any $\lambda \in \mathbb{C}\backslash \{0\}$. Then denoting with 
\begin{equation}
\text{\textsc{C}}_{i,j}(\lambda )=(-1)^{i+j}\det_{2l}\tilde{D}_{i,j}(\lambda )
\label{Frbtcofactor-def}
\end{equation}%
the $(i,j)$ \textit{cofactor} of the matrix $\tilde{D}(\lambda )$, the matrix of elements these cofactors has rank $1$ and so the vectors: 
\begin{equation}
\text{\textsc{V}}_{i}(\lambda )\equiv (\text{\textsc{C}}_{i,1}(\lambda ),%
\text{\textsc{C}}_{i,2}(\lambda ),...,\text{\textsc{C}}_{i,2l+1}(\lambda ))^{%
\ST}\in \mathbb{C}^{p}\text{ \ \ }\forall i\in \{1,...,2l+1\}
\end{equation}%
satisfy the proportionality conditions: 
\begin{equation}
\text{\textsc{V}}_{i}(\lambda )/\text{\textsc{C}}_{i,1}(\lambda )=\text{%
\textsc{V}}_{j}(\lambda )/\text{\textsc{C}}_{j,1}(\lambda )\text{\ \ \ \ }%
\forall i,j\in \{1,...,2l+1\},\text{ }\forall \lambda \in \mathbb{C}.
\label{Frbtcovector-proport}
\end{equation}%
These conditions implies: 
\begin{equation}
\text{\textsc{C}}_{2,2}(\lambda )/\text{\textsc{C}}_{2,1}(\lambda )=\text{%
\textsc{C}}_{1,2}(\lambda )/\text{\textsc{C}}_{1,1}(\lambda ),
\label{Frbtproportionality}
\end{equation}%
and for (\ref{Frbtcofactors-diagonal}) we can rewrite it as:%
\begin{equation}
\text{\textsc{C}}_{1,1}(\lambda q)/\text{\textsc{C}}_{1,2l+1}(\lambda q)=%
\text{\textsc{C}}_{1,2}(\lambda )/\text{\textsc{C}}_{1,1}(\lambda ).
\label{FrbtInter-step}
\end{equation}%
Moreover, from the relation:%
\begin{equation}
\text{\textsc{C}}_{1,2l+1}(\lambda )=q^{\mathsf{N}}\text{\textsc{C}}%
_{1,2}(-\lambda /q),  \label{FrbtC_1,2l+1--C_1,2}
\end{equation}%
which follows from (\ref{Frbtcofactors-diagonal}) and  (\ref{Frbtcofactors-parity-0}), (\ref{Frbtproportionality}) is also equivalent to:%
\begin{equation}
\text{\textsc{C}}_{1,1}(\lambda q)\text{\textsc{C}}_{1,1}(\lambda )=q^{\mathsf{N}}%
\text{\textsc{C}}_{1,2}(\lambda )\text{\textsc{C}}_{1,2}(-\lambda )=q^{-\mathsf{N}}%
\text{\textsc{C}}_{1,2l+1}(\lambda )\text{\textsc{C}}_{1,2l+1}(-\lambda ),
\label{Frbtcofactor-equality}
\end{equation}%
and from the condition $\det_{p}\tilde{D}(\Lambda )=0$, we have: 
\begin{equation}
t(\lambda )\text{\textsc{C}}_{1,1}(\lambda )=\mathtt{a}(\lambda )\text{%
\textsc{C}}_{1,2l+1}(\lambda )+\mathtt{d}(\lambda )\text{\textsc{C}}%
_{1,2}(\lambda ).  \label{FrbtBax-eq}
\end{equation}

Let us note that all the cofactors are Laurent polynomial of maximal degree%
\footnote{%
The $a_{i,j}$ and $b_{i,j}$ are non-negative integers and $\lambda
_{h}^{(i,j)}\neq 0$ for any $h\in \{1,...,4l\mathsf{N}-(a_{i,j}+b_{i,j})\}$.} $2l%
\mathsf{N}$ in $\lambda $: 
\begin{equation}
\text{\textsc{C}}_{i,j}(\lambda )=\text{\textsc{c}}_{i,j}\lambda ^{-2l\mathsf{N}%
+a_{i,j}}\prod_{h=1}^{4l\mathsf{N}-(a_{i,j}+b_{i,j})}(\lambda _{h}^{(i,j)}-\lambda
),
\end{equation}
as it follows from the form of $\mathtt{a}(\lambda ),$ $\mathtt{d}(\lambda )$
and $t(\lambda )\in \Sigma _{{\tau _{2}}}$. Thanks to formula (\ref{Frbtcofactors-parity-0}), the cofactor \textsc{C}$_{1,1}(\lambda )\in \mathbb{C}%
[\lambda ,\lambda ^{-1}]_{2l\mathsf{N}}$ is even in $\lambda $:%
\begin{equation}
\text{\textsc{C}}_{1,1}(\lambda )=\text{\textsc{c}}_{1,1}\lambda ^{-2l\mathsf{N}+2%
\tilde{a}_{1,1}}\prod_{i=1}^{2l\mathsf{N}-(\tilde{a}_{1,1}+\tilde{b}%
_{1,1})}(\lambda _{i}^{(1,1)}-\lambda )(\lambda _{i}^{(1,1)}+\lambda ),
\end{equation}%
then the cofactor equation (\ref{Frbtcofactor-equality}) implies:%
\begin{equation}
a_{1,2}=2\tilde{a}_{1,1}\equiv 2\bar{a}\text{, \ \ }b_{1,2}=2\tilde{b}%
_{1,1}\equiv 2\bar{b}\text{, \ \ \textsc{c}}_{1,2}^{2}=\text{\textsc{c}}%
_{1,1}^{2}q^{-2(\mathsf{N}+\bar{b})}  \label{Frbt1-rel}
\end{equation}%
and:%
\begin{equation}
\left( \lambda _{i}^{(1,1)}\right) ^{2}=\left( \lambda _{i}^{(1,2)}\right)
^{2}\equiv \bar{\lambda}_{i}^{2},\text{ \ }\left( \lambda _{i+2l\mathsf{N}-(\bar{a}%
+\bar{b})}^{(1,2)}\right) ^{2}=\left( \bar{\lambda}_{i}/q\right) ^{2},
\label{Frbtcofactor-zeros}
\end{equation}%
with $\bar{\lambda}_{i}\neq 0$ for any $i\in \{1,...,2l\mathsf{N}-(\bar{a}+\bar{b}%
)\}$ with $\bar{a}$ and $\bar{b}\in \mathbb{Z}^{\geq 0}$. Then, we can write: 
\begin{eqnarray}
\text{\textsc{C}}_{1,1}(\lambda ) &=&\text{\textsc{c}}_{1,1}\lambda ^{-2l\mathsf{N}%
+2\bar{a}}\prod_{i=1}^{2l\mathsf{N}-(\bar{a}+\bar{b})}(\bar{\lambda}_{i}+\lambda )(%
\bar{\lambda}_{i}-\lambda ),  \label{Frbtcof-11} \\
\text{\textsc{C}}_{1,2}(\lambda ) &=&\text{\textsc{c}}_{1,2}q^{\bar{a}+b+\mathsf{N}%
}\lambda ^{-2l\mathsf{N}+2\bar{a}}\prod_{i=1}^{2l\mathsf{N}-(\bar{a}+\bar{b})}(\bar{%
\lambda}_{i}+\lambda )(\epsilon_i \bar{\lambda}_{i}-\lambda q),
\label{Frbtcof-12}
\end{eqnarray}%
where the $\epsilon_i= \pm 1$. The property (\ref{FrbtCo-F prop2}), the relations (\ref{FrbtC_1,2l+1--C_1,2}) and \rf{FrbtInter-step} imply that $\forall i, \epsilon_i=1$. Let us now
introduce the polynomials $\overline{\text{\textsc{C}}}_{1,1}(\lambda )$, $%
\overline{\text{\textsc{C}}}_{1,2l+1}(\lambda )$ and $\overline{\text{%
\textsc{C}}}_{1,2}(\lambda )$ defined by simplifying the common factors in 
\textsc{C}$_{1,1}(\lambda )/$\textsc{c}$_{1,1}$, \textsc{C}$%
_{1,2l+1}(\lambda )/$\textsc{c}$_{1,2l+1}$ and \textsc{C}$_{1,2}(\lambda )/$%
\textsc{c}$_{1,2}$, respectively. Then, $\overline{\text{\textsc{C}}}%
_{1,1}(\lambda )$ has the form: 
\begin{equation}
\overline{\text{\textsc{C}}}_{1,1}(\lambda )=\prod_{h=1}^{\mathsf{N}%
_{1,1}}(\lambda _{h}-\lambda ),\,\,\,\,\ \text{with \ }\mathsf{N}_{1,1}\leq 2l\mathsf{N}%
-(\bar{a}+\bar{b})\text{.}
\end{equation}
Here and in the following, we will write {\large {Z}}$_{f(\la)}$ to denote the set of the zeros of the function $f(\la)$,
then it holds:%
\begin{equation}
\text{{\large {Z}}}_{\overline{\text{\textsc{C}}}_{1,1}(\lambda )}\subset \{\bar{\lambda}_{1},...,\bar{\lambda}_{2l\mathsf{N}-(\bar{a}+b)}\},\,\,\,\,\,\,\,
0\notin \text{{\large {Z}}}_{\overline{\text{\textsc{C}}}_{1,1}(\lambda
)},\,\,\,\,\,\,\,\text{{\large s}}_{p,\lambda _{0}}\not\subset \text{{\large {Z}}}_{%
\overline{\text{\textsc{C}}}_{1,1}(\lambda )},
\end{equation}
for any $\text{{\large s}}_{p,\lambda _{0}}\equiv (\lambda _{0},q\lambda_{0},...,q^{2l}\lambda _{0})$, a $p$-string of center $\lambda _{0}$ $\in \mathbb{C}$.
Now, since by definition $\text{{\large {Z}}}_{\overline{\text{\textsc{C}}}%
_{1,1}(\lambda )}\cap \text{{\large {Z}}}_{\overline{\text{\textsc{C}}}%
_{1,2}(\lambda )}=\text{{\large {Z}}}_{\overline{\text{\textsc{C}}}%
_{1,1}(\lambda )}\cap \text{{\large {Z}}}_{\overline{\text{\textsc{C}}}%
_{1,2l+1}(\lambda )}=\emptyset $, equation (\ref{FrbtInter-step}) implies: 
\begin{equation}
\overline{\text{\textsc{C}}}_{1,2l+1}(\lambda )=q^{\mathsf{N}_{1,1}}\text{$%
\overline{\text{\textsc{C}}}$}_{1,1}(\lambda q^{-1}),\text{ \ \ $\overline{%
\text{\textsc{C}}}$}_{1,2}(\lambda )=q^{-\mathsf{N}_{1,1}}\text{$\overline{\text{%
\textsc{C}}}$}_{1,1}(\lambda q).  \label{Frbtproportionality-cofactor}
\end{equation}%
Then, equation (\ref{FrbtBax-eq}) can be written as a Baxter equation for the
polynomial $\overline{\text{\textsc{C}}}_{1,1}(\lambda )$: 
\begin{equation}
t(\lambda )\text{$\overline{\text{\textsc{C}}}$}_{1,1}(\lambda )=\bar{%
\mathtt{a}}(\lambda )\text{$\overline{\text{\textsc{C}}}$}_{1,1}(\lambda
q^{-1})+\bar{\mathtt{d}}(\lambda )\text{$\overline{\text{\textsc{C}}}$}%
_{1,1}(\lambda q),  \label{Frbtdeform-BAX}
\end{equation}%
with coefficients $\bar{\mathtt{a}}(\lambda )\equiv q^{\mathsf{N}_{1,1}}\varphi 
\mathtt{a}(\lambda )$ and $\bar{\mathtt{d}}(\lambda )\equiv q^{-\mathsf{N}%
_{1,1}}\varphi ^{-1}\mathtt{d}(\lambda )$ and $\varphi \equiv $\textsc{c}$%
_{1,1}/$\textsc{c}$_{1,2}=$\textsc{c}$_{1,2l+1}/$\textsc{c}$_{1,1}$. Note
that the consistence\footnote{%
The proof of this statement coincides step by step with the proof given in
Theorem 2 of \cite{FrbtGN10}.} of the above equation implies that $\varphi $ is
a $p$-root of the unit and then, by (\ref{Frbt1-rel}), we have $\varphi \equiv
q^{(\bar{b}+\mathsf{N})}$. 

Let us define now the polynomial 
\begin{equation}
Q_{t}(\lambda )\equiv \lambda ^{a_{t}}\text{$\overline{\text{\textsc{C}}}$}%
_{1,1}(\lambda )=\lambda ^{a_{t}}\prod_{h=1}^{2l\mathsf{N}-(a_{t}+b_{t})}(\lambda
_{h}-\lambda ),  \label{FrbtQ_t-form}
\end{equation}%
where: 
\begin{equation}
\mathsf{N}_{1,1}\equiv 2l\mathsf{N}-(a_{t}+b_{t}),\text{ \ }a_{t}\in \{0,...,2l\}\text{,
\ }b_{t}\equiv \bar{b}+mp\text{, }m\in \mathbb{Z}^{\geq 0}\text{;}
\end{equation}%
then $q^{-a_{t}}\equiv q^{\mathsf{N}_{1,1}}\varphi $, $Q_{t}(\lambda )$ is
the desired solution of the Baxter equation (\ref{Frbttq-Baxter}) which
belongs to $\mathbb{C}[\lambda ]_{2l\mathsf{N}}$. The characterizations (\ref{Frbtasymptitics-Q}) are then trivial consequences of the asymptotics of $%
t(\lambda )\in \Sigma _{\tau _{2}}^{k}$\ plus those of the coefficients of
the Baxter equation (\ref{Frbttq-Baxter}) satisfied by $Q_{t}(\lambda )$.
\end{proof}
\begin{lem}
The polynomial solution $Q_{t}(\lambda )$ of the Baxter equation (\ref{Frbttq-Baxter})
associated to $t(\lambda )\in \Sigma _{{\tau _{2}}}$
is a $\epsilon$-real polynomial, i.e. it satisfies the following complex-conjugation condition:
\begin{equation}
\left( Q_{t}(\lambda )\right) ^{\ast }\equiv Q_{t}(\epsilon \lambda
^{\ast })\text{ \ }\qquad\forall \lambda \in \mathbb{C},
\label{FrbtQ-Complex-conjugation}
\end{equation}%
where $\epsilon =\pm 1$ is the discrete parameter defined in (\ref{FrbtSelf-adjointness Condition}).
Then, the Baxter $\SQ$-operator, $Q(\lambda )$ defined in \rf{FrbtQ-by-SOV-def}, is $\epsilon$-self-adjoint and it is a polynomial of maximal degree $2l\mathsf{N}$ in the spectral parameter $\la$.

\end{lem}
\begin{proof}Let us remark that the self-adjointness of the transfer matrix $\tau
_{2}(\lambda )$ and its parity properties imply:%
\begin{equation}
\left( t(\lambda )\right) ^{\ast }\equiv t(\lambda ^{\ast })\text{ \ \ and \
\ }t(-\lambda )\equiv (-1)^{\mathsf{N}}t(\lambda )\text{\ \ \ }\forall
\lambda \in \mathbb{C}\text{,}
\end{equation}%
while the coefficients of the Baxter equations satisfy the
complex-conjugation conditions:%
\begin{equation}
\left( \mathtt{a}(\lambda )\right) ^{\ast }\equiv \epsilon ^{\mathsf{N}}%
\mathtt{d}(\epsilon \lambda ^{\ast })\text{ \ \ \ \ }\forall \lambda \in 
\mathbb{C}\text{.}
\end{equation}%
Then, under complex conjugation, the Baxter equation reads: 
\begin{equation}
t(\epsilon \lambda ^{\ast })\left( Q_{t}(\lambda )\right) ^{\ast }=%
\mathtt{d}(\epsilon \lambda ^{\ast })\left( Q_{t}(\lambda q^{-1})\right)
^{\ast }+\mathtt{a}(\epsilon \lambda ^{\ast })\left( Q_{t}(\lambda
q)\right) ^{\ast }\ \ \ \ \forall \lambda \in \mathbb{C},
\end{equation}
and so the injectivity of the map $t(\lambda )\in \Sigma _{\tau _{2}}$ $%
\rightarrow Q_{t}(\lambda )\in \mathbb{C}[\lambda ]_{2l\mathsf{N}}$ implies the
condition (\ref{FrbtQ-Complex-conjugation}).
\end{proof}
{\bf Remark 4.} \ For any $t(\lambda )\in \Sigma _{{\tau_2}}$ the Theorem \ref{FrbtDerivation-Baxter-functional} defines a link between the corresponding polynomial solution of the Baxter equation and a determinant of a $(p-1)\times(p-1)$ tridiagonal matrix. Let us note then in literature there exist also other determinant characterizations of solutions of the Baxter equation. In the quantum periodic Toda chain, linear combinations of determinants of semi-infinite tridiagonal matrices allow to express these solutions \cite{FrbtGP,FrbtGM,FrbtKL99}. A careful analysis of the $\tau_2$-model in the limit $\beta^2\rightarrow\bar{\beta}^2$ with $\bar{\beta}^2$ irrational (i.e. $p',\ p\rightarrow + \infty$) is of clear interest as the dimension of the representation of the Weyl algebra ${\cal W}_{n}$ and the size of the tridiagonal matrix, related to the Baxter equation solution, diverge. Subsequently, one could ponder whether the characterization which holds in the case of the quantum periodic Toda chain also applies to the $\tau_2$-model for irrational $\bar{\beta}^2$. This could lead to a reformulation\footnote{NLIE has been used first to reformulate the spectrum of integrable quantum models in \cite{FrbtKP91,FrbtKBP91}. Similar NLIE were presented in \cite{FrbtKT10} and also in \cite{FrbtBLZ-I,FrbtZa,FrbtBT06,FrbtT}.} of the $\tau_2$-spectrum in terms of solutions of nonlinear integral equations (NLIE) related to the results in \cite{FrbtKT10}. Interestingly, another DDV-type\footnote{This type of NLIE was introduced and analyzed in \cite{FrbtDDV92}-\cite{FrbtR01} for the fermionic lattice regularizations of the sine-Gordon model, see also \cite{FrbtFR02I}-\cite{FrbtFR03II} for a related model.} NLIE reformulation can be introduced and it describe the complete $\tau_2$-spectrum thanks to the completeness of the Bethe ansatz.
\subsection{Completeness of Bethe ansatz type equations}\label{FrbtCompl-BA}
The previous analysis leads to the complete characterization of the transfer matrix spectrum (eigenvalue and eigenstate) in terms of real (in the case $\epsilon=1$) or imaginary ($\epsilon=-1$) polynomial solutions of the Baxter equation \rf{Frbttq-Baxter}.

\begin{propn}To any $\tau_2$-self-adjoint representation which further satisfies the conditions \rf{FrbtS-Adj-c} is associated a system of Bethe ansatz type equations and there exists a one to one map between the set of all the p-strings free and $\epsilon$-self-adjoint solutions of it and the set $\Sigma _{\tau_2}$.
\end{propn}
\begin{proof}Let us construct the isomorphism. Theorem \ref{FrbtDerivation-Baxter-functional} and the
previous lemma imply that to any $t(\lambda )\in \Sigma _{{\tau _{2}}}^{k}$
it is associated one and only one $p$-strings free and $\epsilon $-self-adjoint\footnote{It means that the tuple $(\lambda _{1},\dots ,\lambda _{2l\mathsf{N}%
-(a_{t}+b_{t})})$ satisfies the following property under complex conjugation:%
\begin{equation*}
(\lambda _{1}^{\ast },\dots ,\lambda _{2l\mathsf{N}-(a_{t}+b_{t})}^{\ast
})=(\epsilon \lambda _{\pi (1)},\dots ,\epsilon \lambda _{\pi \left( 2l%
\mathsf{N}-(a_{t}+b_{t})\right) })
\end{equation*}%
where $\pi $ is a permutation of the indexes $\{1,...,2l\mathsf{N}-(a_{t}+b_{t})\}$%
.} solution $(\lambda _{1},\dots ,\lambda _{2l\mathsf{N}-(a_{t}+b_{t})})$ of the system of Bethe ansatz equations:%
\begin{equation}
\frac{\mathtt{a}(\lambda _{c})}{\mathtt{d}(\lambda _{c})}=-q^{2a_{t}}%
\prod_{h=1}^{2l\mathsf{N}-(a_{t}+b_{t})}\frac{(q\lambda _{c}-\lambda _{h})}{%
(\lambda _{c}/q-\lambda _{h})},\,\,\,\,\,\,\forall c\in \{1,...,2l\mathsf{N}%
-(a_{t}+b_{t})\},  \label{Frbteq-Bethe}
\end{equation}%
with $a_{t}=\pm k\,\mathsf{mod}\,p,\,\,b_{t}=\pm k\,\mathsf{mod}\,p$.

Conversely, given a $p$-strings free and $\epsilon $-self-adjoint solution of (\ref{Frbteq-Bethe}) with $a=\pm b\,\mathsf{mod}\,p$%
, then it defines uniquely a $\epsilon $-real polynomial $Q(\lambda )$ by the equation (%
\ref{FrbtQ_t-definition}). Now, by using $Q(\lambda )$ we can construct the
function: 
\begin{equation}
t(\lambda )\,=\,(\mathtt{a}(\lambda )Q(q^{-1}\lambda )+\mathtt{d}(\lambda
)Q(q\lambda ))/Q(\lambda )\,, \label{FrbtTfromQ}
\end{equation}%
which thanks to the Bethe equations (\ref{Frbteq-Bethe}) is nonsingular for $%
\lambda =\lambda _{c},\,\,\forall c\in \{1,\dots ,2l\mathsf{N}-(a+b)\}$ and is a
$\epsilon $-real Laurent polynomial of degree $\mathsf{N}$ in $\lambda $. Moreover, we can
also construct two states $|\,t_{\pm a}\,\rangle $ by inserting $Q(\lambda )$
in equation (\ref{FrbtQeigenstate-even}). So, $t(\lambda )$ and $\Psi _{t}(\eta )\equiv \langle \eta _{1},...,\eta _{\mathsf{N}%
}\,|\,t\,\rangle $ satisfy the system (\ref{FrbtSOVBax1}) and (\ref{FrbtSOVBax2}), for $k=\pm a$. Then the fixed solution of the Bethe equation allows us to reconstruct uniquely the $\tau _{2}$-eigenvalue $t(\lambda )$ and the simultaneous eigenstates $
|\,t_{\pm a}\,\rangle $  of $\tau _{2}(\lambda )$ and $\Theta $, with $\Theta $-eigenvalue $q^{\pm a}$.
\end{proof}
Note that the previous results can be rephrased as \textit{completeness of the Bethe ansatz type equations generated by SOV}, as from this type of Bethe solutions we can reconstruct the complete set of $\tau_2$-eigenstates.

\section{Spectrum characterization of the inhomogeneous chiral Potts model}
\setcounter{equation}{0}In \cite{FrbtBS} Bazhanov and Stroganov have proven the following remarkable connections between the integrable chiral Potts model and the $\tau_2$-model:
\begin{itemize}
\item[i)] The fundamental R-matrix intertwining the $\tau_2$-Lax operator in the
quantum space is given by the product of four chiral Potts Boltzmann
weights.
\item[ii)] The transfer matrix of the chiral Potts model is a Baxter $\SQ$-operator for the ${\tau _{2}}$-model.
\end{itemize}Both these statements are true when the $\tau_2$-model is restricted to parametrization by points on the algebraic curves\footnote{See the next subsection for precise definitions.} $\mathcal{C}_{k}$. In this section, we use the property ii) to characterize the eigenstates of the inhomogeneous chiral Potts transfer matrix  by the SOV construction.
\subsection{\label{FrbtBaxter-Q-BS}Transfer matrix of chiral Potts model as Baxter operator of $\tau_2$-model on the chP curves}
Let us consider the following tuple p$\equiv (a_{\text{p}},b_{\text{p}},c_{%
\text{p}},d_{\text{p}})\in \mathbb{C}^{4}$ and let us introduce the
following notations: 
\begin{equation}
x_{\text{p}}\equiv a_{\text{p}}/d_{\text{p}},\text{ \ }y_{\text{p}}\equiv b_{%
\text{p}}/c_{\text{p}},\text{ \ }s_{\text{p}}\equiv d_{\text{p}}/c_{\text{p}%
},\,t_{\text{p}}\equiv x_{\text{p}}y_{\text{p}},\text{ \ }k^{2}+(k^{^{\prime
}})^{2}=1,
\end{equation}%
then the algebraic curve $\mathcal{C}_{k}$ of modulus $k$ is by
definition the locus of the points p which satisfy the equations: 
\begin{equation}
x_{\text{p}}^{p}+y_{\text{p}}^{p}=k(1+x_{\text{p}}^{p}y_{\text{p}}^{p}),%
\text{ \ \ }kx_{\text{p}}^{p}=1-k^{^{\prime }}s_{\text{p}}^{-p},\text{ \ \ }%
ky_{\text{p}}^{p}=1-k^{^{\prime }}s_{\text{p}}^{p}.  \label{Frbtcurve-eq}
\end{equation}%
Similarly to \cite{FrbtBS}, we can introduce the parametrization:
\begin{equation}
\begin{array}{ll}
\lambda \alpha _{n}\equiv -t_{\text{p}}^{-1/2}b_{\text{q}_{n}}b_{\text{r}%
_{n}},\text{ \ \ \ \ \ \ } & \mathbbm{b}_{n}/rq^{1/2}\equiv (x_{\text{p}}/y_{\text{p}%
})^{1/2}a_{\text{q}_{n}}d_{\text{r}_{n}}/q^{2}, \\ 
\beta _{n}/\lambda \equiv -t_{\text{p}}^{1/2}d_{\text{q}_{n}}d_{\text{r}%
_{n}},\text{ \ \ \ \ \ \ } & q^{1/2}\mathbbm{a}_{n}/r\equiv -(x_{\text{p}}/y_{\text{p}%
})^{1/2}c_{\text{q}_{n}}b_{\text{r}_{n}}, \\ 
\delta _{n}\lambda \equiv q^{-2}t_{\text{p}}^{-1/2}a_{\text{q}_{n}}a_{\text{r%
}_{n}},\text{ \ \ \ \ \ \ } & q^{1/2}\mathbbm{d}_{n}r\equiv -(y_{\text{p}}/x_{\text{p}%
})^{1/2}d_{\text{q}_{n}}a_{\text{r}_{n}}, \\ 
\gamma _{n}/\lambda \equiv t_{\text{p}}^{1/2}c_{\text{q}_{n}}c_{\text{r}%
_{n}},\text{ \ \ \ \ \ \ } & q^{-1/2}\mathbbm{c}_{n}r\equiv (y_{\text{p}}/x_{\text{p}%
})^{1/2}b_{\text{q}_{n}}c_{\text{r}_{n}},%
\end{array}
\label{FrbtPoints on the curve-0}
\end{equation}
of the $\tau_2$-Lax operator in terms of the points p, q$_{n}$ and r$_{n}$ of the curve \rf{Frbtcurve-eq} of modulus $k$. Note that we can identify:%
\begin{equation}\label{Frbtgauge-r}
r\equiv (y_{\text{p}}/x_{\text{p}})^{1/2},
\end{equation}%
so that the off-diagonal elements of the $\tau _{2}$-Lax operator do not depend on the point p on the curve and the parameter $t_{\text{p}}^{-1/2}$
is then proportional to our spectral parameter:%
\begin{equation}
t_{\text{p}}^{-1/2}\equiv\lambda/\mathsf{c}_{0}.
\end{equation}

Then, we can introduce the operator $\ST_{\lambda }^{{\small \text{chP}}}$,
characterized by the kernel\footnote{%
For a direct comparison see formula (4.12) of \cite{FrbtBa08}
with the following identifications:%
\[
z_{j}\equiv q^{2\sigma _{j}^{\prime }},\text{ \ }z_{j}^{\prime }\equiv
q^{2\sigma _{j}}\text{ \ }\forall j\in \{1,...,\mathsf{N}\}.
\]Note that $\ST_{\lambda }^{{\small \text{chP}}}$ is well defined, the $W$%
-functions (\ref{FrbtW-function})\ being cyclic functions of their arguments.}:%
\begin{equation}
\ST_{\lambda }^{{\small \text{chP}}}(\text{z},\text{z}^{\prime })\equiv
\langle \text{z}|\ST_{\lambda }^{{\small \text{chP}}}|\text{z}^{\prime
}\rangle =\prod_{n=1}^{\mathsf{N}}W_{\text{q}_{n}\text{p}}(z_{n}/z_{n}^{\prime })%
\bar{W}_{\text{r}_{n}\text{p}}(z_{n}/z_{n+1}^{\prime }),  \label{Frbtkernel}
\end{equation}%
in the left and right $\su_{n}$-eigenbasis (\ref{Frbtu-basis}), where the functions $W$ and $\bar{W}$ are defined in (\ref{FrbtW-function}). $\ST_{\lambda }^{%
{\small \text{chP}}}$ coincides with the transfer matrix of the
inhomogeneous chiral Potts model \cite{FrbtBS}. As anticipated, $\ST_{\lambda }^{{\small 
\text{chP}}}$ is a Baxter $\SQ$-operator w.r.t. the $\tau _{2}$-transfer
matrix: 
\begin{equation}
\tau _{2}(\lambda )\ST_{\lambda }^{{\small \text{chP}}}=a_{\text{BS}%
}(\lambda )\ST_{\lambda /q}^{{\small \text{chP}}}+d_{\text{BS}}(\lambda
)\ST_{q\lambda }^{{\small \text{chP}}},  \label{FrbtBax-ChP-T-II}
\end{equation}%
where the coefficients read: 
\begin{eqnarray}
a_{\text{BS}}(\lambda ) &=&(-1)^{N}\prod_{n=1}^{N}\beta _{n}f_{\text{pq}_{n}%
\text{r}_{n}}(\frac{1}{\lambda }+q^{-1}\frac{\alpha _{n}\mathbbm{d}_{n}}{\beta
_{n}\mathbbm{c}_{n}}\lambda )\frac{1+\frac{q^{1/2}\lambda \mathbbm{b}_{n}}{\beta _{n}r}}{1+\frac{%
\lambda \alpha _{n}}{q^{1/2}r\mathbbm{c}_{n}}},  \label{FrbtQ-coeff-a} \\
d_{\text{BS}}(\lambda ) &=&(-1)^{N}\prod_{n=1}^{N}\beta _{n}f_{\text{pq}_{n}%
\text{r}_{n}}(\frac{1}{\lambda }+q\frac{\alpha _{n}\mathbbm{b}_{n}}{\beta _{n}\mathbbm{a}_{n}}%
\lambda )\frac{1-\frac{\lambda \mathbbm{d}_{n}r}{q^{1/2}\beta _{n}}}{1-\frac{%
q^{1/2}r\lambda \alpha _{n}}{\mathbbm{a}_{n}}},  \label{FrbtQ-coeff-d}
\end{eqnarray}%
and:%
\begin{equation}
f_{\text{pq}_{n}\text{r}_{n}}=\frac{W_{\text{q}_{n}\text{p}}(z(l))}{W_{\text{q}_{n}\text{p}}(z(0))}\frac{\bar{W}_{\text{r}_{n}\text{p}}(z(l))}{\bar{W}_{\text{r}_{n}\text{p}}(z(0))}.
\end{equation}
Note that to write these coefficients we have inverted the formulae (\ref{FrbtPoints on the curve-0}):%
\begin{eqnarray}
\frac{x_{\text{q}_{n}}}{y_{\text{p}}} &=&-q^{3/2}\frac{\lambda \mathbbm{b}_{n}}{\beta
_{n}r},\text{ \ \ }\frac{x_{\text{r}_{n}}}{x_{\text{p}}}=q^{1/2}\frac{%
\lambda \mathbbm{d}_{n}r}{\beta _{n}},  \label{Frbteq1} \\
\frac{y_{\text{q}_{n}}}{x_{\text{p}}} &=&q^{-1/2}\frac{r\lambda \alpha _{n}}{%
\mathbbm{a}_{n}},\text{ \ \ \ }\frac{y_{\text{r}_{n}}}{y_{\text{p}}}=-q^{1/2}\frac{%
\lambda \alpha _{n}}{r\mathbbm{c}_{n}},\text{ \ \ }s_{\text{q}_{n}}s_{\text{r}_{n}}=-%
\frac{\alpha _{n}\beta _{n}}{\mathbbm{c}_{n}\mathbbm{a}_{n}}.  \label{Frbteq2}
\end{eqnarray}%
It is worth pointing out that while the Baxter equation \rf{FrbtBax-ChP-T-II} holds in the general
inhomogeneous representations, the commutativity properties:%
\begin{equation}
\lbrack \tau _{2}(\lambda ),\ST_{\lambda }^{{\small \text{chP}}}]=0,%
\text{ \ \ \ }[\ST_{\lambda }^{{\small \text{chP}}},\ST_{\mu }^{%
{\small \text{chP}}}]=0\text{ \ \ \ }\forall \lambda ,\mu \in \mathbb{C}
\label{FrbtChP-TII.C.M.R.}
\end{equation}%
only hold under the further restrictions:%
\begin{equation}
\text{q}_{n}\equiv \text{r}_{n}\text{ \ \ }\forall n\in\{1,...,\mathsf{N}\}.
\label{FrbtComm-restrictions}
\end{equation}%
Moreover, we also have:%
\begin{equation}
\lbrack \Theta ,\ST_{\lambda }^{{\small \text{chP}}}]=0,
\label{FrbtChP-Theta-C.M.R.}
\end{equation}%
as it simply follows by computing the matrix elements on the left and right $%
\su_{n}$-eigenbasis: 
\begin{eqnarray}
\langle \text{z}|\Theta \ST_{\lambda }^{{\small \text{chP}}}|\text{z}%
^{\prime }\rangle  &=&\langle z_{1}/q,...,z_{\mathsf{N}}/q|\ST_{\lambda }^{%
{\small \text{chP}}}|\text{z}^{\prime }\rangle =\prod_{n=1}^{\mathsf{N}}W_{\text{q}%
_{n}\text{p}}(z_{n}/qz_{n}^{\prime })\bar{W}_{\text{r}_{n}\text{p}%
}(z_{n}/qz_{n+1}^{\prime })  \notag \\
&=&\langle \text{z}|\ST_{\lambda }^{{\small \text{chP}}}|qz_{1}^{\prime
},...,qz_{\mathsf{N}}^{\prime }\rangle =\langle \text{z}|\ST_{\lambda }^{%
{\small \text{chP}}}\Theta |\text{z}^{\prime }\rangle .
\end{eqnarray}
\subsection{General representations: spectrum simplicity and eigenstates characterization}

Let us denote with $\mathcal{R}_{\mathsf{N}}^{{\small \text{chP}}}$ the
subvariety of the $\tau _{2}$-representations of real dimension $4\mathsf{N}+4$
parametrized by:%
\begin{eqnarray}
\alpha _{n} &=&-b_{\text{q}_{n}}^{2}/\mathsf{c}_{0},\text{ \ \ \ }%
\mathbbm{b}_{n}=-\mathbbm{d}_{n}/q=-a_{\text{q}_{n}}d_{\text{q}_{n}}/q^{3/2},
\label{FrbtPar-on-ChP-inho-1} \\
\beta _{n} &=&-\mathsf{c}_{0}d_{\text{q}_{n}}^{2},\text{ \ \ \ \ }%
\mathbbm{c}_{n}=-\mathbbm{a}_{n}q=b_{\text{q}_{n}}c_{\text{q}_{n}}q^{1/2},
\label{FrbtPar-on-ChP-inho-2}
\end{eqnarray}%
and\footnote{Note that the dimension of $\mathcal{R}_{\mathsf{N}}^{{\small \text{chP}}}$ is defined by counting two independent complex parameters for each points
q$_{n}\in \mathcal{C}_{k}$. For example $a_{\text{q}_{n}}$ and $d_{\text{q}%
_{n}}$ can be taken as independent complex parameters while the other two $%
b_{\text{q}_{n}}$ and $c_{\text{q}_{n}}$ are fixed by the equations of the
curve $\mathcal{C}_{k}$.}:%
\begin{equation}
\mathsf{c}_{0}\in \mathbb{C},\text{ \ \ \ q}_{n}\in \mathcal{C}_{k},\text{ \
\ \ }k\in \mathbb{C}.  \label{FrbtPar-on-ChP-inho-3}
\end{equation}%
Now, for the representations $\mathcal{R}_{\mathsf{N}}^{{\small \text{chP}}}$, we
can prove:
\begin{propn}\label{FrbtChP-parametrization}All right and left eigenstates of the
transfer matrix $\ST_{\lambda }^{{\small \text{chP}}}$ are eigenstates of $%
\tau _{2}(\lambda )$ with corresponding eigenvalues related by the
functional Baxter equation:%
\begin{equation}
t(\lambda )=\frac{a_{\text{BS}}(\lambda )\mathsf{q}_{\lambda /q}^{{\small 
\text{chP}}}+d_{\text{BS}}(\lambda )\mathsf{q}_{q\lambda }^{{\small \text{chP%
}}}}{\mathsf{q}_{\lambda }^{{\small \text{chP}}}},
\label{Frbteigenvalue-T-II-ChP}
\end{equation}%
where \textsf{q}$_{\lambda }^{{\small \text{chP}}}$ denotes the $\ST_{\lambda
}^{{\small \text{chP}}}$-eigenvalue. Moreover, for almost all the
representations in $\mathcal{R}_{\mathsf{N}}^{{\small \text{chP}}}$, \textsf{q}$%
_{\lambda }^{{\small \text{chP}}}$ is non-degenerate and the simultaneous
eigenstate of ($\ST_{\lambda }^{{\small \text{chP}}},\tau _{2}(\lambda )$) is
the one associated to the eigenvalue $t(\lambda )$ as constructed in Theorem %
\ref{FrbtC:T-eigenstates}.
\end{propn}
\begin{proof}The first statement concerning right eigenstates is a simple and well known
consequence of the Baxter equation (\ref{FrbtBax-ChP-T-II}) from which we also
have that the corresponding $\tau _{2}$-eigenvalue $t(\lambda )$ is given by
(\ref{Frbteigenvalue-T-II-ChP}). To prove the statement for left
eigenstates we also have to use the first commutation relation in (\ref{FrbtChP-TII.C.M.R.}). Now, the fact that the $\tau _{2}$-spectrum is simple for
almost all the representations in $\mathcal{R}_{\mathsf{N}}^{{\small \text{chP}}} $ implies the remaining statements of the proposition, once we recall
that the $\tau _{2}$-spectrum characterization of Theorem \ref{FrbtC:T-eigenstates} holds for completely general representations and so in
particular for $\mathcal{R}_{\mathsf{N}}^{{\small \text{chP}}}$.
\end{proof}
{\bf Remark 5.} \ There exists a proper subvariety of $\mathcal{R}_{\mathsf{N}}^{{\small \text{chP}}}$ for which $\tau _{2}(\lambda )$ has double
degeneracy eigenvalues. In these representations we have shown that the degeneracy is
resolved by the charge $\Theta $ then, thanks to the commutativity \rf{FrbtChP-Theta-C.M.R.}, it is the simultaneous spectrum of $\Theta $ and $\ST_{\lambda }^{{\small \text{chP}}}$ that is simple here.

\subsection{Self-adjoint representations: complete spectrum characterization}

The Proposition \ref{FrbtChP-parametrization} characterizes completely the $\ST_{\lambda }^{{\small 
\text{chP}}}$-eigenspace showing that it is contained in the $\tau _{2}$%
-eigenspace. The reverse inclusion of eigenspaces however is not proven for
general $\mathcal{R}_{\mathsf{N}}^{{\small \text{chP}}}$. To clarify this point it
is worth to remark that while $\ST_{\lambda }^{{\small \text{chP}}}$\ and $%
\tau _{2}(\lambda )$\ are both one-parameter families of commuting operators,
their mutual commutativity is proven only when the spectral parameters
coincide, see (\ref{FrbtChP-TII.C.M.R.}).\ This means that we cannot argue,
by using this commutativity only\footnote{%
Note that the reverse statement is proven thanks to the Baxter
equation (\ref{FrbtBax-ChP-T-II}).}, that the $\tau _{2}$-eigenstates are
eigenstates of $\ST_{\lambda }^{{\small \text{chP}}}$; in fact, there may be
some residue degeneracy in the $\tau _{2}$-spectrum for fixed value of the
spectral parameter. Note that the simplicity property of $\tau _{2}$%
-spectrum only says that, taken any two different $\tau _{2}$-eigenstates, the
corresponding two $\tau _{2}$-eigenvalue functions are not identical. This
is a less strong statement that: for any $\tau _{2}$%
-eigenstate $|\,t\,\rangle $ there exists at least one value $\lambda _{t}$ of
the spectral parameter such that the $\tau _{2}$-eigenvalue in $\lambda _{t}$
is non-degenerate. If this last statement can be shown, the
commutativity (\ref{FrbtChP-TII.C.M.R.}) implies the coincidence of the $\tau
_{2}$ and $\ST_{\lambda }^{{\small \text{chP}}}$-eigenspaces.

Of course, the diagonalizability of $\ST_{\lambda }^{{\small \text{chP}}}$
defines a sufficient criterion to overcome deeper analysis of the $\tau
_{2} $-spectrum degeneracy.
Indeed, in this case the Proposition \ref{FrbtChP-parametrization} directly implies that the $\ST%
_{\lambda }^{{\small \text{chP}}}$-eigenbasis is simultaneously a $\tau _{2}$%
-eigenbasis. Then it is relevant to know the precise characterization of
the representations of $\mathcal{R}_{\mathsf{N}}^{{\small \text{chP}}}$ for which
the chiral Potts transfer matrix $\ST_{\lambda }^{{\small \text{chP}}}$ is proven to be diagonalizable. Here, we prove that this holds in particular for the
representations on the non-trivial curves for which the $\tau _{2}$-transfer
matrix is self-adjoint. Let us denote with $\mathcal{R}_{\mathsf{N}}^{{\small 
\text{ChP,s-adj}}}\equiv \mathcal{R}_{\mathsf{N}}^{{\small \text{chP}}}\cap 
\mathcal{R}_{\mathsf{N}}^{{\small \text{S-adj}}}$ such subset of $\mathcal{R}_{%
\mathsf{N}}^{{\small \text{chP}}}$, then the following
characterization holds:

\bigskip

\begin {lem} The subvariety  $\mathcal{R}_{\mathsf{N}}^{{\small \text{ChP,s-adj}}}$ is
characterized by (\ref{FrbtPar-on-ChP-inho-1})-(\ref{FrbtPar-on-ChP-inho-3}) and the
following points q$_{n}$ on the chP curves:%
\begin{equation}
\text{q}_{n}=(a_{\text{q}_{n}},\epsilon q\epsilon _{0,n}a_{\text{q}%
_{n}}^{\ast },\epsilon _{0,n}d_{\text{q}_{n}}^{\ast },d_{\text{q}_{n}})\in 
\mathcal{C}_{k},\ \ \ \epsilon _{0,n}=\pm 1,\ \ \ k^{\ast }=\epsilon k,
\label{Frbtq-S-adj}
\end{equation}%
where we have fixed $\mathsf{c}_{0}=1$, just to simplify the parametrization.
\end{lem}
\begin{proof}This subvariety is of real dimension $3\mathsf{N}+1$ and it is parametrized
by: 
\begin{equation}
\beta _{n}=-\frac{\left( \mathbbm{a}_{n}^{\ast }\right) ^{2}}{q\alpha _{n}^{\ast }},%
\text{ \ \ \ \ }\mathbbm{d}_{n}=-\epsilon \mathbbm{a}_{n}^{\ast },\text{ \ \ \ \ }\mathbbm{c}_{n}=-\mathbbm{a}_{n}q,%
\text{ \ \ \ \ }\mathbbm{b}_{n}=\epsilon \mathbbm{a}_{n}^{\ast }/q,  \label{FrbtTau_2-link-S-adj}
\end{equation}%
with:%
\begin{equation}
\alpha _{n}=-(qa_{\text{q}_{n}}^{\ast })^{2},\text{ \ \ \ \ }\mathbbm{a}_{n}=-\epsilon
q^{1/2}a_{\text{q}_{n}}^{\ast }d_{\text{q}_{n}}^{\ast },
\label{Frbt+alpha-a-cond}
\end{equation}%
where the point q$_{n}$ is characterized by (\ref{Frbtq-S-adj}). Indeed, it is a
trivial check to verify that the above parametrization satisfies both the
condition (\ref{FrbtSelf-adjointness Condition}) and those defining $\mathcal{R}%
_{\mathsf{N}}^{{\small \text{chP}}}$. The $\epsilon $-reality condition on the
modulus $k$ of the chP curves is then derived from the equation of the curve%
\footnote{%
It is worth noticing that for $\epsilon =1$ the equation (\ref{Frbtk-S-adj})
implies the following restriction on the modulus of the curve $k\in ]-1,1[$.}%
:%
\begin{equation}
k=\frac{\varphi _{x_{\text{q}_{n}}}^{p}+\epsilon \varphi _{x_{\text{q}%
_{n}}}^{-p}}{\epsilon |x_{\text{q}_{n}}|^{p}+1/|x_{\text{q}_{n}}|^{p}},
\label{Frbtk-S-adj}
\end{equation}%
where we have denoted $x_{\text{q}_{n}}=\varphi _{x_{\text{q}_{n}}}|x_{\text{%
q}_{n}}|$ with $\varphi _{x_{\text{q}_{n}}}$ the phase of $x_{\text{q}_{n}}$.
\end{proof}
\begin{lem} Let q be a point on the algebraic curve characterized by the condition
(\ref{Frbtq-S-adj}), then the W-functions \rf{FrbtW-function}
satisfy the following property under complex conjugation:%
\begin{equation}
\left( \frac{W_{\text{qp}}(z(n))}{W_{\text{qp}}(z(0))}\right) ^{\ast }=\frac{%
\bar{W}_{\text{qp}}(z(p-n))}{\bar{W}_{\text{qp}}(z(0))}\text{ \ 
}\forall z\in \mathbb{S}_{p},  \label{FrbtComplex-CJ-W}
\end{equation}%
where we have the following restriction on the point p$\in \mathcal{C}_{k}$, with $k^{\ast }=\epsilon k$:
\begin{equation}
x_{\text{p}}^{\ast }=\epsilon q^{-1}x_{\text{p}},\text{ \ }y_{\text{p}}^{\ast
}=\epsilon q y_{\text{p}},\text{ \ \ }s_{\text{p}}^{\ast }=s_{\text{p}}.
\label{FrbtChoice-p}
\end{equation}
\end{lem}
\begin{proof} It is important to point out that the restrictions on the point p are
compatible with the condition p$\in \mathcal{C}_{k}$. Now, by the
definition \rf{FrbtW-function} and the condition (\ref{Frbtq-S-adj}) satisfied by the point q and (\ref{FrbtChoice-p}), we have:%
\begin{equation}
\left( \frac{W_{\text{qp}}(z(n))}{W_{\text{qp}}(z(0))}\right) ^{\ast }=(s{}_{%
\text{p}}s{}_{\text{q}})^{-n}\prod_{k=0}^{n-1}\frac{y_{\text{p}}-q^{2k}y_{\text{q}}}{q^{-2}x_{\text{q}}-q^{2k}x_{\text{p}}}
\end{equation}%
where we have used that $s{}_{\text{q}}^{\ast }=s{}_{\text{q}}^{-1}$ and $%
s{}_{\text{p}}^{\ast }=s{}_{\text{p}}$. Once we use the
cyclicity condition satisfied by the $\bar{W}_{\text{qp}}(z)$, the r.h.s. of the above equation
coincides with the r.h.s. of (\ref{FrbtComplex-CJ-W}).
\end{proof}
\begin{propn}For any representation in $\mathcal{R}_{\mathsf{N}}^{{\small \text{%
chP,s-adj}}}$, the inhomogeneous chiral Potts transfer matrix $\ST_{\lambda
}^{{\small \text{chP}}}$ is a one-parameter family of normal operators
w.r.t. the points p defined in (\ref{FrbtChoice-p}). In particular, it satisfies the
following Hermitian conjugation property:%
\begin{equation}
\left( \left. \mathsf{T}_{\lambda }^{{\small \text{chP}}}\right\vert _{(%
\text{p},\text{q}_{n})}\right) ^{\dagger }=g_{(\text{p},\text{q}_{n})}\left. 
\mathsf{\hat{T}}_{\lambda }^{{\small \text{chP}}}\right\vert _{(\text{p}%
,\text{q}_{n})},\text{ \ \ }g_{(\text{p},\text{q}_{n})}\equiv
\prod_{n=1}^{\mathsf{N}}\frac{W_{\text{q}_{n}\text{p}}^{\ast }(z(0))\bar W_{\text{q}_{n}%
\text{p}}^{\ast }(z(0))}{W_{\text{q}_{n}\text{p}}(z(0))\bar W_{\text{q}%
_{n}\text{p}}(z(0))},  \label{FrbtCplex-CJ-Tchp}
\end{equation}%
where $\mathsf{\hat{T}}_{\lambda }^{{\small \text{chP}}}$ is the second
chiral Potts transfer matrix defined by:%
\begin{equation}
\left. \mathsf{\hat{T}}_{\lambda }^{{\small \text{chP}}}\right\vert _{(\text{%
p},\text{q}_{n})}(\text{z},\text{z}^{\prime })\equiv \langle \text{z}|\mathsf{\hat{T}}
_{\lambda }^{{\small \text{chP}}}|\text{z}^{\prime }\rangle =\prod_{n=1}^{%
\mathsf{N}}W_{\text{q}_{n}\text{p}}(z_{n+1}/z_{n}^{\prime })\bar{W}_{\text{q}_{n}%
\text{p}}(z_{n}/z_{n}^{\prime }).  \label{Frbtkernel-at}
\end{equation}
\end{propn}
\begin{proof}The proof of (\ref{FrbtCplex-CJ-Tchp}) is a simple consequence of the previous
lemma, then the known commutativity \cite{FrbtBS}:
\begin{equation}
\lbrack \left. \mathsf{T}_{\lambda }^{{\small \text{chP}}}\right\vert _{(%
\text{p},\text{q}_{n})},\left. \mathsf{\hat{T}}_{\lambda^{\prime} }^{{\small \text{chP%
}}}\right\vert _{(\text{p}^{\prime },\text{q}_{n})}]=0\text{ \ \ }\forall 
\text{p, p}^{\prime }\in \mathcal{C}_{k},
\end{equation}%
implies the statement of normality thanks to \rf{FrbtChoice-p}.
\end{proof}

Let us denote with $\Sigma _{\mathsf{T}_{\lambda }^{{\small \text{chP}}}}$
the set of all the eigenvalue functions of $\mathsf{T}_{\lambda }^{{\small 
\text{chP}}}$, then:
\begin{thm}\label{FrbtIso-ChP-TII}For any representation in $\mathcal{R}_{\mathsf{N}}^{{\small \text{%
chP,s-adj}}}$ the formula (\ref{Frbteigenvalue-T-II-ChP}) defines a one-to-one map between the eigenvalue sets $\Sigma _{\mathsf{T}_{\lambda }^{{\small 
\text{chP}}}}$\ and $\Sigma _{\tau _{2}}$. Moreover, Theorem \ref{FrbtC:T-eigenstates} allows to construct the full ($\ST_{\lambda }^{{\small 
\text{chP}}},\tau _{2}(\lambda ),\Theta $)-eigenbasis associating to any $%
t(\lambda )\in \Sigma _{\tau _{2}}$ the corresponding eigenstate.
\end{thm}
\begin{proof}The diagonalizability of $\mathsf{T}_{\lambda }^{{\small \text{chP}}}$
implies, thanks to Proposition \ref{FrbtChP-parametrization}, that $\mathsf{T}_{\lambda }^{{\small \text{%
chP}}}$ shares with $\tau _{2}(\lambda )$ and $\Theta $ a complete set of
eigenstates from which the statement about the eigenvalue functions
follows. Then, the map (\ref{Frbteigenvalue-T-II-ChP}) being
an isomorphism, we can label all the eigenstates of $\mathsf{T}_{\lambda }^{{\small 
\text{chP}}}$ in terms of the eigenvalues of $\tau _{2}(\lambda )$ and $%
\Theta $ and, by using Theorem \ref{FrbtC:T-eigenstates}, we can construct the
entire simultaneous eigenbasis.
\end{proof}

\subsection{Completeness of Bethe ansatz type equations}

Let us denote with $\mathcal{\bar{R}}_{\mathsf{N}}^{{\small \text{chP,s-adj}}}$
the subset of $\mathcal{R}_{\mathsf{N}}^{{\small \text{chP,s-adj}}}$  which contains the
representations which satisfy the constrains (\ref{FrbtS-Adj-c}).
Then the following characterization holds:

\begin{lem}In the case of $\mathsf{N}$ even, the subvariety $\mathcal{\bar{R}}_{\mathsf{N}}^{%
{\small \text{chP,s-adj}}}$ is characterized by (\ref{FrbtPar-on-ChP-inho-1})-(%
\ref{FrbtPar-on-ChP-inho-3}) and the points q$_{n}$ on the curves \rf{Frbtq-S-adj} with moreover:%
\begin{equation}
a_{\text{q}_{2n+i}}=\pm \bar{\epsilon}_{2n+i}^{1/2}\varphi _{a_{\text{q}%
_{1}}}^{(1-2\delta _{i,0})}q^{2\delta _{i,0}}|a_{\text{q}_{2n+i}}|,\text{ }%
d_{\text{q}_{2n+i}}=\pm \epsilon _{1,2n+i}^{1/2}\bar{\epsilon}%
_{2n+i}^{-1/2}\left( \varphi _{a_{\text{q}_{1}}}/q\right) ^{(2\delta
_{i,0}-1)}|d_{\text{q}_{2n+i}}|,\text{ \ }i=0,1,  \label{FrbtSpecial-S-Adj}
\end{equation}%
where $\varphi _{a_{\text{q}_{1}}}$ is the phase of $a_{\text{q}_{1}}$ and%
\begin{equation}
\bar{\epsilon}_{n}=\prod_{h=1}^{n-1}\epsilon _{2,h},\ \ \ \epsilon
_{1,n}=\pm 1,\ \ \ \epsilon _{2,n}=\pm 1.
\end{equation}%
In the case of $\mathsf{N}$ odd, $\mathcal{\bar{R}}_{\mathsf{N}}^{{\small \text{chP,s-adj%
}}}$ is obtained by the further requirements: 
\begin{equation}
\varphi _{a_{\text{q}_{1}}}^{2}=\pm q^{2},\text{ \ \ \ \ }\epsilon
_{1,n}=\epsilon \text{ \ \ \ }\forall n\in \{1,...,\mathsf{N}\}.
\end{equation}
\end{lem}
\begin{proof}
Let us impose the second and third conditions of (\ref{FrbtS-Adj-c}), these
together with (\ref{FrbtTau_2-link-S-adj}) imply:%
\begin{equation}
\left( \frac{\mathbbm{a}_{n}^{\ast }}{\mathbbm{a}_{n}}\right) ^{2}=q^{2},\text{ \ \ }\frac{%
\alpha _{n+1}^{\ast }}{\alpha _{n+1}}=\frac{\alpha _{n}}{\alpha _{n}^{\ast }}%
,  \label{FrbtCond-3}
\end{equation}%
from which the first condition in (\ref{FrbtS-Adj-c}) is automatically satisfied
for $\mathsf{N}$ even while it imposes $\alpha _{n}\in \mathbb{R}$ for $\mathsf{N}$ odd.
The conditions (\ref{FrbtCond-3}) imply for the points q$_{n}$ the constrains:%
\begin{equation}
\frac{d_{\text{q}_{n}}}{d_{\text{q}_{n}}^{\ast }}=\epsilon _{1,n}q^{2}\frac{%
a_{\text{q}_{n}}^{\ast }}{a_{\text{q}_{n}}},\text{ \ \ \ \ }q^{2}\frac{a_{%
\text{q}_{n+1}}^{\ast }}{a_{\text{q}_{n+1}}}=\epsilon _{2,n}\frac{a_{\text{q}%
_{n}}}{q^{2}a_{\text{q}_{n}}^{\ast }},
\end{equation}%
with $\epsilon _{i,n}=\pm 1$. Then, the parametrization for the points q$_{n}
$ is given by (\ref{Frbtq-S-adj}) plus (\ref{FrbtSpecial-S-Adj}) and the equations
of the curves for these points q$_{n}$ read:%
\begin{equation}
k=\left( \epsilon _{1,2n+i}\epsilon \right) ^{(1-\delta _{i,0})}\bar{\epsilon%
}_{2n+i}\left( \epsilon _{1,2n+i}\right) ^{1/2}\frac{\varphi _{a_{\text{q}%
_{1}}}^{2p}+\epsilon _{1,2n+i}\epsilon \varphi _{a_{\text{q}_{1}}}^{-2p}}{%
\epsilon |x_{\text{q}_{2n+i}}|^{p}+1/|x_{\text{q}_{2n+i}}|^{p}}.
\end{equation}%
In the case $\mathsf{N}$ even we do not have to add further constraints, and it is then
clear that the subvariety $\mathcal{\bar{R}}_{\mathsf{N}}^{{\small \text{chP,s-adj}%
}}$ can be for example parameterized by the $\mathsf{N}+2$ parameters $(k,|d_{%
\text{q}_{1}}|,|a_{\text{q}_{1}}|,|a_{\text{q}_{2}}|,....,|a_{\text{q}_{\mathsf{N}%
}}|)$. In the case $\mathsf{N}$ odd we have still to impose $\alpha _{n}\in 
\mathbb{R}$, which reads:%
\begin{equation}
\varphi _{a_{\text{q}_{1}}}^{2}=\pm q^{2}.
\end{equation}%
In this case, only the choices:%
\begin{equation}
\epsilon _{1,n}=\epsilon \text{ \ \ \ }\forall n\in \{1,...,\mathsf{N}\},
\end{equation}%
lead to non-trivial curves, $k\neq 0$:%
\begin{equation}
k=\frac{\pm 2\bar{\epsilon}_{n}\epsilon ^{1/2}}{\epsilon |x_{\text{q}%
_{n}}|^{p}+1/|x_{\text{q}_{n}}|^{p}},
\end{equation}%
and so it is clear that the subvariety $\mathcal{\bar{R}}_{\mathsf{N}}^{{\small 
\text{chP,s-adj}}}$ can be for example parameterized by the $\mathsf{N}+1$
parameters $(k,|a_{\text{q}_{1}}|,....,|a_{\text{q}_{\mathsf{N}}}|)$.
\end{proof}
The interest toward these representations is due to the complete
characterization of the $\tau _{2}$-eigenbasis by the $\epsilon $%
-self-adjoint solutions of the associated Bethe equations (\ref{Frbteq-Bethe}).

\begin{thm}
For any representation in $\mathcal{\bar{R}}_{\mathsf{N}}^{{\small \text{chP,s-adj}%
}}$, there is a one to one map between $\Sigma _{\mathsf{T}_{\lambda }^{%
{\small \text{chP}}}}$\ and the set of all the $\epsilon $-self-adjoint\
solutions of the system of Bethe ansatz type equations (\ref{Frbteq-Bethe}). Moreover, the
two simultaneous eigenstates of $\ST_{\lambda }^{{\small \text{chP}}}$ and $%
\Theta $ corresponding to the given \textsf{q}$_{\lambda }^{{\small \text{chP%
}}}\in \Sigma _{\mathsf{T}_{\lambda }^{{\small \text{chP}}}}$ are then
constructed by inserting in equation (\ref{FrbtQeigenstate-even}) the $\epsilon $%
-real polynomial $Q(\lambda )$ defined by (\ref{FrbtQ_t-definition}) in terms of
this Bethe equation solution.
\end{thm}
\begin{proof}
The isomorphism stated in the above theorem is just the composition of two
isomorphisms. First, since by definition $\mathcal{\bar{R}}_{\mathsf{N}}^{{\small 
\text{chP,s-adj}}}\subset \mathcal{R}_{\mathsf{N}}^{{\small \text{chP,s-adj}}}$,
the isomorphism defined in Theorem \ref{FrbtIso-ChP-TII} holds and (\ref{Frbteigenvalue-T-II-ChP})\
associates to any eigenvalue \textsf{q}$_{\lambda }^{{\small \text{chP}}}\in
\Sigma _{\mathsf{T}_{\lambda }^{{\small \text{chP}}}}$ one and only one $%
t(\lambda )\in \Sigma _{\tau_2}$, and
viceversa. Then, we can use the Theorem \ref{FrbtDerivation-Baxter-functional}
to uniquely associate the $\epsilon $-real polynomial $Q(\lambda )$ defined
by (\ref{FrbtQ_t-definition}) whose zeros define a $\epsilon $-self-adjoint
solution of the Bethe system of equations (\ref{Frbteq-Bethe}). The uniqueness of
the polynomial $Q(\lambda )$ is proven as in Theorem \ref{FrbtC:T-eigenstates},
indeed for any representation in $\mathcal{\bar{R}}_{\mathsf{N}}^{{\small \text{%
chP,s-adj}}}$, the coefficients of the Baxter equation (\ref{Frbttq-Baxter})
have different average values:%
\begin{equation}
\frac{\prod_{n=1}^{p}\mathtt{a}(\lambda q^{n})}{\prod_{n=1}^{p}\mathtt{d}%
(\lambda q^{n})}=(-1)^{\mathsf{N}}\prod_{n=1}^{\mathsf{N}}\frac{(1-|x_{\text{q}%
_{n}}|\Lambda )^{2}}{(1+|x_{\text{q}_{n}}|\Lambda )^{2}}\neq 1.
\end{equation}
\end{proof}

\section{Outlook}

One of our motivations to the present work is to define the required SOV setup to generalize the approach used in \cite{FrbtGMN12-SG} to the $\tau_2$-model associated to the most general cyclic representations of the 6-vertex Yang-Baxter algebra, such a generalization will be presented in \cite{FrbtGMN12-T2}. There, the obtained reconstruction of
local operators in terms of quantum separate variables and the derived
determinant formula for the scalar products of states are used to compute the form factors of local operators on the transfer matrix eigenstates and to express them as sums of determinants given by simple modifications of the scalar product ones.

It is worth mentioning that this approach is not restricted to cyclic representations of 6-vertex Yang-Baxter algebras, as it is addressed to the large class
of integrable quantum models whose spectrum (eigenvalues \& eigenstates) can
be determined by implementing Sklyanin's quantum separation of variables. Such a statement has already been verified for some key quantum integrable models associated to highest weight representations of the Yang-Baxter algebras and their generalization; these results will appear shortly in a series of papers.

\vspace*{1mm}
{\par\small
{\em Acknowledgements.} We would like to thank  J.-M. Maillet and B. McCoy for stimulating discussions on related subjects and for the interest shown in this work and J.-M. Maillet for his attentive reading of the paper. N. G. is supported by the ENS Lyon and ANR grant ANR-10-BLAN-0120-04-DIADEMS. N. G would like to thank the YITP Institute of Stony Brook for hospitality. G. N. is supported by National Science Foundation grants PHY-0969739. G. N. gratefully acknowledge the YITP Institute of Stony Brook for the opportunity to develop his research programs. G. N. would like to thank the Theoretical Physics Group of the Laboratory of Physics at ENS-Lyon for hospitality, under support of ANR-10-BLAN-0120-04-DIADEMS, which made possible this collaboration.}

\appendix
\section{Constructive proof of the existence of cyclic SOV representations}\label{FrbtSOV construction-0}
\setcounter{equation}{0}

In the following subsections we will show by recursive construction the Theorem \ref{FrbtSOVthm}.

\subsection{Recursive construction of $\SB$-eigenstates}

\label{FrbtSOV construction} Similarly to \cite{FrbtGN10}, we will construct the eigenstates $\langle \,\eta
\,|$ of $\SB(\la)\equiv \SB_{\srn}(\la)$ recursively by induction on $\mathsf{N}$.
In particular, assuming we have the $\SB_{\srm}$-eigenbasis for
any $\SRM<\mathsf{N}$, the eigenstates $\langle \,\eta \,|$ of $\SB_{\srn}(\la%
)$ can be written in the following form
\begin{equation}\label{Frbtb-eigen-rec}
\langle \,\eta \,|\,=\,\sum_{\chi _{\1}^{{}}}\sum_{\chi _{\2}^{{}}}\,K_{\srn%
}^{{}}(\,\eta \,|\,\chi _{\2}^{{}}|\chi _{\1}^{{}}\,)\,\langle \,\chi _{\2%
}^{{}}\,|\ot\langle \,\chi _{\1}^{{}}\,|\;,
\end{equation}
in terms of $\langle \,\chi _{\2}^{{}}\,|$  and $\langle \,\chi _{\1}^{{}}\,|$, respectively, the $\SB_{\srm}$-eigenstates and the $\SB_{\srn-\srm}$ -eigenstates with eigenvalues defined as in \rf{FrbtBdef} by the tuples $\chi _{\2}^{{}}=(\chi _{\2%
a}^{{}})_{a=1,\dots ,\srm}^{{}}$ and $\chi _{\1}^{{}}=(\chi _{\1%
a}^{{}})_{a=1,\dots ,\srn-\srm}^{{}}$. Computing the 
\textit{kernel} $K_{\srn}^{{}}(\,\eta \,|\,\chi _{\2}^{{}}|\chi _{\1}^{{}}\,)
$ will give the recursive construction of the states $\langle \,\eta\,|$.

\textit{SOV-representation for $\mathsf{N}=1$:\thinspace \thinspace \thinspace }%
let us construct first the SOV representation for $\mathsf{N}=1$. We introduce the
states: 
\begin{equation}
\langle \eta _{1}^{(0)} q^{h_{1}}|\equiv \sum_{k=1}^{p}\frac{q^{-k(h_{1}+1/2)}%
\prod_{r=1}^{k}(q^{r-1/2}\mathbbm{a}_{1}+q^{-(r-1/2)}\mathbbm{b}_{1})}{(\mathbbm{a}_{1}^{p}+\mathbbm{b}_{1}^{p})^{k/p}}\langle 1,k|,\text{\thinspace \thinspace
\thinspace\ }\forall h_{1}\in \{1,...,p\},  \label{FrbtB-eigen-Site-N}
\end{equation}%
where $\langle 1,k|$ are the states of the v$_{1}$-eigenbasis. $\langle
\eta _{1}^{(0)} q^{h_{1}}|$ for $h_{1}\in \{1,...,p\}$ then defines a $\SB_{1}$-eigenbasis:%
\begin{equation}
\langle \eta _{1}^{(0)} q^{h_{1}}|\SB_{1}=\eta _{1}^{(0)} q^{h_{1}}\langle \eta
_{1}^{(0)} q^{h_{1}}|\,\,\,\,\,\,\text{ \ with\thinspace \thinspace \thinspace
\thinspace\ }\left(\eta _{1}^{(0)}\right)^{p}\equiv q^{p/2}(\mathbbm{a}_{1}^{p}+\mathbbm{b}_{1}^{p}).
\end{equation}%
Subsequently, $\SA_{1}(\lambda )$\ and $\SD_{1}(\lambda )$ have the following
representation in the $\SB_{1}$-eigenbasis: 
\begin{equation}
\langle \eta _{1}^{(0)} q^{h_{1}}| \SA_{1}(\lambda )=\lambda \alpha _{1} \langle \eta _{1}^{(0)} q^{h_{1}-1}| -\beta _{1} \lambda^{-1} \langle \eta _{1}^{(0)} q^{h_{1}+1}| ,\text{ \ \ \ \ \ }\langle \eta _{1}^{(0)} q^{h_{1}}|\SD_{1}(\lambda )=\gamma _{1} \lambda^{-1}\langle \eta _{1}^{(0)} q^{h_{1}-1}| -\delta _{1}\lambda\langle \eta _{1}^{(0)} q^{h_{1}+1}|.
\end{equation}%

\begin{itemize}
\item[]\hspace{-0.8cm}{\it\ref{FrbtSOV construction}.a} \  \ {\it Reconstruction of the kernel $K_\srn(\,\eta\,|\,{\chi}_{\2}^{}|{\chi}_\1^{}\,)$ w.r.t. $\chi_{\2\,\, {a=1,\dots ,\srm-1}}$ and $\chi_{\1\,\, {a=1,\dots ,\srn-\srm-1}}$.}
\end{itemize}
The decomposition
\begin{equation}\label{Frbtb-12}
\SB_\srn(\la)=\SA_{\2 \ \srm}(\la)\SB_{\1 \ \srn-\srm}(\la)+{\SB}_{\2 \ \srm}(\la)\SD_{\1 \ \srn-\srm}(\la)
\end{equation}
implies that the matrix elements of the kernel $K_\srn(\,\eta\,|\,{\chi}_{\2}^{};{\chi}_\1^{}\,)$ satisfy the equations
\begin{equation}\begin{aligned}\label{Frbtrecrel}
\big(\SA_{\2 \ \srm}(\la)\SB_{\1 \ \srn-\srm}(\la) +{\SB}_{\2 \ \srm}(\la)& \SD_{\1 \ \srn-\srm}(\la)\big)^t\,K_\srn^{}(\,\eta\,|\,\chi_{\2}^{};\chi_\1^{}\,)
\\
&\,=\,\eta_\srn\prod_{a=1}^{\mathsf{N}-1}\left( \la/\eta_a-\eta_a/\la\right)\,K_\srn^{}(\,\eta\,|\,\chi_{\2}^{};\chi_{\1}^{}\,)\,,
\end{aligned}\end{equation}
where $\SO^t$ stays for the transpose of an operator $\SO$. Taking
\begin{equation}
\chi_{\1a}q^{h_{1}}\notin\text{\large{Z}}_{{\rm det_q}\SM_{\1, \mathsf{N}-\SRM}(\la)}, \ \ \chi _{\2b}q^{h_{2}}\notin\text{\large{Z}}_{{\rm det_q}\SM_{\2, \SRM}(\la)}\ \ \text{and} \ \ \chi _{\1a}q^{h_{1}}\neq \chi _{\2b}q^{h_{2}},
\end{equation}
where $h_{i}\in \{1,...,p\}$, $a\in \{1,...,\mathsf{N}-\SRM\}$ and $b\in \{1,...,\SRM\}$, the previous equations can be written as recursion relations for the kernel w.r.t. $\chi_{\1 a}$ and $\chi_{\2 b}$ by simply fixing $\la=\chi_{\1\, a=1,...,\srn-\srm-1}$ and $\la=\chi_{\2\, b=1,...,\srm-1}$. Indeed, for $\la=\chi_{\1 a}$ the \rf{Frbtrecrel} leads to
\begin{equation}\label{Frbtrecrel1}\begin{aligned}
 K_\srn^{}(\,\eta\,|\,\chi_{\2}^{}| {q^{-\delta_{a}}\chi_{\1}^{})\,}\;&{{\tt d}^{(SOV)}_\1(q^{-1}\chi_{\1a}^{})}\; \,\chi_{\2\,\srm}\prod_{b=1}^{\rm M-1}(\chi_{\1 a}^{}/\chi_{\2 b}^{}-\chi_{\2 b}^{}/\chi_{\1 a}^{})\, \\[-1ex]
& \,=\, {K_\srn^{}(\,\eta\,|\,\chi_{\2}^{};\chi_{\1}^{}\,)}\; \,\eta_\srn\prod_{b=1}^{\mathsf{N}-1}(\chi_{\1 a}^{}/\eta_b^{}-\eta_b^{}/\chi_{\1 a}^{})
\,,
\end{aligned}\end{equation}
while for $\la=\chi_{\2a}$ it holds
\begin{equation}\label{Frbtrecrel2}\begin{aligned}
{ K_\srn^{}(\,\eta\,|\,q^{\delta_{a}}\chi_{\2}^{}|\chi_{\1}^{}\,)}\; &{{\tt a}^{(SOV)}_\2(q^{+1}\chi_{\2a}^{})}\;
\chi_{\1\,\srn-\srm}\prod_{b=1}^{\mathsf{N}-\rm M-1}(\chi_{\2 a}^{}/\chi_{\1 b}^{}-\chi_{\1 b}^{}/\chi_{\2 a}^{} )\,
\\[-1ex] &\,=\,{K_\srn^{}(\,\eta\,|\,\chi_{\2}^{};\eta_{\1}^{}\,)}\; \,\eta_\srn\prod_{b=1}^{\mathsf{N}-1}
(\chi_{\2 a}^{}/\eta_b^{}-\eta_b^{}/\chi_{\2 a}^{})
\,.
\end{aligned}\end{equation}
Here ${{\tt d}^{(SOV)}_\1(\chi_{\1 a}^{})}$ and ${{\tt a}^{(SOV)}_\2(\chi_{\2 a}^{})}$ are the known coefficients of the SOV representations in the subchains $\1$ and $\2$ while $q^{\pm\delta_{r}}$ is defined in \rf{FrbtT_r}.
\begin{itemize}
\item[]\hspace{-0.8cm}{\it\ref{FrbtSOV construction}.b} \  \ {\it Reconstruction of the kernel $K_\srn(\,\eta\,|\,{\chi}_{\2}^{}|{\chi}_\1^{}\,)$ w.r.t. $\chi_{\2\,\,{\srm}}$ and $\chi_{\1\,\,{\srn-\srm}}$.}
\end{itemize}  
In the $\tau _{2}$-model the form of the asymptotics of the Yang-Baxter
generators w.r.t. $\lambda $ leads to some complication in the computation
of the recursions satisfied by the kernel w.r.t. the variables $\chi _{\2%
\text{ }\srm}^{{}}$ and $\chi _{\1\text{ }\srn-\srm}^{{}}$. Here, we show
that just doing a discrete Fourier transform w.r.t. these variables we get a
simpler characterization. Let us introduce the state: 
\begin{equation}
\overline{\langle{\eta}_{\srn}^{(0)} q^{h_{\srn}}},{\eta}^{(0)} _{\srn-1}q^{h_{%
\srn-1}},....,{\eta}^{(0)} _{1}q^{h_{1}}|=\sum_{k_{\srn}=1}^{p}q^{k_{\srn}h_{%
\srn}}\langle{\eta}_{\srn}^{(0)} q^{k_{\srn}},{\eta}_{\srn-1}^{(0)} q^{h_{\srn%
-1}},....,{\eta}_{1}^{(0)} q^{h_{1}}|,  \label{FrbtbarB-eigenstates}
\end{equation}%
then it is an eigenstates of the $\Theta $-charge:%
\begin{equation}
\overline{\langle{\eta}_{\srn}^{(0)} q^{h_{\srn}}},\{\eta _{i}\}|\Theta _{\srn%
}=q^{h_{\srn}}\overline{\langle \eta_{\srn}^{(0)} q^{h_{\srn}}},\{\eta _{i}\}|
\end{equation}%
and moreover it holds:%
\begin{equation}
\overline{\langle{\eta}_{\srn}^{(0)} q^{h_{\srn}}},\{\eta _{i}\}|\SB_{\srn%
}(\lambda )=\bar{\eta}_{\srn}\prod_{n=1}^{\mathsf{N}-1}\left( \frac{\lambda }{\eta
_{n}}-\frac{\eta _{n}}{\lambda }\right) \overline{\langle \eta _{\srn}q^{h_{%
\srn}+1}},\{\eta _{i}\}|.
\end{equation}%
Note that by using the recursion formula (\ref{Frbtb-eigen-rec}) and the
representation (\ref{FrbtbarB-eigenstates}) we get: 
\begin{equation}
\langle{\eta}_{\srn}^{(0)} q^{h_{\srn}},\{\eta _{i}\}|\,=\sum_{\chi _{\1%
}^{{}}}\sum_{\chi _{\2}^{{}}}\,\bar{K}_{\srn}^{{}}(\,\eta \,|\,\chi _{\2%
}^{{}}|\chi _{\1}^{{}}\,)\overline{\langle \chi _{\2,\text{ \ }\srm}},\chi _{%
\2,\text{ \ }a}^{{}}|\,\otimes \overline{\langle \chi _{\1,\text{ \ }\srn-%
\srm}},\chi _{\1,\text{ \ }a}|,  \label{FrbtbarB-eigenstate-d}
\end{equation}%
with%
\begin{align}
& \,\bar{K}_{\srn}(\eta|\bar{\chi}_{%
\2,\text{ \ }\srm}q^{h_{\2,\text{ \ }\srm}},\{\chi _{\2,\text{ \ }a}\}|\bar{%
\chi}_{\1,\text{ \ }\srn-\srm}q^{h_{\1,\text{ \ }\srn-\srm}},\{\chi _{\1,%
\text{ \ }a}\})\left. =\right.   \notag \\
& \frac{1}{p^2}\sum_{k_{\2,\text{ \ }\srm},k_{\1,\text{ \ }\srn-\srm}=1}^{p}q^{-(k_{\1,%
\text{ \ }\srn-\srm}h_{\1,\text{ \ }\srn-\srm}+k_{\2,\text{ \ }\srm}h_{\2,%
\text{ \ }\srm})}K_{\srn}(\eta|\bar{\chi}%
_{\2,\text{ \ }\srm}q^{k_{\2,\text{ \ }\srm}},\{\bar{\chi}_{\2,\text{ \ }%
a}\}|\bar{\chi}_{\1,\text{ \ }\srn-\srm}q^{k_{\1,\text{ \ }\srn-\srm%
}},\{\chi _{\1,\text{ \ }a}\}\,).
\end{align}%
To obtain the dependence of $\bar{K}_{\srn}$ w.r.t. $\chi _{\2,\text{ \ }%
\srm}$ and $\chi _{\1,\text{ \ }\srn-\srm}$, let us introduce the asymptotic operators $\SB_{\srn}^{(\mp )}$ whose action reads:%
\begin{equation}
\langle\,\eta\,|\SB_{\srn}^{(\mp )}\equiv (\mp 1)^{(\mathsf{N}-1)}\lim_{\log \lambda \rightarrow
\mp \infty }\lambda ^{\pm (\mathsf{N}-1)}\langle\,\eta\,| \SB_{\srn}(\lambda )=\eta _{\srn%
}\prod_{i=1}^{\mathsf{N}-1}\eta _{i}^{\pm 1} \langle\,\eta\,|,
\end{equation}%
while the decomposition (\ref{Frbtb-12}) implies:%
\begin{equation}
\SB_{\srn}^{(\mp )}=a_{\mp ,\2}\Theta _{\2}^{\mp }\chi _{\1,\text{ \ }\srn-%
\srm}\prod_{i=1}^{\mathsf{N}-\SRM-1}\chi _{\1,\text{\ }i}^{\pm 1}+d_{\mp ,\1}\Theta _{\1}^{\pm }\chi _{\2,\text{ }\srm}\prod_{i=1}^{\SRM-1}\chi _{\2,%
\text{\ }i}^{\pm 1}.
\end{equation}%
Acting on the generic state $\langle \eta _{\srn}^{(0)} ,\{\eta _{i}\}|$ with $\SB_{%
\srn}^{(\mp )}$ we get a system of two equations:%
\begin{equation}
\bar{K}_{\srn}^{{}}(\,\eta \,|\,\chi _{\2}^{{}}|\chi _{\1}^{{}}\,)=\text{x}%
_{\epsilon }\bar{K}_{\srn}^{{}}(\,\eta \,|\,\chi _{\2}^{{}}|q^{-\delta_{\srn-\srm}}\chi _{\1}^{{}}\,)+\text{y}_{\epsilon }\bar{K}_{\srn}^{}(\,\eta \,|\,q^{-\delta_{\srm}}\chi _{\2}^{{}}|\chi
_{\1}^{{}}\,),\qquad\epsilon =\pm 1.
\end{equation}%
Here the coefficients read:%
\begin{equation}
\text{x}_{\epsilon } \equiv\frac{a_{\epsilon ,\2}q^{\epsilon h_{%
\2,\text{ \ }\srm}}\bar{\chi}_{\1,\text{ \ }\srn-\srm}\prod_{i=1}^{\mathsf{N}-\SRM%
-1}\chi _{\1,\text{\ }i}^{-\epsilon }}{\bar{\eta}_{\srn}\prod_{i=1}^{\mathsf{N}%
-1}\eta _{i}^{-\epsilon }},\,\,\,\,\,\,\qquad 
\text{y}_{\epsilon } \equiv \frac{d_{\epsilon ,\1}q^{-\epsilon h_{%
\1,\text{ \ }\srn-\srm}}\bar{\chi}_{\2,\text{ }\srm}\prod_{i=1}^{\SRM-1}\chi
_{\2,\text{\ }i}^{-\epsilon }}{\bar{\eta}_{\srn}\prod_{i=1}^{\mathsf{N}-1}\eta
_{i}^{-\epsilon }}.
\end{equation}
Then solving this system we get our result:%
\begin{equation}
\frac{\bar{K}_{\srn}^{{}}(\,\eta \,|\,q^{-\delta_{\srm}}\chi _{\2}^{{}}|\chi _{\1}^{{}}\,)}{\bar{K}_{\srn}^{{}}(\,\eta \,|\,\chi _{\2%
}^{{}}|\chi _{\1}^{{}}\,)}=\frac{\text{x}_{-}-\text{x}_{+}}{\text{x}_{-}%
\text{y}_{+}-\text{x}_{+}\text{y}_{-}},\qquad\text{ \ \ }\frac{\bar{K}_{\srn}^{{}}(\,\eta \,|\,\chi _{\2}^{{}}|q^{-\delta_{\srn-
\srm}}\chi _{\1}^{{}}\,)%
}{\bar{K}_{\srn}^{{}}(\,\eta \,|\,\chi _{\2}^{{}}|\chi _{\1}^{{}}\,)}=\frac{%
\text{y}_{+}-\text{y}_{-}}{\text{x}_{-}\text{y}_{+}-\text{x}_{+}\text{y}_{-}}%
.  \label{Frbtrecrelzero1}
\end{equation}

\begin{itemize}
\item[]\hspace{-0.8cm}{\it\ref{FrbtSOV construction}.c} \  \  {\it Reconstruction of the kernel $K_\srn(\,\eta\,|\,{\chi}_{\2}^{};{\chi}_\1^{}\,)$ w.r.t. $\eta$.}
\end{itemize}
The formula
\begin{equation}\label{FrbtD-exp}
\SD_\srn(\la)=\SD_{\2 \ \srm}(\la)\SD_{\1 \ \srn-\srm}(\la)+{\SC}_{\2 \ \srm}(\la)\SB_{\1 \ \srn-\srm}(\la)
\end{equation}
implies the identity
\begin{equation}
\langle \eta |\SD_\srn(\eta _{i})=-{\rm det_q}\SM_{\2 \ \srm}(q\eta _{i})\langle \eta
|{\SB}_{\2 \ \srm}^{-1}(q\eta _{i}){\SB}_{\1 \ \srn-\srm}(\eta _{i}),
\end{equation}
writing ${\SC}_{\2 \ \srm}(\la)$ by the quantum determinant ${\rm det_q}\SM_{\2 \ \srm}(\la)$ in the subchain $\2$ and using
\begin{equation}
\langle \eta |{\SB}_{\1 \ \srn-\srm}(\eta _{i})=-\langle \eta |{\SD}_{\1 \ \srn-\srm}(\eta _{i}){\SA}_{\2 \ \srm}^{-1}(q\eta _{i}){\SB}_{\2 \ \srm}(q\eta _{i}).
\end{equation}
Then, the following recursion relations is obtained for kernel $K_\srn(\,\eta\,|\,{\chi}_{\2}^{};{\chi}_\1^{}\,)$ w.r.t. $\eta_{i=1,...,\mathsf{N}-1}$:
\begin{equation}\label{Frbtrecrel+1}
\begin{aligned}
 \frac{{ K_\srn^{}(\,q^{\delta_i}\eta\,|\,\chi_{\2}^{};\chi_{\1}^{}\,)}}{K_\srn^{}(\,\eta\,|\,\chi_{\2}^{};\chi_{\1}^{}\,)}}=-\frac{{\rm det_q}\SM_{\2 \ \srm}(q\eta _{i})}{d_\srn(\eta_{i}^{})}{\frac{\chi_{\1 \srn-\srm}\prod_{a=1}^{\mathsf{N}-\SRM-1}\left( \eta _{i}/\chi_{\1 a}-\chi_{\1 a}/\eta _{i}\right) }{\chi_{\2 \srm}\prod_{a=1}^{\SRM-1}\left(q\eta _{i}/\chi _{\2 a}-\chi _{\2 a}/q\eta _{i}\right) }
\,.
\end{aligned}
\end{equation}
Finally, \rf{FrbtD-exp} leads to the following recursion relation which fix the dependence of the kernel w.r.t. $\eta_{\srn}$:
\begin{equation}
\frac{{K_{\srn}^{{}}(\,q^{\delta_{\srn}}\eta \,|\,\chi _{\2}^{{}};\chi _{\1%
}^{{}}\,)}}{{K_{\srn}^{{}}(\,\eta
\,|\,q^{\delta_{\srm}}\chi _{\2}^{{}};q^{\delta_{\srn-\srm}}\chi _{\1}^{{}}\,)}}{\;}\,=\,\;\frac{\chi _{\2\ \SD%
}^{{(+)}}\chi _{\1\ \SD}^{{(+)}}\prod_{b=1}^{\SRM-1}\chi _{\2\ b}^{{}}\prod_{b=1}^{%
\mathsf{N}-\SRM-1}\chi _{\1\ b}^{{}}}{{\eta _{\SD}^{{(+)}}}\,\prod_{a=1}^{\mathrm{\mathsf{N}-1}%
}\eta _{a}^{{}}}\,\,.  \label{Frbtrecrel+zero}
\end{equation}%

\subsection{Gauge-invariant SOV dates: $Z_r,\,Z^{(\pm)}_{\SA},\,Z^{(\pm)}_{\SD},\,\CA^{}(Z_r),\,\CD^{}(Z_r)$}\label{Frbtgauge-invariant-dates}
The recursion relations \rf{Frbtrecrel1}-\rf{Frbtrecrel2} and \rf{Frbtrecrelzero1} and the requirement of cyclicity give a system of $\mathsf{N}$ algebraic equations in the $\mathsf{N}$ unknowns
$Z_a\equiv\eta_a^p$:
\begin{equation}\label{FrbtExrecrel1}
\CD_{\1 \,\mathsf{N}-\SRM}^{}(\chi_{\1a}^{p})\; \chi_{\2\srm}^{p}
\prod_{b=1}^{{\rm M}-1}
 (\chi_{\1 a}^{p}/\chi_{\2 b}^{p}-\chi_{\2 b}^{p}/\chi_{\1 a}^{p})\,
\,=\, \eta_\srn^{p}\prod_{b=1}^{\mathsf{N}-1}(\chi_{\1 a}^{p}/\eta_b^{p}-\eta_b^{p}/\chi_{\1 a}^{p})\,,
\end{equation}
\begin{equation}\label{FrbtExrecrel2}
\CA_{\2 \,\SRM}^{}(\chi_{\2 a}^p)\;\chi_{\1\srn-\srm}^{p}
 \prod_{b=1}^{\mathsf{N}-\rm M}(\chi_{\2 a}^{p}/\chi_{\1 b}^{p}-\chi_{\1 b}^{p}/\chi_{\2 a}^{p} )\,=\,
 \eta_\srn^{p}\prod_{b=1}^{\mathsf{N}-1}(\chi_{\2 a}^{p}/\eta_b^{p}-\eta_b^{p}/\chi_{\2 a}^{p})\,,\,\,\,\,\,\,\,\,\,\,\,\,\,\,\,\,
\end{equation}
and
\begin{equation}
Z_{\mathsf{N}}\prod_{n=1}^{\mathsf{N}-1}Z_{n}^{\pm 1}=(\pm1)^{\srn-1}a_{\2,\text{ }\pm
}^{p}Z_{\1\text{ \ }\srn-\srm}\prod_{a=1}^{\mathsf{N}-\SRM-1}Z_{\1\text{ \ }%
a}^{\pm 1}+(\pm1)^{\srn-1}d_{\1,\text{ }\pm }^{p}Z_{\2\text{ \ }\srm%
}\prod_{a=1}^{\SRM-1}Z_{\2\text{ \ }a}^{\pm 1}
\end{equation}%
where we have used \rf{FrbtZAD-asymp} in the subchains $\1$ and $\2$ to express $Z_{\2,\text{ }\SA}^{(\pm )}\prod_{a=1}^{\SRM-1}Z_{\2\text{ \ }a}^{\pm 1}$
and $Z_{\1,\text{ }\SD}^{(\pm )}\prod_{a=1}^{\mathsf{N}-\SRM-1}Z_{\1\text{ \ }%
a}^{\pm 1}$.
The above system of equations completely determines the $Z_a$ in terms of $\chi_{\2 a}^p$, $\chi_{\1 a}^p$, as we can rewrite it in terms of the following Laurent polynomial equation:
\begin{equation}\label{FrbtCB'}
\mathcal{A}_{\SRM}(\Lambda )\mathcal{B}%
_{\mathsf{N}-\SRM}(\Lambda )+\mathcal{B}_{\SRM}(\Lambda )\mathcal{D}_{\mathsf{N}-\SRM}(\Lambda )=
Z_{{\mathsf{N}}}\prod_{a=1}^{{[}\mathsf{N}{]}}(\Lambda/Z_a-Z_a/\Lambda),
\end{equation}
as a consequence of the simplicity of the spectrum of $\SB(\la)$ in both the subchains $\1$ and $\2$. Note that the l.h.s. of the above equation is formed out of the known average values of the monodromy matrix elements in the subchains $\1$ and $\2$. So the $Z_a$ for $a\in \{1,...,\mathsf{N}-1\}$ are fixed determining the zeros of the known Laurent polynomial at the l.h.s. of \rf{FrbtCB'}. The cyclicity allows to fix the remaining gauge-invariant SOV dates:
\begin{equation}\label{FrbtExrecrel+1}
\CD_{\mathsf{N}}^{}(Z_{i})=\, -{\rm det_q}\CM_{\SRM}(Z_{i})\frac{\CB_{\mathsf{N}-\SRM}(Z_{i})}{\CB_{\SRM}(Z_{i})} , \ \ 
\CA_{\mathsf{N}}^{}(Z_{i})=\, -{\rm det_q}\CM_{\mathsf{N}-\SRM}(Z_{i})\frac{\CB_{\SRM}(Z_{i})}{\CB_{\mathsf{N}-\SRM}(Z_{i})},
\end{equation}
for any $i\in\{1,...,\mathsf{N}-1\}$ where:
\begin{equation}
{\rm det_q}\CM_{X}(\Lambda)\equiv\prod_{a=1}^{p}{\rm det_q}\SM_{X}(q^{a}\la), \ \ X={\SRM,\mathsf{N}-\SRM, \mathsf{N}}.
\end{equation}
where it holds:
\begin{equation}
{\rm det_q}\CM_{X}(\Lambda)=\CA_{X}(\Lambda)\CD_{X}(\Lambda)-\CB_{X}(\Lambda)\CC_{X}(\Lambda), \ \ X={\SRM,\mathsf{N}-\SRM}.
\end{equation}
\subsection{$\SB$-spectrum completeness and simplicity}\label{Frbtnondegapp}
Let us prove that the previously constructed set of $\SB$-eigenstates $\langle\,\eta\,|$ is complete by showing that there are $p^\mathsf{N}$ distinct $\SB$-eigenvalues.
\begin{propn}
\label{FrbtB-simplicity}The SOV dates $Z_{r}$ with $r\in \{1,...,\mathsf{N}-1\}$ are all distinct for almost all the values of the parameters
$\alpha^p_n,\,\beta^p_n,\,\mathbbm{a}^p_n,\,\mathbbm{b}^p_n,\,\mathbbm{c}^p_n,\,\mathbbm{d}^p_n$ of the $\tau_2$-model.
\end{propn}
\begin{proof}From the functional dependence\footnote{See Remark 6.} of $Z_1^{},\dots,Z_{\mathsf{N}-1}^{}$ w.r.t. these parameters we just have to prove that the Jacobian:
\begin{equation}\label{FrbtJ-nonzero}
J(\alpha^p,\,\beta^p,\,\mathbbm{a}^p,\,\mathbbm{b}^p,\,\mathbbm{c}^p,\,\mathbbm{d}^p)\,\equiv\,
{\rm det}\left(\frac{\pa Z_r}{\pa \alpha^p_s}\right)_{r,s=1,\dots,\mathsf{N}-1}\neq\,0\,
\end{equation}
does not vanish for some special values to derive $J\neq 0$ for almost all the values of the parameters. In fact, for $J\neq 0$, the map $Z=Z(\alpha^p_1,\dots,\alpha^p_{\mathsf{N}-1})$ is invertible from which the claim of the proposition follows.

To show that \rf{FrbtJ-nonzero}, let us consider the representations which satisfy the following constrains:
\begin{equation}
\mathbbm{b}_{n}^{p}+\mathbbm{a}_{n}^{p}=0,\text{ \ }\mathbbm{c}_{n}^{p}+\mathbbm{d}_{n}^{p}=0\text{ \ \ }\forall
n\in \{1,...,\mathsf{N}-1\},  \label{FrbtSuperCH-P}
\end{equation}%
while we leave the representation in the site $\mathsf{N}$ unrestricted. Then the averages \rf{FrbtAverage-L} simplify to:
\begin{equation}
\mathcal{L}_{n}(\Lambda )\equiv \left( 
\begin{array}{cc}
\Lambda \alpha _{n}^{p}-\beta _{n}^{p}/\Lambda  & 0 \\ 
0 & \gamma _{n}^{p}/\Lambda -\Lambda \delta _{n}^{p}%
\end{array}%
\right) \quad \quad \text{ \ \ }\forall n\in \{1,...,\mathsf{N}-1\}
\end{equation}
Then by \rf{FrbtRRel1a} we can compute the l.h.s. of \rf{FrbtCB'} and we obtain:
\begin{equation}
q^{p/2}(\mathbbm{a}_{\mathsf{N}}^{p}+\mathbbm{b}_{\mathsf{N}}^{p})\prod_{n=1}^{\mathsf{N}-1}\frac{\mathbbm{a}_{n}^{p}\mathbbm{d}_{n}^{p}}{%
\alpha _{n}^{p/2}\beta {}_{n}^{p/2}}(\frac{\Lambda \alpha _{n}^{p/2}}{\beta
{}_{n}^{p/2}}-\frac{\beta {}_{n}^{p/2}}{\Lambda \alpha _{n}^{p/2}})=Z_{{\mathsf{N}}%
}\prod_{a=1}^{\mathsf{N}-1}(\frac{\Lambda }{Z_{a}}-\frac{Z_{a}}{\Lambda }),
\label{FrbtCBsimp}
\end{equation}
and then $J\neq 0$ trivially follows. We can then define $\eta_a^{(0)}$ as a $p$-root of $Z_a$. They satisfy $\left( \eta^{(0)}_a\right) ^p\neq \left( \eta^{(0)}_b\right) ^p$ for any $a\neq b \in \{1,\dots,\mathsf{N}-1\}$.
\end{proof}
Finally, keeping in any subchain the representation of the $\SB$-spectrum simple, the Theorem \ref{FrbtSOVthm} follows by induction.

\section{Baxter $\SQ$-operator construction for general $\tau_2$-model}\label{FrbtBaxter-Q}

\setcounter{equation}{0}
In \cite{FrbtBa04},  Baxter has extended the construction of the $\SQ$-operator by gauge
transformations of \cite{FrbtBS} to the $\tau_2$-model for general cyclic representations not
restricted to those parameterized by points on chP algebraic curves. The main tool there was a
generalization of the discrete dilogarithm functions. In particular, he has remarked that asking
the cyclicity only for the products of couples of these functions does not prevent the $\SQ$-operator from being well defined for general cyclic representations of the model. Here, we reproduce and
adapt Baxter construction in our notations and we point out some interesting connection with the SOV construction.
\subsection{Dilogarithm functions on the algebraic curves $\mathcal{\ C}_{k}$}
\subsubsection{Dilogarithm functions}
The dilogarithm functions have been previously introduced
and analyzed in \cite{FrbtFK2,FrbtBT06}. Here we use the notation: 
\begin{equation}\label{FrbtW-function}
\frac{W_{\text{qp}}(z(n))}{W_{\text{qp}}(z(0))}=(\frac{s_{\text{q}}}{s_{%
\text{p}}})^{n}\prod_{k=1}^{n}\frac{y_{\text{p}}-q^{-2k}x_{\text{q}}}{y_{%
\text{q}}-q^{-2k}x_{\text{p}}},\text{ \ \ \ \ }\frac{\bar{W}_{\text{qp}%
}(z(n))}{\bar{W}_{\text{qp}}(z(0))}=(s_{\text{p}}s_{\text{q}%
})^{n}\prod_{k=1}^{n}\frac{q^{-2}x_{\text{q}}-q^{-2k}x_{\text{p}}}{y_{\text{p%
}}-q^{-2k}y_{\text{q}}},
\end{equation}%
where $z(n)=q^{-2n}\text{,\ }\forall n\in \{0,...,2l\}$. They are solutions
of the following recursion relations:%
\begin{equation}
\frac{W_{\text{qp}}(zq)}{W_{\text{qp}}(zq^{-1})}=-z\frac{s_{\text{p}}}{s_{%
\text{q}}}\frac{x_{\text{p}}}{y_{\text{p}}}q^{-1}\frac{1-\frac{y_{\text{q}}}{%
x_{\text{p}}}qz^{-1}}{1-\frac{x_{\text{q}}}{y_{\text{p}}}q^{-1}z},\text{ \ \ 
}\frac{\bar{W}_{\text{qp}}(zq)}{\bar{W}_{\text{qp}}(zq^{-1})}=-\frac{qz^{-1}%
}{s_{\text{p}}s_{\text{q}}}\frac{y_{\text{p}}}{x_{\text{p}}}\frac{1-\frac{y_{%
\text{q}}}{y_{\text{p}}}q^{-1}z}{1-\frac{x_{\text{q}}}{x_{\text{p}}}%
q^{-1}z^{-1}}.  \label{FrbtRecursion-w-1}
\end{equation}%
If the points p and q belong to the curves $\mathcal{C}_{k}$, they are
well defined functions of $z\in \mathbb{S}_{p}$ which satisfy the cyclicity
condition:%
\begin{equation}
\frac{\bar{W}_{\text{qp}}(z(p))}{\bar{W}_{\text{qp}}(z(0))}=1,\text{ \ \ \ \ 
}\frac{W_{\text{qp}}(z(p))}{W_{\text{qp}}(z(0))}=1.
\end{equation}
\subsubsection{Recursion relations}
It is worth defining the following three discrete automorphisms of the $%
\mathcal{C}_{k}$-curve: 
\begin{align}
i)\text{ \ \ \ \ \ \ \ }\Xi & :\text{p}=(a_{\text{p}},b_{\text{p}},c_{\text{p%
}},d_{\text{p}})\in \mathcal{C}_{k}\,\rightarrow \Xi (\text{p})=(qa_{%
\text{p}},qb_{\text{p}},c_{\text{p}},d_{\text{p}})\in \mathcal{C}_{k},
\\
ii)\text{ \ \ \ \ \ \ \ }\Delta & :\text{p}=(a_{\text{p}},b_{\text{p}},c_{%
\text{p}},d_{\text{p}})\in \mathcal{C}_{k}\rightarrow \Delta (\text{p}%
)=(qb_{\text{p}},a_{\text{p}},qd_{\text{p}},c_{\text{p}})\in \mathcal{C}%
_{k}, \\
iii)\text{ \ \ \ \ \ \ \ }\Upsilon & :\text{p}=(a_{\text{p}},b_{\text{p}},c_{%
\text{p}},d_{\text{p}})\in \mathcal{C}_{k}\rightarrow \Upsilon (\text{p%
})=(-q^{2}d_{\text{p}},c_{\text{p}},b_{\text{p}},-a_{\text{p}})\in \mathcal{C%
}_{k}.
\end{align}
If the points p and q belong to the curves $\mathcal{C}_{k}$, the $W$%
-functions satisfy the following recursion w.r.t. the action of $\Xi $ on
the point p: 
\begin{equation}
\frac{W_{\text{qp}}(qz)}{W_{\text{q}\Xi (\text{p})}(z)}=z^{1/2}\frac{W_{\text{qp}}(z(l))}{W_{\text{qp}}(z(0))}\frac{1-\frac{x_{\text{q}}}{y_{\text{p}}}q^{-1}}{1-\frac{x_{%
\text{q}}}{y_{\text{p}}}q^{-1}z},\qquad\frac{W_{\text{qp}}(qz)}{W_{\text{%
q}\Xi ^{-1}(\text{p})}(z)}=z^{1/2}\frac{W_{\text{qp}}(z(l))}{W_{\text{qp}}(z(0))}\frac{1-\frac{y_{\text{q%
}}}{x_{\text{p}}}qz^{-1}}{1-\frac{y_{\text{q}}}{x_{\text{p}}}q},
\label{FrbtRecursion-w-2}
\end{equation}%
and%
\begin{equation}
\frac{\bar{W}_{\text{qp}}(qz)}{\bar{W}_{\text{q}\Xi (\text{p})}(z)}=z^{-1/2}\frac{\bar{W}_{\text{qp}}(z(l))}{\bar{W}_{\text{qp}}(z(0))}\frac{1-\frac{y_{\text{q}}}{y_{\text{p}}}q^{-1}z}{1-%
\frac{y_{\text{q}}}{y_{\text{p}}}q^{-1}},\qquad\frac{\bar{W}_{\text{qp}%
}(qz)}{\bar{W}_{\text{q}\Xi ^{-1}(\text{p})}(z)}=z^{-1/2}\frac{\bar{W}_{\text{qp}}(z(l))}{\bar{W}_{\text{qp}}(z(0))}\frac{1-\frac{x_{\text{q}}}{x_{\text{p}}}q}{1-\frac{x_{\text{q%
}}}{x_{\text{p}}}qz^{-1}}.  \label{FrbtRecursion-w-3}
\end{equation}
\subsection{Parametrization of $\tau_2$-model by points in $\mathbb{C}^{3}$}
Let us implement the following gauge transformation on the $\protect\tau _{2}$-Lax operator:%
\begin{equation}
\tilde{\SL}_{n}(\lambda )\equiv \left( 
\begin{array}{cc}
1 & 0 \\ 
1/r_{n+1}z_{n+1}^{\prime } & 1%
\end{array}%
\right) \SL_{n}(\lambda )\left( 
\begin{array}{cc}
1 & 0 \\ 
-1/r_{n}z_{n}^{\prime } & 1%
\end{array}%
\right) ,
\end{equation}%
where $r_{n}\in \mathbb{C}$ and%
\begin{equation}
z_{n}^{\prime }\in \mathbb{S}_{p}\equiv \{q^{2h},\text{ }h=1,..,p\}.
\end{equation}

We associate to the tuple $(a_{\text{v}},b_{\text{v}},c_{\text{v}},d_{\text{v%
}})\in \mathbb{C}^{4}$ the following point: 
\begin{equation}
\text{v}\equiv (x_{\text{v}}\equiv a_{\text{v}}/d_{\text{v}},y_{\text{v}%
}\equiv b_{\text{v}}/c_{\text{v}},s_{\text{v}}\equiv d_{\text{v}}/c_{\text{v}%
})\in \mathbb{C}^{3},
\end{equation}%
and we define:%
\begin{equation}
t_{\text{v}}\equiv x_{\text{v}}y_{\text{v}}.
\end{equation}%
Now, let us consider the point $\overline{\text{p}}\equiv (x_{\overline{\text{p}}},y_{\overline{\text{p}}},s_{\overline{\text{p}}})\in \mathbb{C}^{3}$ and from it the points:
\begin{equation}
\overline{\text{p}}_{n}\equiv (x_{\overline{\text{p}}}/\sigma _{n},y_{\overline{\text{p}}}\sigma
_{n},s_{\overline{\text{p}}}\sigma _{n})\in \mathbb{C}^{3}\,,\,\,\,\,\sigma _{n}\in \mathbb{C%
},\quad\forall n\in \{1,...,\mathsf{N}\},  \label{Frbtcyclic-parameters}
\end{equation}
then it holds:
\begin{equation}
t_{\text{$\overline{\text{p}}$}_{n}}=t_{\overline{\text{p}}}.
\end{equation}%
We can introduce the following parametrization:%
\begin{equation}
\begin{array}{ll}
\lambda \alpha _{n}\equiv -t_{\overline{\text{p}}}^{-1/2}b_{\text{q}_{n}}b_{\text{r}%
_{n}},\text{ \ \ \ \ \ \ } & \mathbbm{b}_{n}/r_{n}q^{1/2}\equiv (x_{\text{$\overline{%
\text{p}}$}_{n}}/y_{\text{$\overline{\text{p}}$}_{n}})^{1/2}a_{\text{q}%
_{n}}d_{\text{r}_{n}}/q^{2}, \\ 
\beta _{n}/\lambda \equiv -t_{\overline{\text{p}}}^{1/2}d_{\text{q}_{n}}d_{\text{r}%
_{n}},\text{ \ \ \ \ \ \ } & q^{1/2}\mathbbm{a}_{n}/r_{n}\equiv -(x_{\text{$\overline{%
\text{p}}$}_{n}}/y_{\text{$\overline{\text{p}}$}_{n}})^{1/2}c_{\text{q}%
_{n}}b_{\text{r}_{n}}, \\ 
\delta _{n}\lambda \equiv q^{-2}t_{\overline{\text{p}}}^{-1/2}a_{\text{q}_{n}}a_{\text{r%
}_{n}},\text{ \ \ \ \ \ \ } & q^{1/2}\mathbbm{d}_{n}r_{n+1}\equiv -(y_{\text{$%
\overline{\text{p}}$}_{n+1}}/x_{\text{$\overline{\text{p}}$}_{n+1}})^{1/2}d_{%
\text{q}_{n}}a_{\text{r}_{n}}, \\ 
\gamma _{n}/\lambda \equiv t_{\overline{\text{p}}}^{1/2}c_{\text{q}_{n}}c_{\text{r}%
_{n}},\text{ \ \ \ \ \ \ } & q^{-1/2}\mathbbm{c}_{n}r_{n+1}\equiv (y_{\text{$\overline{%
\text{p}}$}_{n+1}}/x_{\text{$\overline{\text{p}}$}_{n+1}})^{1/2}b_{\text{q}%
_{n}}c_{\text{r}_{n}},%
\end{array}
\label{FrbtPoints on the curve}
\end{equation}%
of the gauge transformed Lax parameters in terms of the points $\overline{%
\text{p}}_{n}$ and the set of $2\mathsf{N}$ points q$_{n}$ and r$_{n}$. Note that we
can identify:%
\begin{equation}\label{Frbtgauge-r_n}
r_{n}\equiv (y_{\text{$\overline{\text{p}}$}_{n}}/x_{\text{$\overline{\text{p%
}}$}_{n}})^{1/2}=\sigma _{n}(y_{\overline{\text{p}}}/x_{\overline{\text{p}}})^{1/2},
\end{equation}%
so that the off-diagonal elements of the $\tau _{2}$-Lax operator are
independent from the points $\overline{\text{p}}_{n}$ and the parameter $t_{\overline{\text{p}}}$ is related to our spectral parameter by:%
\begin{equation}
t_{\overline{\text{p}}}\propto \lambda ^{-2}.
\end{equation}%
It is a simple exercise to derive from the above parametrization (\ref{FrbtPoints on the curve}) the relations:%
\begin{eqnarray}
\frac{x_{\text{q}_{n}}}{y_{\text{$\overline{\text{p}}$}_{n}}} &=&-q^{3/2}%
\frac{\lambda \mathbbm{b}_{n}}{\beta _{n}r_{n}},\text{ \ \ }\frac{x_{\text{r}_{n}}}{x_{%
\text{$\overline{\text{p}}$}_{n+1}}}=q^{1/2}\frac{\lambda \mathbbm{d}_{n}r_{n+1}}{%
\beta _{n}},  \label{Frbteq1} \\
\frac{y_{\text{q}_{n}}}{x_{\text{$\overline{\text{p}}$}_{n}}} &=&q^{-1/2}%
\frac{r_{n}\lambda \alpha _{n}}{\mathbbm{a}_{n}},\text{ \ \ \ }\frac{y_{\text{r}_{n}}}{%
y_{\text{$\overline{\text{p}}$}_{n+1}}}=-q^{1/2}\frac{\lambda \alpha _{n}}{%
r_{n+1}\mathbbm{c}_{n}},\text{ \ \ }s_{\text{q}_{n}}s_{\text{r}_{n}}=-\frac{\alpha
_{n}\beta _{n}}{\mathbbm{c}_{n}\mathbbm{a}_{n}}.  \label{Frbteq2}
\end{eqnarray}

\subsection{Generalized dilogarithm functions}\label{FrbtGeneralized dilogarithm functions}
Let us fix $2\mathsf{N}+1$ points $(\overline{\text{p}},$q$_{1},...,$q$_{\mathsf{N}}$,r$_{1},...,$r$_{\mathsf{N}})\in (%
\mathbb{C}^{3})^{(2\mathsf{N}+1)}$, then we can define the following $\mathsf{N}$ functions:%
\begin{equation}
\SY_{\lambda }^{(n)}(z_{n}|z_{n}^{\prime },z_{n+1}^{\prime })\equiv \mathsf{N}_{\text{q}%
_{n}\text{$\overline{\text{p}}$}_{n}}(z_{n}^{\prime })\bar{\mathsf{N}}_{\text{r}_{n}%
\text{$\overline{\text{p}}$}_{n+1}}(z_{n+1}^{\prime })W_{\text{q}_{n}\text{$%
\overline{\text{p}}$}_{n}}(z_{n}/z_{n}^{\prime })\bar{W}_{\text{r}_{n}\text{$%
\overline{\text{p}}$}_{n+1}}(z_{n}/z_{n+1}^{\prime }),
\end{equation}%
of the three variables $(z_{n},z_{n}^{\prime },z_{n+1}^{\prime })\in $ $%
\mathbb{S}_{p}^{3}$, where:%
\begin{equation}
\mathsf{N}_{\text{q}_{n}\text{$\overline{\text{p}}$}_{n}}(z_{n}^{\prime })\equiv
\left( \frac{s_{\overline{\text{p}}}^{p}(\sigma _{n}^{p}y_{\text{q}_{n}}^{p}-x_{\overline{\text{p}}}^{p})}{s_{\text{q}_{n}}^{p}(\sigma _{n}^{p}y_{\overline{\text{p}}}^{p}-x_{\text{q}_{n}}^{p})}\right) ^{h_{n}^{\prime }/p},\text{ \ \ }%
\bar{\mathsf{N}}_{\text{r}_{n}\text{$\overline{\text{p}}$}_{n+1}}(z_{n+1}^{\prime
})\equiv \left( \frac{\sigma _{n+1}^{p}y_{\overline{\text{p}}}^{p}-y_{\text{r}%
_{n}}^{p}}{s_{\overline{\text{p}}}^{p}s_{\text{r}_{n}}^{p}(\sigma _{n+1}^{p}x_{\text{r}_{n}}^{p}-x_{\overline{\text{p}}}^{p})}\right) ^{h_{n+1}^{\prime }/p}
\label{FrbtCyclic-normalizations}
\end{equation}%
and $z_{n}^{\prime }\equiv q^{2h_{n}^{\prime }}\in $ $\mathbb{S}_{p}$, while
the points $(\overline{\text{p}}_{1},...,\overline{\text{p}}_{\mathsf{N}})\in (%
\mathbb{C}^{3})^{\mathsf{N}}$ are defined by \rf{Frbtcyclic-parameters} where the $%
(\sigma _{1},...,\sigma _{\mathsf{N}})\in \mathbb{C}^{\mathsf{N}}$ are fixed by the
conditions: 
\begin{equation}
s_{\text{q}_{n}}^{p}s_{\text{r}_{n}}^{p}\frac{{\sigma}_{n+1}^{p}x_{\text{r}_{n}}^{p}-x_{\overline{\text{p}}}^{p}}{{\sigma}_{n+1}^{p}y_{\overline{\text{p}}}^{p}-y_{\text{r}_{n}}^{p}}\frac{{\sigma}_{n}^{p}y_{\overline{\text{p}}}^{p}-x_{\text{q}_{n}}^{p}}{{\sigma}_{n}^{p}y_{\text{q}_{n}}^{p}-x_{\overline{\text{p}}}^{p}}=1\quad\quad\text{ \ \ }\forall n\in \{1,...,\mathsf{N}\}.
\label{Frbtcyclicity}
\end{equation}%
Note that the $\SY_{\lambda }^{(n)}(z_{n}|z_{n}^{\prime },z_{n+1}^{\prime })$
are well defined functions of the three variables $(z_{n},z_{n}^{\prime
},z_{n+1}^{\prime })\in $ $\mathbb{S}_{p}^{3}$. Indeed, they satisfy the
cyclicity conditions w.r.t. $z_{n}^{\prime }$ and $z_{n+1}^{\prime }$ thanks
to the normalization functions defined in \rf{FrbtCyclic-normalizations}
while the cyclicity w.r.t. $z_{n}$\ is given by the conditions \rf{Frbtcyclicity}.

\subsubsection{Recursion relations}

The generalized dilogarithm functions satisfy the recursions:%
\begin{equation}
\frac{\SY_{\lambda }^{(n)}(z_{n}q^{-1}|z_{n}^{\prime },z_{n+1}^{\prime })}{%
\SY_{\lambda }^{(n)}(z_{n}q|z_{n}^{\prime },z_{n+1}^{\prime })}=s_{\text{q}%
_{n}}s_{\text{r}_{n}}\frac{\sigma _{n}z_{n}^{\prime }}{\sigma
_{n+1}z_{n+1}^{\prime }}\frac{\left( 1-\frac{x_{\text{q}_{n}}}{y_{\text{$%
\overline{\text{p}}$}_{n}}}q^{-1}\frac{z_{n}}{z_{n}^{\prime }}\right) }{%
\left( 1-\frac{y_{\text{q}_{n}}}{x_{\text{$\overline{\text{p}}$}_{n}}}q\frac{%
z_{n}^{\prime }}{z_{n}}\right) }\frac{\left( 1-\frac{x_{\text{r}_{n}}}{x_{%
\text{$\overline{\text{p}}$}_{n+1}}}q^{-1}\frac{z_{n+1}^{\prime }}{z_{n}}%
\right) }{\left( 1-\frac{y_{\text{r}_{n}}}{y_{\text{$\overline{\text{p}}$}%
_{n+1}}}q^{-1}\frac{z_{n}}{z_{n+1}^{\prime }}\right) },
\label{FrbtRecursion0Y-1}
\end{equation}%
as it simply follows from the recursions \rf{FrbtRecursion-w-1}. Moreover, we can use the recursions \rf{FrbtRecursion-w-2}-\rf{FrbtRecursion-w-3} to derive:%
\begin{equation}
\frac{\SY_{\lambda }^{(n)}(z_{n}q|z_{n}^{\prime },z_{n+1}^{\prime })}{%
\SY_{\lambda /q}^{n}(z_{n}|z_{n}^{\prime },z_{n+1}^{\prime })}=\left( \frac{%
z_{n+1}^{\prime }}{z_{n}^{\prime }}\right) ^{1/2}f_{\text{$\overline{\text{p}%
}$}_{n}\text{$\overline{\text{p}}$}_{n+1}\text{q}_{n}\text{r}_{n}}\frac{1-%
\frac{x_{\text{q}_{n}}}{y_{\text{$\overline{\text{p}}$}_{n}}}q^{-1}}{1-\frac{%
x_{\text{q}_{n}}}{y_{\text{$\overline{\text{p}}$}_{n}}}q^{-1}\frac{z_{n}}{%
z_{n}^{\prime }}}\frac{1-\frac{y_{\text{r}_{n}}}{y_{\text{$\overline{\text{p}%
}$}_{n+1}}}q^{-1}\frac{z_{n}}{z_{n+1}^{\prime }}}{1-\frac{y_{\text{r}_{n}}}{%
y_{\text{$\overline{\text{p}}$}_{n+1}}}q^{-1}},  \label{FrbtRecursion0Y-2}
\end{equation}%
and%
\begin{equation}
\frac{\SY_{\lambda }^{(n)}(z_{n}q|z_{n}^{\prime },z_{n+1}^{\prime })}{%
\SY_{q\lambda }^{n}(z_{n}|z_{n}^{\prime },z_{n+1}^{\prime })}=\left( \frac{%
z_{n+1}^{\prime }}{z_{n}^{\prime }}\right) ^{1/2}f_{\text{$\overline{\text{p}%
}$}_{n}\text{$\overline{\text{p}}$}_{n+1}\text{q}_{n}\text{r}_{n}}\frac{1-%
\frac{y_{\text{q}_{n}}}{x_{\text{$\overline{\text{p}}$}_{n}}}q\frac{%
z_{n}^{\prime }}{z_{n}}}{1-\frac{y_{\text{q}_{n}}}{x_{\text{$\overline{\text{%
p}}$}_{n}}}q}\frac{1-\frac{x_{\text{r}_{n}}}{x_{\text{$\overline{\text{p}}$}%
_{n+1}}}q}{1-\frac{x_{\text{r}_{n}}}{x_{\text{$\overline{\text{p}}$}_{n+1}}}q%
\frac{z_{n+1}^{\prime }}{z_{n}}},  \label{FrbtRecursion0Y-3}
\end{equation}%
where:%
\begin{equation}
f_{\text{$\overline{\text{p}}$}_{n}\text{$\overline{\text{p}}$}_{n+1}\text{q}%
_{n}\text{r}_{n}}=\frac{W_{\text{q}_{n}\text{$\overline{\text{p}}$}_{n}}(z(l))}{W_{\text{q}_{n}\text{$\overline{\text{p}}$}_{n}}(z(0))}\frac{\bar{W}_{\text{r}_{n}\text{$\overline{\text{p}}$}_{n+1}}(z(l))}{\bar{W}_{\text{r}_{n}\text{$\overline{\text{p}}$}_{n+1}}(z(0))}.
\label{FrbtDefinition-f}
\end{equation}

\subsection{\label{FrbtBaxter-Q}Baxter operator construction by gauge transformation}
Here, we present explicitly the construction of the Baxter $\SQ$-operator which can be seen as a
"generalized chiral Potts" transfer matrix and the computation of the
coefficients of the corresponding Baxter equation.
\subsubsection{$\SQ$-operator}
Let us define the kernel of the $\SQ$-operator by the ansatz:%
\begin{equation}
\SQ_{\lambda }(\text{z},\text{z}^{\prime })\equiv \langle \text{z}|\SQ_{\lambda
}|\text{z}^{\prime }\rangle =\prod_{n=1}^{\mathsf{N}}\SY_{\lambda
}^{(n)}(z_{n}|z_{n}^{\prime },z_{n+1}^{\prime }),  \label{Frbtkernel}
\end{equation}%
where $\langle $z$|$ and $|$z$\rangle$ are the generic elements of the left
and right $\su_{n}$-eigenbasis and the $\SY_{\lambda }^{(n)}(z_{n}|z_{n}^{\prime },z_{n+1}^{\prime })$ are
defined in \rf{FrbtRecursion0Y-1}. So the recursion \rf{FrbtRecursion0Y-1}, by using the parameter identifications \rf{Frbteq1}-\rf{Frbteq2}, reads:%
\begin{equation}
\frac{\SY_{\lambda }^{(n)}(z_{n}q^{-1}|z_{n}^{\prime },z_{n+1}^{\prime })}{%
\SY_{\lambda }^{(n)}(z_{n}q|z_{n}^{\prime },z_{n+1}^{\prime })}=-\frac{%
r_{n}z_{n}^{\prime }}{r_{n+1}z_{n+1}^{\prime }}\frac{\alpha
_{n}\beta _{n}}{\mathbbm{c}_{n}\mathbbm{a}_{n}}\frac{\left( 1+\frac{q^{1/2}\lambda \mathbbm{b}_{n}}{\beta
_{n}r_{n}}\frac{z_{n}}{z_{n}^{\prime }}\right) \left( 1-\frac{\lambda
\mathbbm{d}_{n}r_{n+1}}{q^{1/2}\beta _{n}}\frac{z_{n+1}^{\prime }}{z_{n}}\right) }{%
\left( 1-\frac{q^{1/2}r_{n}\lambda \alpha _{n}}{\mathbbm{a}_{n}}\frac{z_{n}^{\prime }}{%
z_{n}}\right) \left( 1+\frac{\lambda \alpha _{n}}{q^{1/2}r_{n+1}\mathbbm{c}_{n}}\frac{%
z_{n}}{z_{n+1}^{\prime }}\right) },  \label{FrbtRecursionY-1}
\end{equation}%
and then:%
\begin{equation}
\langle \text{z}|\tilde{L}_{n}(\lambda )_{21}\SQ_{\lambda }|\text{z%
}^{\prime }\rangle =\prod_{h\neq n,h=1}^{\mathsf{N}}\SY_{\lambda
}^{(h)}(z_{h}|z_{h}^{\prime },z_{h+1}^{\prime })\text{\~{L}}_{n}(\lambda )_{21}\SY_{\lambda }^{(n)}(z_{n}|z_{n}^{\prime },z_{n+1}^{\prime
})\overset{\rf{FrbtRecursionY-1}}{=}0,  \label{FrbtTriangularity}
\end{equation}%
i.e. the defining condition of the $\SQ$-operator.

\subsubsection{Baxter equation}

Let us derive the Baxter equation, from the condition (\ref{FrbtTriangularity}),
we have that:%
\begin{equation}
\langle \text{z}|\tau _{2}(\lambda )\SQ_{\lambda }|\text{z}^{\prime }\rangle \equiv
\prod_{n=1}^{\mathsf{N}}\text{\~{L}}_{n}(\lambda )_{11}\SY_{\lambda
}^{(n)}(z_{n}|z_{n}^{\prime },z_{n+1}^{\prime })+\prod_{n=1}^{\mathsf{N}}\text{\~{L}}%
_{n}(\lambda )_{22}\SY_{\lambda }^{(n)}(z_{n}|z_{n}^{\prime
},z_{n+1}^{\prime }),
\end{equation}%
where:%
\begin{align}
& \text{\~{L}}_{n}(\lambda )_{11}\SY_{\lambda
}^{(n)}(z_{n}|z_{n}^{\prime },z_{n+1}^{\prime })\left. =\right. [(\lambda
\alpha _{n}-\frac{\mathbbm{a}_{n}z_{n}}{q^{1/2}r_{n}z_{n}^{\prime }})\text{v}%
_{n}-(\beta _{n}/\lambda +\frac{q^{1/2}\mathbbm{b}_{n}z_{n}}{r_{n}z_{n}^{\prime }})%
\text{v}_{n}^{-1}]\SY_{\lambda }^{(n)}(z_{n}|z_{n}^{\prime },z_{n+1}^{\prime })
\notag \\
& \text{ \ \ \ \ \ }\overset{\text{\rf{FrbtRecursionY-1}}}{=}-\beta _{n}(q^{-1}\frac{\alpha _{n}\mathbbm{d}_{n}}{\beta _{n}\mathbbm{c}_{n}}\lambda +\frac{1}{%
\lambda })\frac{\left( 1+\frac{q^{1/2}\lambda \mathbbm{b}_{n}}{\beta _{n}r_{n}}\frac{%
z_{n}}{z_{n}^{\prime }}\right) }{\left( 1+\frac{\lambda \alpha _{n}}{%
q^{1/2}r_{n}\mathbbm{c}_{n}}\frac{z_{n}}{z_{n+1}^{\prime }}\right) }\text{v}%
_{n}^{-1}\SY_{\lambda }^{(n)}(z_{n}|z_{n}^{\prime },z_{n+1}^{\prime }),
\end{align}%
and: 
\begin{align}
& \text{\~{L}}_{n}(\lambda )_{22}\SY_{\lambda
}^{(n)}(z_{n}|z_{n}^{\prime },z_{n+1}^{\prime })\left. =\right. [(\gamma
_{n}/\lambda +\frac{\mathbbm{a}_{n}z_{n}}{q^{1/2}r_{n+1}z_{n+1}^{\prime }})\text{v}%
_{n}-(\delta _{n}\lambda -\frac{q^{1/2}\mathbbm{b}_{n}z_{n}}{r_{n+1}z_{n+1}^{\prime }})%
\text{v}_{n}^{-1}]\SY_{\lambda }^{(n)}(z_{n}|z_{n}^{\prime },z_{n+1}^{\prime })
\notag \\
& \text{ \ \ \ \ \ }\overset{\text{\rf{FrbtRecursionY-1}}}{=}-\frac{%
r_{n}z_{n}^{\prime }}{r_{n+1}z_{n+1}^{\prime }}\beta _{n}(q\frac{\alpha
_{n}\mathbbm{b}_{n}}{\beta _{n}\mathbbm{a}_{n}}\lambda +\frac{1}{\lambda })\frac{(1-\frac{%
\lambda \mathbbm{d}_{n}r}{q^{1/2}\beta _{n}}\frac{z_{n+1}^{\prime }}{z_{n}})}{\left( 1-%
\frac{q^{1/2}r\lambda \alpha _{n}}{\mathbbm{a}_{n}}\frac{z_{n}^{\prime }}{z_{n}}%
\right) }\text{v}_{n}^{-1}\SY_{\lambda }^{(n)}(z_{n}|z_{n}^{\prime
},z_{n+1}^{\prime }).
\end{align}%
Now the recursions \rf{FrbtRecursion0Y-2}-\rf{FrbtRecursion0Y-3}, by \rf{Frbteq1}-\rf{Frbteq2}, read:%
\begin{equation}
\frac{\SY_{\lambda }^{(n)}(z_{n}q|z_{n}^{\prime },z_{n+1}^{\prime })}{%
\SY_{\lambda /q}^{n}(z_{n}|z_{n}^{\prime },z_{n+1}^{\prime })}=\left( \frac{%
z_{n+1}^{\prime }}{z_{n}^{\prime }}\right) ^{1/2}f_{\text{$\overline{\text{p}%
}$}_{n}\text{$\overline{\text{p}}$}_{n+1}\text{q}_{n}\text{r}_{n}}\frac{1+%
\frac{q^{1/2}\lambda \mathbbm{b}_{n}}{\beta _{n}r_{n}}}{1+\frac{q^{1/2}\lambda \mathbbm{b}_{n}}{%
\beta _{n}r_{n}}\frac{z_{n}}{z_{n}^{\prime }}}\frac{1+\frac{\lambda \alpha
_{n}}{q^{1/2}r_{n+1}\mathbbm{c}_{n}}\frac{z_{n}}{z_{n+1}^{\prime }}}{1+\frac{\lambda
\alpha _{n}}{q^{1/2}r_{n+1}\mathbbm{c}_{n}}},  \label{FrbtRecursionY-2}
\end{equation}%
and%
\begin{equation}
\frac{\SY_{\lambda }^{(n)}(z_{n}q|z_{n}^{\prime },z_{n+1}^{\prime })}{%
\SY_{q\lambda }^{n}(z_{n}|z_{n}^{\prime },z_{n+1}^{\prime })}=\left( \frac{%
z_{n+1}^{\prime }}{z_{n}^{\prime }}\right) ^{1/2}f_{\text{$\overline{\text{p}%
}$}_{n}\text{$\overline{\text{p}}$}_{n+1}\text{q}_{n}\text{r}_{n}}\frac{1-%
\frac{q^{1/2}r_{n}\lambda \alpha _{n}}{\mathbbm{a}_{n}}\frac{z_{n}^{\prime }}{z_{n}}}{%
1-\frac{q^{1/2}r_{n}\lambda \alpha _{n}}{\mathbbm{a}_{n}}}\frac{1-\frac{\lambda
\mathbbm{d}_{n}r_{n+1}}{q^{1/2}\beta _{n}}}{1-\frac{\lambda \mathbbm{d}_{n}r_{n+1}}{q^{1/2}\beta
_{n}}\frac{z_{n+1}^{\prime }}{z_{n}}}. \label{FrbtRecursionY-3}
\end{equation}%
So using (\ref{FrbtRecursionY-2}) we get:%
\begin{eqnarray}
\text{\~{L}}_{n}(\lambda )_{11}\SY_{\lambda }^{(n)}(z_{n}|z_{n}^{\prime
},z_{n+1}^{\prime }) &=&-\left( \frac{z_{n+1}^{\prime }%
}{z_{n}^{\prime }}\right) ^{1/2}\beta _{n}(q^{-1}\frac{\alpha _{n}\mathbbm{d}_{n}}{%
\beta _{n}\mathbbm{c}_{n}}\lambda +\frac{1}{\lambda })\frac{1+\frac{q^{1/2}\lambda
\mathbbm{b}_{n}}{\beta _{n}r_{n}}}{1+\frac{\lambda \alpha _{n}}{q^{1/2}r_{n+1}\mathbbm{c}_{n}}} 
\notag \\
&&\times f_{\text{$\overline{\text{p}}$}_{n}\text{$\overline{\text{p}}$}%
_{n+1}\text{q}_{n}\text{r}_{n}}\SY_{\lambda /q}^{n}(z_{n}|z_{n}^{\prime
},z_{n+1}^{\prime }),
\end{eqnarray}%
and analogously using (\ref{FrbtRecursionY-3}) we get:%
\begin{eqnarray}
\text{\~{L}}_{n}(\lambda )_{22}\SY_{\lambda }^{(n)}(z_{n}|z_{n}^{\prime
},z_{n+1}^{\prime }) &=&-\frac{r_{n}}{r_{n+1}}\left( \frac{z_{n}^{\prime }}{%
z_{n+1}^{\prime }}\right) ^{1/2}\beta _{n}(q\frac{\alpha _{n}\mathbbm{b}_{n}}{\beta
_{n}\mathbbm{a}_{n}}\lambda +\frac{1}{\lambda })\frac{1-\frac{\lambda \mathbbm{d}_{n}r_{n+1}}{%
q^{1/2}\beta _{n}}}{1-\frac{q^{1/2}r_{n}\lambda \alpha _{n}}{\mathbbm{a}_{n}}}  \notag
\\
&&\times f_{\text{$\overline{\text{p}}$}_{n}\text{$\overline{\text{p}}$}%
_{n+1}\text{q}_{n}\text{r}_{n}}\SY_{q\lambda }^{n}(z_{n}|z_{n}^{\prime
},z_{n+1}^{\prime }).
\end{eqnarray}%
Finally, we have the Baxter equation:
\begin{equation}\label{FrbtGeneralized-Q-Bax-eq}
\tau _{2}(\lambda )\SQ_{\lambda }=a_{B}(\lambda )\SQ_{\lambda /q}+d_{B}(\lambda
)\SQ_{q\lambda },
\end{equation}
with coefficients which read:%
\begin{eqnarray}
a_{B}(\lambda ) &=&(-1)^{\mathsf{N}}\prod_{n=1}^{\mathsf{N}}\beta _{n}f_{\text{$\overline{%
\text{p}}$}_{n}\text{$\overline{\text{p}}$}_{n+1}\text{q}_{n}\text{r}_{n}}(%
\frac{1}{\lambda }+q^{-1}\frac{\alpha _{n}\mathbbm{d}_{n}}{\beta _{n}\mathbbm{c}_{n}}\lambda )%
\frac{1+\frac{q^{1/2}\lambda \mathbbm{b}_{n}}{\beta _{n}r_{n}}}{1+\frac{\lambda \alpha
_{n}}{q^{1/2}r_{n+1}\mathbbm{c}_{n}}},  \label{FrbtQ-coeff-a} \\
d_{B}(\lambda ) &=&(-1)^{\mathsf{N}}\prod_{n=1}^{\mathsf{N}}\beta _{n}f_{\text{$\overline{%
\text{p}}$}_{n}\text{$\overline{\text{p}}$}_{n+1}\text{q}_{n}\text{r}_{n}}(%
\frac{1}{\lambda }+q\frac{\alpha _{n}\mathbbm{b}_{n}}{\beta _{n}\mathbbm{a}_{n}}\lambda )\frac{1-%
\frac{\lambda \mathbbm{d}_{n}r_{n+1}}{q^{1/2}\beta _{n}}}{1-\frac{q^{1/2}r_{n}\lambda
\alpha _{n}}{\mathbbm{a}_{n}}}.  \label{FrbtQ-coeff-d}
\end{eqnarray}
{\bf Remark 7.}\  It is worth pointing out that these proofs hold also for the case of the $\tau_2$-model on the chP curves just imposing that the $r_n$ in \rf{Frbtgauge-r_n} are all equal to the $r$ in \rf{Frbtgauge-r} plus the requirement that q$_{n}$ and r$_{n}$ are on the curve.
\subsection{\label{FrbtQ-Connection-SOV}Connection between generalized Baxter $\SQ$-operator and SOV construction}
 In the next two subsections we point out the connection among the averages of the coefficients of
the Baxter equation \rf{FrbtGeneralized-Q-Bax-eq} and the averages values of the Yang-Baxter operators in
this way establishing the connection with the coefficients of the
SOV-representations.
\subsubsection{Averages of Baxter equation coefficients as eigenvalues of$\ 
\mathcal{M}(\Lambda )$}\label{Frbtrecursion-Y}
In this subsection we show that the averages of the coefficients of the Baxter
equation and the cyclicity parameters $\{\sigma _{1}^{p},...,\sigma _{%
\mathsf{N}}^{p}\}$ are completely characterized in terms of the matrix $%
\mathcal{M}(\Lambda )$ composed by the averages of the Yang-Baxter
generators. In particular, the averages of the Baxter equation coefficients (\ref{FrbtQ-coeff-a}%
) and (\ref{FrbtQ-coeff-d}) coincide with the
eigenvalues of $\mathcal{M}(\Lambda )$\ while the eigenstates of $\mathcal{M}%
(\Lambda )$ fix the $\{\sigma _{1}^{p},...,\sigma _{\mathsf{N}}^{p}\}$.

\begin{itemize}
\item[]\hspace{-0.8cm}{\it\ref{Frbtrecursion-Y}.a} \  \ {\it Matrix characterization of cyclicity conditions}
\end{itemize}  
Let us recall that in appendix \ref{FrbtGeneralized dilogarithm functions}, we have introduced the cyclicity
conditions \rf{Frbtcyclicity} to assure that the generalized $\SY$-functions are well
defined. These cyclicity conditions plus the closure condition $\sigma _{%
\mathsf{N}+1}^{p}=\sigma _{1}^{p}$ fix the parameters $\{\sigma
_{1}^{p},...,\sigma _{\mathsf{N}}^{p}\}$. In particular, the following
matrix characterization holds:

\begin{lem}
Let us define the $2\times 2$ complex matrices:%
\begin{equation}
\mathbb{A}_{n}=A_{n}\cdots A_{2}A_{1},
\end{equation}%
with%
\begin{equation}
A_{n}=%
\begin{pmatrix}
s_{q_{n}}^{p}s_{r_{n}}^{p}x_{p}^{p}y_{p}^{p}-y_{q_{n}}^{p}y_{r_{n}}^{p} & 
x_{p}^{p}y_{r_{n}}^{p}-s_{q_{n}}^{p}s_{r_{n}}^{p}x_{p}^{p}x_{q_{n}}^{p} \\ 
s_{q_{n}}^{p}s_{r_{n}}^{p}y_{p}^{p}x_{r_{n}}^{p}-y_{p}^{p}y_{q_{n}}^{p} & 
x_{p}^{p}y_{p}^{p}-s_{q_{n}}^{p}s_{r_{n}}^{p}x_{q_{n}}^{p}x_{r_{n}}^{p}%
\end{pmatrix}%
,
\end{equation}%
then%
\begin{equation}
\sigma _{n+1}^{p}=f_{\mathbb{A}_{n}}(\sigma _{1}^{p})  \label{Frbtrecursion-n}
\end{equation}%
and $\sigma _{1}^{p}$ is a solution of the quadratic fix-point equation:%
\begin{equation}
\sigma _{1}^{p}=f_{\mathbb{A}_{\mathsf{N}}}(\sigma _{1}^{p})\text{.}
\label{Frbtclosure-recarcion}
\end{equation}%
Here, we have defined:%
\begin{equation}
f_{A}(x)\equiv \frac{ax+b}{cx+d}\qquad\text{ \ for any }2\times 2\text{ complex
matrix }\qquad A=%
\begin{pmatrix}
a & b \\ 
c & d%
\end{pmatrix}%
\text{.}  \label{FrbtMobius-g}
\end{equation}
\end{lem}
\begin{proof}
Let us point out that the cyclicity conditions \rf{Frbtcyclicity} can be seen as recursion
relations on the $\sigma _{n}^{p}$ parameters. In particular, the cyclicity
conditions for the couple ($\sigma _{n}^{p},\sigma _{n+1}^{p}$) is
clearly equivalent to:%
\begin{equation}
\sigma _{n+1}^{p}=f_{A_{n}}(\sigma _{n}^{p}).
\end{equation}%
It is worth to recall now that (\ref{FrbtMobius-g}) defines a group morphism
between $GL(2,\mathbb{C})$ and the group of M\"{o}bius
transformation: 
\begin{equation}
f_{AB}=f_{A}\circ f_{B}\text{ \ \ }\forall A,B\in GL(2,\mathbb{C}).
\end{equation}%
Then, from this property we derive (\ref{Frbtrecursion-n}), which for $n=$\textsf{N} and by the closure condition $\sigma _{\mathsf{N}+1}^{p}=\sigma
_{1}^{p}$ gives (\ref{Frbtclosure-recarcion}).
\end{proof}
\begin{itemize}
\item[]\hspace{-0.8cm}{\it\ref{Frbtrecursion-Y}.b} \  \ {\it Solution of cyclicity conditions and averages of Baxter equation
coefficients}
\end{itemize}
\begin{propn}\label{FrbtAB_B-vs-Average}
Let us denote:%
\begin{equation}
\mathcal{M}(\Lambda )%
\begin{pmatrix}
e_{+}^{+} \\ 
e_{+}^{-}%
\end{pmatrix}%
=\Omega _{+}\left( \Lambda \right) 
\begin{pmatrix}
e_{+}^{+} \\ 
e_{+}^{-}%
\end{pmatrix}%
,\text{ \ \ }\mathcal{M}(\Lambda )%
\begin{pmatrix}
e_{-}^{+} \\ 
e_{-}^{-}%
\end{pmatrix}%
=\Omega _{-}\left( \Lambda \right) 
\begin{pmatrix}
e_{-}^{+} \\ 
e_{-}^{-}%
\end{pmatrix}%
,
\end{equation}%
then we have one of the following cases:%
\begin{eqnarray}
\sigma _{1}^{p} &=&-(x_{p}/y_{p})^{p/2}e_{+}^{+}/e_{+}^{-},\text{ \ \ }%
\prod_{n=1}^{p}a_{B}(\lambda q^{n})=\Omega _{+}\left( \Lambda \right) ,\text{
\ \ }\prod_{n=1}^{p}d_{B}(\lambda q^{n})=\Omega _{-}\left( \Lambda \right) ,
\\
\sigma _{1}^{p} &=&-(x_{p}/y_{p})^{p/2}e_{-}^{+}/e_{-}^{-},\text{ \ \ }%
\prod_{n=1}^{p}a_{B}(\lambda q^{n})=\Omega _{-}\left( \Lambda \right) ,\text{
\ \ }\prod_{n=1}^{p}d_{B}(\lambda q^{n})=\Omega _{+}\left( \Lambda \right) .
\end{eqnarray}
\end{propn}
\begin{proof}
Let $x$ and $y$ be two complex numbers and define:%
\begin{equation}
\begin{pmatrix}\label{FrbtMatrix-f}
x^{\prime } \\ 
y^{\prime }%
\end{pmatrix}%
=A\cdot 
\begin{pmatrix}
x \\ 
y%
\end{pmatrix}%
\end{equation}%
it holds: 
\begin{equation}
f_{A}\left( \frac{x}{y}\right) =\frac{x^{\prime }}{y^{\prime }}.
\end{equation}%
Let $(h_{1},k_{1})$ a couple of complex numbers such that $\sigma
_{1}^{p}=h_{1}/k_{1}$, then we can define the following sequence: 
\begin{equation}\label{FrbtF-Matrix-f}
\begin{pmatrix}
h_{n} \\ 
k_{n}%
\end{pmatrix}%
=\mathbb{A}_{n-1}%
\begin{pmatrix}
h_{1} \\ 
k_{1}%
\end{pmatrix}%
\end{equation}%
which satisfies the property $\sigma _{n}^{p}=h_{n}/k_{n}$ thanks to \rf{Frbtrecursion-n}, \rf{FrbtMatrix-f} and \rf{FrbtF-Matrix-f}. So the closure
relation (\ref{Frbtclosure-recarcion}) is equivalent to the solution of the
spectral problem for the $2\times 2$ complex matrix $\mathbb{A}_{\mathsf{N}}$%
:%
\begin{equation}
\Lambda _{\mathbb{A}_{\mathsf{N}}}%
\begin{pmatrix}
h_{1} \\ 
k_{1}%
\end{pmatrix}%
=\mathbb{A}_{\mathsf{N}}%
\begin{pmatrix}
h_{1} \\ 
k_{1}%
\end{pmatrix}%
=%
\begin{pmatrix}
h_{\mathsf{N}+1} \\ 
k_{\mathsf{N}+1}%
\end{pmatrix}%
.  \label{FrbtRecursion-Sp-A}
\end{equation}%
From the definition \rf{FrbtF-Matrix-f}, it is simple to show that the sequence $(h_{n},k_{n})$
enjoys the interesting property: 
\begin{equation}
\frac{k_{n+1}}{k_{n}}=\left(
x_{p}^{p}y_{p}^{p}-x_{r_{n}}^{p}y_{r_{n}}^{p}\right) \frac{\sigma
_{n}^{p}y_{q_{n}}^{p}-x_{p}^{p}}{\sigma _{n+1}^{p}x_{r_{n}}^{p}-x_{p}^{p}},%
\text{ \ }\frac{h_{n}}{h_{n+1}}=\frac{1}{\det A_{n}}\left(
x_{p}^{p}y_{p}^{p}-x_{q_{n}}^{p}y_{q_{n}}^{p}\right) \frac{\sigma
_{n+1}^{p}y_{p}^{p}-x_{r_{n}}^{p}}{\sigma _{n}^{p}y_{p}^{p}-x_{q_{n}}^{p}},
\end{equation}%
from which we derive:%
\begin{equation}
\prod_{h=1}^{p}a_{B}(\lambda q^{h})=\frac{k_{\mathsf{N}+1}}{k_{1}}%
\prod_{n=1}^{N}\left( \frac{t_{p}\lambda \alpha _{n}}{\mathbbm{a}_{n}\mathbbm{c}_{n}}\right)
^{-p},\text{ \ \ \ }\prod_{h=1}^{p}d_{B}(\lambda q^{h})=\det \mathbb{A}_{%
\mathsf{N}}\frac{h_{1}}{h_{\mathsf{N}+1}}\prod_{n=1}^{N}\left( \frac{%
t_{p}\lambda \alpha _{n}}{\mathbbm{a}_{n}\mathbbm{c}_{n}}\right) ^{-p},
\end{equation}%
being: 
\begin{eqnarray}
\prod_{h=1}^{p}a_{B}(\lambda q^{h}) &=&\prod_{n=1}^{N}\frac{\left(
x_{p}^{p}y_{p}^{p}-x_{r_{n}}^{p}y_{r_{n}}^{p}\right) }{\left(
\frac{ t_{p}\lambda \alpha _{n}}{\mathbbm{a}_{n}\mathbbm{c}_{n}}\right) ^{p}}\frac{\sigma
_{n}^{p}y_{q_{n}}^{p}-x_{p}^{p}}{\sigma _{n+1}^{p}x_{r_{n}}^{p}-x_{p}^{p}},
\\
\prod_{h=1}^{p}d_{B}(\lambda q^{h}) &=&\prod_{n=1}^{N}\frac{\left(
x_{p}^{p}y_{p}^{p}-x_{q_{n}}^{p}y_{q_{n}}^{p}\right) }{\left(
\frac{ t_{p}\lambda \alpha _{n}}{\mathbbm{a}_{n}\mathbbm{c}_{n}}\right) ^{p}}\frac{\sigma
_{n+1}^{p}y_{p}^{p}-x_{r_{n}}^{p}}{\sigma _{n}^{p}y_{p}^{p}-x_{q_{n}}^{p}}.
\end{eqnarray}%
To write the above averages formulae for the coefficients $a_{B}(\la)$ and $d_{B}(\la)$, we have used the
formulae (\ref{FrbtQ-coeff-a})
and (\ref{FrbtQ-coeff-d}), the correspondence \rf{FrbtPoints on the curve} and the fact that from the
definition (\ref{FrbtDefinition-f}) of $f_{\text{$\overline{\text{p}}$}_{n}\text{$\overline{\text{p}}$%
}_{n+1}\text{q}_{n}\text{r}_{n}}^{2}$, we have:
\begin{equation}
\prod_{a=0}^{p-1}\Xi ^{a}\left( f_{\text{$\overline{\text{p}}$}_{n}\text{$%
\overline{\text{p}}$}_{n+1}\text{q}_{n}\text{r}_{n}}^{2}\right) =\left( s_{%
\text{q}_{n}}^{p}s_{\text{r}_{n}}^{p}\frac{{\sigma }_{n+1}^{p}x_{\text{r}%
_{n}}^{p}-x_{\text{p}}^{p}}{{\sigma }_{n+1}^{p}y_{\text{p}}^{p}-y_{\text{r}%
_{n}}^{p}}\frac{{\sigma }_{n}^{p}y_{\text{p}}^{p}-x_{\text{q}_{n}}^{p}}{{%
\sigma }_{n}^{p}y_{\text{q}_{n}}^{p}-x_{\text{p}}^{p}}\right) ^{2l}\underset{%
\rf{Frbtcyclicity}}{=}1.
\end{equation}
Now remarking that:%
\begin{equation}
\frac{k_{\mathsf{N}+1}}{k_{1}}=\left( \frac{h_{1}}{h_{\mathsf{N}+1}}\right)
^{-1}=\Lambda _{\mathbb{A}_{\mathsf{N}}}\text{ \ \ and\ \ \ }\det \mathbb{A}%
_{\mathsf{N}}=\Lambda _{\mathbb{A}_{\mathsf{N}}}\Lambda _{\mathbb{A}_{%
\mathsf{N}}}^{\prime },
\end{equation}%
with $\Lambda _{\mathbb{A}_{\mathsf{N}}}^{\prime }$ the second eigenvalue of 
$\mathbb{A}_{\mathsf{N}}$, we obtain:%
\begin{equation}
\prod_{h=1}^{p}a_{B}(\lambda q^{h})=\Lambda _{\mathbb{A}_{\mathsf{N}%
}}\prod_{n=1}^{N}\left( \frac{t_{p}\lambda \alpha _{n}}{\mathbbm{a}_{n}\mathbbm{c}_{n}}\right)
^{-p},\text{ \ \ \ \ }\prod_{h=1}^{p}d_{B}(\lambda q^{h})=\Lambda _{\mathbb{A%
}_{\mathsf{N}}}^{\prime }\prod_{n=1}^{N}\left( \frac{t_{p}\lambda \alpha _{n}%
}{\mathbbm{a}_{n}\mathbbm{c}_{n}}\right) ^{-p}.
\end{equation}%
Finally, to derive our results we have just to remark that the following identities
hold: 
\begin{equation}
A_{n}=\left( \frac{t_{p}\lambda \alpha _{n}}{\mathbbm{a}_{n}\mathbbm{c}_{n}}\right) ^{p}%
\begin{pmatrix}
-(x_{p}/y_{p})^{p/2} & 0 \\ 
0 & 1%
\end{pmatrix}%
\mathcal{L}_{n}(\Lambda )%
\begin{pmatrix}
-(x_{p}/y_{p})^{-p/2} & 0 \\ 
0 & 1%
\end{pmatrix}%
,
\end{equation}%
where $\mathcal{L}_{n}(\Lambda )$ is the average matrix \rf{FrbtAverage-L} and so by Proposition \ref{Frbtcentral}:
\begin{equation}
\mathbb{A}_{\mathsf{N}}=\left( \prod_{n=1}^{N}\frac{t_{p}\lambda \alpha _{n}%
}{\mathbbm{a}_{n}\mathbbm{c}_{n}}\right) ^{p}%
\begin{pmatrix}
-(x_{p}/y_{p})^{p/2} & 0 \\ 
0 & 1%
\end{pmatrix}%
\mathcal{M}(\Lambda )%
\begin{pmatrix}
-(x_{p}/y_{p})^{-p/2} & 0 \\ 
0 & 1%
\end{pmatrix}%
.  \label{FrbtRelation-M-average}
\end{equation}
\end{proof}
\subsubsection{Relations to the averages of the SOV coefficients}
The following identities:%
\begin{align}
& \prod_{n=1}^{p}a_{B}(\lambda q^{n})+\prod_{n=1}^{p}d_{B}(\lambda
q^{n})\left. =\right. \mathcal{A}(\Lambda )+\mathcal{D}(\Lambda ),
\label{Frbtcompatibility-detD} \\
& \prod_{n=1}^{p}a_{B}(\lambda q^{n})d_{B}(\lambda q^{n})\left. =\right.
\det \mathcal{M}(\Lambda )\left. =\right. \prod_{i=1}^{p}\det \,\hspace{%
-0.08cm}_{\text{q}}\SM_{\mathsf{N}}(\lambda q^{i}),  \label{FrbtQ-det-ChP-Baxter}
\end{align}%
are simply consequences of Proposition \ref{FrbtAB_B-vs-Average}. They are important as they explicitly imply
that the sum and the product of the averages of the Baxter equation coefficients are
Laurent polynomials in $\Lambda $. Let us define the following set of complex numbers:\begin{equation}
\mathsf{z}_{\mathsf{B}}\equiv \cup _{a=1}^{N-1}\cup _{k_a=0}^{p-1}\{\eta _{a}^{\left(k_a\right)}\}
\end{equation}%
where the $\eta _{a}^{\left( k_a\right) }$ were defined in \rf{FrbtZ_B}, then we can prove:
\begin{lem}\label{Frbt11} It is possible to fix:
\begin{equation}
\left. \sigma _{1_{{}}}^{p}\right\vert _{\lambda \in \mathsf{z}_{\mathsf{B}}}\neq 0,
\end{equation}%
then, under this choice, the following identities holds:%
\begin{equation}
\mathcal{A}(\Lambda )=\prod_{n=1}^{p}a_{B}(\lambda q^{n}),\text{\ \ \ }%
\mathcal{D}(\Lambda )=\prod_{n=1}^{p}d_{B}(\lambda q^{n}),\text{ \ \ }%
\forall \lambda \in \mathsf{z}_{\mathsf{B}}.  \label{FrbtConnection-Q-SOV}
\end{equation}
\end{lem}
\begin{proof}Let us remark that from Proposition \ref{FrbtAB_B-vs-Average} it holds:%
\begin{align}
\sigma _{1,\epsilon }^{p}& =-(x_{p}/y_{p})^{p/2}\frac{\mathcal{A}(\Lambda
)-\mathcal{D}(\Lambda )+\epsilon \sqrt{(\mathcal{A}(\Lambda )-\mathcal{D}%
(\Lambda ))^{2}-4\mathcal{B}(\Lambda )\mathcal{C}(\Lambda )}}{2},
\end{align}%
with $\epsilon =\pm$, then in the zeros of \textsf{B}$(\lambda )$, we
have:%
\begin{equation}
\sigma _{1,+}^{p}=-(x_{p}/y_{p})^{p/2}\left( \mathcal{A}(\Lambda )-\mathcal{D%
}(\Lambda )\right) \text{ \ \ or \ \ }\sigma _{1,-}^{p}=0.
\end{equation}%
Then it is clear that with the choice $\epsilon =+$\ it holds:%
\begin{equation}
\left. \prod_{n=1}^{p}a_{B}(\lambda q^{n})\right\vert _{\epsilon=+}=\mathcal{A}%
(\Lambda ),\text{ \ \ \ \ \ }\forall \lambda \in \mathsf{z}_{\mathsf{B}},
\end{equation}%
and so our statement (\ref{FrbtConnection-Q-SOV}) follows once we use the quantum
determinant relation (\ref{FrbtQ-det-ChP-Baxter}) at the zeros of \textsf{B}$%
(\lambda )$: 
\begin{equation}
\mathcal{A}(\Lambda )\mathcal{D}(\Lambda )=\text{det}_{\text{q}}\mathcal{M}_{%
\mathsf{N}}(\Lambda ),\text{ \ \ \ }\forall \lambda \in \mathsf{z}_{\mathsf{B}}.
\end{equation}
\end{proof}

The Lemma \ref{Frbt11} and the characterization of subsection \ref{FrbtEigenvalues-T} of the
coefficients of the Baxter equation \rf{FrbtBaxter-eq-eigenvalues} implies that we can
always chose:
\begin{equation}
\prod_{n=1}^{p}\bar{\textsc{a}}(\lambda
q^{n})=\prod_{n=1}^{p}a_{B}(\lambda q^{n}),\text{ \ \ \ }\prod_{n=1}^{p}\bar{\textsc{d}}(\lambda q^{n})=\prod_{n=1}^{p}d_{B}(\lambda q^{n})%
\text{.}
\end{equation}%
The above average identities imply that there exist two functions $g\left(
\lambda \right) $ and $f\left( \lambda \right) $ such that:%
\begin{equation}
a_{B}(\lambda )=\frac{g\left( \lambda \right) }{g\left( \lambda /q\right) }%
\bar{\textsc{a}}(\lambda ),\text{ \ \ \ \ \ }d_{B}(\lambda )=\frac{f\left(
\lambda \right) }{f\left( q\lambda \right) }\bar{\textsc{d}}(\lambda ),
\end{equation}%
here we show that in fact these relations are stronger being:%
\begin{equation}
g\left( \lambda \right)\propto f\left( \lambda \right) .
\end{equation}%
The above statement follows by remarking that: 
\begin{equation}
\frac{a_{B}(\lambda )d_{B}(\lambda /q)}{\bar{\textsc{a}}(\lambda )\bar{%
\textsc{d}}(\lambda /q)}=N_{B}(\lambda ),
\end{equation}%
where: 
\begin{equation}
N_{B}(\lambda )=\prod_{n=1}^{\mathsf{N}}\left( f_{%
\bar{\text{p}}_{n}\bar{\text{p}}_{n+1}\text{q}_{n}\text{r}_{n}}f_{\Xi (\bar{%
\text{p}}_{n})\Xi (\bar{\text{p}}_{n+1})\text{q}_{n}\text{r}_{n}}\left(-\frac{\beta
_{n}\alpha _{n}}{\mathbbm{a}_{n}\mathbbm{c}_{n}}\right)\frac{1+\frac{q^{1/2}\lambda \mathbbm{b}_{n}}{\beta
_{n}r_{n}}}{1+\frac{\lambda \alpha _{n}}{q^{1/2}r_{n+1}\mathbbm{c}_{n}}}\frac{1-\frac{%
\lambda \mathbbm{d}_{n}r_{n+1}}{q^{3/2}\beta _{n}}}{1-\frac{r_{n}\lambda \alpha _{n}}{%
q^{1/2}\mathbbm{a}_{n}}}\right) ,
\end{equation}%
and from the following:

\begin{lem} The function N$_{B}(\lambda )$ is the identity for all the
representations of the $\tau _{2}$-model.
\end{lem}
\begin{proof}
Let us rewrite the functions $f_{\bar{\text{p}}_{n}\bar{\text{p}}_{n+1}\text{%
q}_{n}\text{r}_{n}}$\ by using their definition \rf{FrbtDefinition-f} and the parametrization  \rf{Frbteq1}-\rf{Frbteq2}; then they reads: 
\begin{equation}
f_{\bar{p}_{n},\bar{p}_{n+1},q_{n},r_{n}}=\left( -\frac{\sigma _{n+1}}{%
\sigma _{n}}\frac{r_{n}^{2}}{r_{n+1}^{2}}\frac{\alpha _{n}\beta _{n}}{%
\mathbbm{a}_{n}\mathbbm{c}_{n}}\right) ^{l}\prod_{k=1}^{l}\frac{1+q^{-2k}q^{3/2}\frac{\mathbbm{b}_{n}}{%
\beta _{n}}\frac{\lambda }{r_{n}}}{1+q^{-2k}q^{1/2}\frac{\alpha _{n}}{\mathbbm{c}_{n}}%
\frac{\lambda }{r_{n+1}}}\frac{1-q^{2k}q^{-3/2}\frac{\mathbbm{d}_{n}}{\beta _{n}}%
\lambda r_{n+1}}{1-q^{2k}q^{-1/2}\frac{\alpha _{n}}{\mathbbm{a}_{n}}\lambda r_{n}},
\end{equation}%
and we can write: 
\begin{equation}
N_{B}(\lambda )=\prod_{n=1}^{N}\left[ \left( -\frac{\beta _{n}\alpha
_{n}}{\mathbbm{a}_{n}\mathbbm{c}_{n}}\right) ^{p}\frac{1+q^{p/2}\left( \frac{\mathbbm{b}_{n}}{\beta _{n}}%
\frac{\lambda }{r_{n}}\right) ^{p}}{1+q^{p/2}\left( \frac{\alpha _{n}}{\mathbbm{c}_{n}}%
\frac{\lambda }{r_{n+1}}\right) ^{p}}\frac{1-q^{p/2}\left( \frac{\mathbbm{d}_{n}}{%
\beta _{n}}\lambda r_{n+1}\right) ^{p}}{1-q^{p/2}\left( \frac{\alpha _{n}}{%
\mathbbm{a}_{n}}\lambda r_{n}\right) ^{p}}\right] .
\end{equation}%
Thanks to the cyclicity conditions \rf{Frbtcyclicity}
the r. h. s. of the last equation is $1$ as announced.
\end{proof}

\textbf{Remark 8.} Let us point out that for general representations the
generalized Baxter $\SQ$-operator is not proven to define a commuting family
w.r.t. the spectral parameter $\lambda $ as well as it is not proven to commute with $\tau _{2}(\lambda )$. In this general setting the eigenstates of $\tau _{2}(\lambda )$ are not necessarily eigenstates of the generalized
Baxter $\SQ$-operator. Vice versa, the Baxter equation \rf{FrbtGeneralized-Q-Bax-eq} implies that any eigenstate of the generalized Baxter $\SQ$-operator is also eigenstate
of $\tau _{2}(\lambda )$ which so will be still constructed as explained in
Theorem \ref{FrbtC:T-eigenstates}. Finally, let us remark that the previous results on the coefficients implies that the Baxter equations characterizing the generalized Baxter $\SQ$-operator and the Baxter $\SQ$-operator constructed by SOV coincide up to a gauge transformation. 

\section{Properties of the cofactors \textsc{C}$_{i,j}(\protect\lambda )$}\label{FrbtCo-F Properties}

\setcounter{equation}{0}

\begin{lem}\label{FrbtAp1}
Let $t(\la )$ be a real  Laurent polynomial of degree $\mathsf{N}$ in $\la$, then the cofactors of the matrix $D(\la )$ satisfy the following identities: 
\begin{equation}
\text{\textsc{C}}_{h+i,k+i}(\lambda )=\text{\textsc{C}}_{h,k}(\lambda q^{i})%
\text{ \ \ \ }\forall i,h,k\in \{1,...,2l+1\}\text{,}
\label{Frbtcofactors-diagonal}
\end{equation}
\begin{equation}
\text{\textsc{C}}_{1,1}(\lambda )=\text{\textsc{C}}_{1,1}(-\lambda ),\,\,\,\,\,\text{\ \textsc{C}}_{2,1}(\lambda )=q^{\mathsf{N}}\text{\textsc{C}}_{1,2}(-\lambda ).
\label{Frbtcofactors-parity-0}
\end{equation}
and: 
\begin{equation}
\left( \text{\textsc{C}}_{1,1}(\lambda )\right) ^{\ast }\equiv \text{\textsc{%
C}}_{1,1}(\epsilon \lambda ^{\ast }),\text{ \ \ \ }\left( \text{\textsc{C}%
}_{1,2}(\lambda )\right) ^{\ast }\equiv \text{\textsc{C}}_{1,2l+1}(%
\epsilon \lambda ^{\ast }).  \label{Frbtcofactor-cc}
\end{equation}
\end{lem}
\begin{proof}
The proof of properties (\ref{Frbtcofactors-diagonal}) and (\ref{Frbtcofactors-parity-0}) coincides step by step with that given in Lemma 4 of 
\cite{FrbtGN10}. Let us prove the property (\ref{Frbtcofactor-cc}) for the
cofactor \textsc{C}$_{1,1}(\lambda )=\det_{2l}$$D_{1,1}(\lambda )$, where: 
\begin{equation}
D_{1,1}(\lambda )\equiv \left\Vert t(\lambda q^{h})\delta _{h,k}-\mathtt{a}%
(\lambda q^{h})\delta _{h,k+1}-\mathtt{d}(\lambda q^{h})\delta
_{h,k-1}\right\Vert _{1\leq h\leq 2l,1\leq k\leq 2l},  \label{FrbtD11}
\end{equation}%
then by the properties under complex conjugation of $t(\lambda )$ and being $%
\mathtt{a}(\lambda )^{\ast }\equiv \epsilon ^{\mathsf{N}}\mathtt{d}(\epsilon
\lambda ^{\ast })$ it holds: 
\begin{equation}
(D_{1,1}(\lambda ))^{\ast }\equiv \left\Vert \epsilon ^{\mathsf{N}}\left(
t(\epsilon \lambda ^{\ast }q^{h})\delta _{h,k}-\mathtt{d}(\epsilon
\lambda ^{\ast }q^{p-h})\delta _{h,k+1}-\mathtt{a}(\lambda ^{\ast
}q^{p-h})\delta _{h,k-1}\right) \right\Vert _{1\leq h\leq 2l,1\leq k\leq 2l},
\label{FrbtD11*}
\end{equation}%
Let $D_{1,1}^{\CC}$ be the $2l\times 2l$ matrix of columns: 
\begin{equation}
\,C_{a}^{D_{1,1}^{\CC}}\equiv \,C_{p-a}^{D_{1,1}},\,\,\,\,\,\,\,\,\,\forall
a\in \{1,...,2l\},
\end{equation}%
where $C_{a}^{X}$ denotes the column $a$ of the matrix $X,$ and let $D_{1,1}^{\CC,\CR}$ be the $2l\times 2l$ matrix of rows: 
\begin{equation}
R_{a}^{D_{1,1}^{\CC,\CR}}\equiv \,R_{p-a}^{D_{1,1}^{\CC}},\,\,\,\,\,\,\,\,\,%
\forall a\in \{1,...,2l\},
\end{equation}%
where $R_{a}^{X}$ denotes for the row $a$ of the matrix $X$, then the following identity holds: 
\begin{equation}
\det_{2l}D_{1,1}^{\CC,\CR}(\epsilon \lambda ^{\ast })\equiv
(\det_{2l}D_{1,1}(\lambda ))^{\ast }\,\,\,\rightarrow \,\,\,\left( \text{%
\textsc{C}}_{1,1}(\lambda )\right) ^{\ast }\equiv \text{\textsc{C}}%
_{1,1}(\epsilon \lambda ^{\ast }).
\end{equation}%
Analogously, it holds: 
\begin{equation}\label{FrbtD_(1,2)-prop}
(\det_{2l-1}D_{(1,2),(1,2)}(\lambda ))^{\ast }\equiv \epsilon ^{\mathsf{N}%
}\det_{2l-1}D_{(1,2),(1,2)}(\epsilon \lambda ^{\ast }/q),
\end{equation}
and by using the expansions: 
\begin{eqnarray}
\text{\textsc{C}}_{1,2l+1}(\lambda ) &=&\prod_{h=1}^{2l}\mathtt{a}(\lambda
q^{h})+\mathtt{d}(\lambda /q)\det_{2l-1}D_{(1,2),(1,2)}(\lambda /q),
\label{FrbtC_2,1-expan} \\
\text{\textsc{C}}_{1,2}(\lambda ) &=&\prod_{h=1}^{2l}\mathtt{d}(\lambda
q^{h})+\mathtt{a}(\lambda q)\det_{2l-1}D_{(1,2),(1,2)}(\lambda )\text{.}
\label{FrbtC_1,2-expan}
\end{eqnarray}%
the second identity in (\ref{Frbtcofactor-cc}) is just a consequence of (\ref{FrbtD_(1,2)-prop}).
\end{proof}
\begin{lem}\label{Frbtcofactors-zeros}
For any $t(\lambda )\in \Sigma _{\tau _{2}}$ the following identities are verified: 
\begin{equation}
\text{{\large {Z}}}_{\text{\textsc{C}}_{1,1}(\lambda )}\cap \text{{\large {Z}%
}}_{\text{\textsc{C}}_{1,2}(\lambda )}\equiv \text{{\large {Z}}}_{\text{%
\textsc{C}}_{1,1}(\lambda )}\cap \text{{\large {Z}}}_{\text{\textsc{C}}%
_{1,2l+1}(\lambda )},\hspace{1.5cm}  \label{FrbtCo-F prop2}
\end{equation}%
and
\begin{equation}
\exists \,\,\text{{\large s}}_{p,\lambda _{0}}\subset \text{{\large {Z}}}_{%
\text{\textsc{C}}_{1,1}(\lambda )}\,\,\,\Rightarrow \,\,\,\text{{\large s}}%
_{p,\lambda _{0}}\cap \text{{\large {Z}}}_{\text{\textsc{C}}_{1,2}(\lambda
)}\neq \emptyset ,\hspace{1.3cm}  \label{FrbtCo-F string}
\end{equation}%
for $\text{{\large s}}_{p,\lambda _{0}}\equiv (\lambda _{0},q\lambda
_{0},...,q^{2l}\lambda _{0})$ any $p$-string.
\end{lem}
\begin{proof}
The proof of identity (\ref{FrbtCo-F prop2}) follows step by step the proof given in Lemma 5 of the paper \cite{FrbtGN10}. Let us assume that (\ref{FrbtCo-F string}) is not satisfied, then (\ref{FrbtInter-step}) implies that $\text{{\large s}}_{p,\lambda_0}\subset\text{%
{\large {Z}}}_{\text{\textsc{C}}_{1,2l+1}(\lambda)}$ from which (\ref{FrbtBax-eq}) holds only if $\text{%
{\large s}}_{p,\lambda_0}\subset\text{{\large {Z}}}_{\mathtt{d}(\lambda)}$,
which is not verified by the definition of $\mathtt{d}(\lambda)$.
\end{proof}

\begin{small} 

\end{small}
\end{document}